\documentclass[11pt]{article}
\usepackage{fullpage}

\usepackage{comment,amsfonts,amsmath,amsthm,graphicx,
            algorithmic,mathtools,tablefootnote,multirow,bbm, caption}

\usepackage[colorlinks,linkcolor=blue,filecolor=blue,citecolor=blue,urlcolor
=blue,pdfstartview=FitH,linktocpage=true]{hyperref}

\newtheorem*{theorem*}{Theorem}
\newtheorem{theorem}{Theorem}
\newtheorem*{definition*}{Definition}
\newtheorem{definition}{Definition}
\newtheorem*{lemma*}{Lemma}
\newtheorem{lemma}{Lemma}
\newtheorem*{claim*}{Claim}
\newtheorem{claim}{Claim}
\newtheorem*{corollary*}{Corollary}
\newtheorem{corollary}{Corollary}
\newtheorem{remark}{Remark}
\newtheorem{fact}[theorem]{Fact}
\newtheorem*{fact*}{Fact}

\newtheorem{observation}{Observation}
\newtheorem{question}{Open Question}

\begin{document}
\title{Spanning Tree Covers for Path-Separable Graphs: Trading Stretch for Size}

\author{Michael Elkin\footnote{Ben-Gurion University of the Negev, Beer-Sheva, Israel. elkinm@bgu.ac.il}, Idan Shabat\footnote{Ben-Gurion University of the Negev, Beer-Sheva, Israel. shabati@post.bgu.ac.il}{\let\thefootnote\relax\footnote{{Supported by Lynn and William Frankel Center for Computer Sciences and ISF grant 3413/25.}}
}
}
\date{}
\maketitle

\pagenumbering{gobble}

\begin{abstract}

Given a graph $G=(V,E)$, a collection $\mathcal{T}$ of spanning trees of $G$ is called a \textit{spanning tree cover} of size $|\mathcal{T}|$ and stretch $\alpha$ for $G$ if for every vertex pair $u,v\in V$ there exists a tree $T_{uv}\in\mathcal{T}$, such that
\[d_{T_{uv}}(u,v)\leq\alpha\cdot d_G(u,v)~.\]

Spanning tree covers were introduced in the pioneering work of Gupta et al. \cite{GKR04}, that showed that $p$-path-separable graphs admit stretch-$3$ spanning tree covers with size $O(p\log n)$. Many subsequent papers focused on a relaxed notion of \textit{non-spanning tree covers}, in which the trees are required to be dominating, but may use edges that do not belong to the graph. In particular, Bartal et al. \cite{BFN22} devised a construction of non-spanning tree covers with stretch $1+\epsilon$ and size $O(p\cdot\frac{\log^2n}{\epsilon^2})$. Recently, Chang et al. \cite{CCLMST23,CCLMST24} devised a non-spanning tree cover for $K_r$-minor-free graphs, with stretch $1+\epsilon$ and size $2^{\frac{1}{\epsilon}r^{O(r)}}$, and an exact (stretch-$1$) spanning tree cover for \textit{unweighted} $K_r$-minor-free graphs with size $r^{O(diam(G))}$. However, the problem of devising spanning tree covers with stretch smaller than $3$ and small size for general $p$-path-separable graphs, or even for general $K_r$-minor-free graphs, remained open.

We show that $p$-path-separable graphs admit \textit{spanning} tree covers with stretch $1+\epsilon$ and size $O(p\cdot\frac{\log^2n}{\epsilon})$, improving the state-of-the-art construction of \cite{BFN22} that provides only \textit{non-spanning} tree covers. Moreover, we demonstrate that one can trade stretch for size, and devise spanning tree covers with stretch $O(k\log\log p)$ and size $O(kp^{\frac{1}{k}}\cdot\log^{2}n)$ for strongly $p$-path-separable graphs. We also provide a tradeoff for weakly path-separable graphs. In the particularly important special case of $K_r$-minor-free graphs, we devise spanning tree covers with stretch $O(k\log\log r)$ and size $O(kr^{2+\frac{1}{k}}\cdot\log^{2}n)$. Currently, the best-known bound on $p$ for $K_r$-minor-free graphs is $p=r^{4602}$ \cite{GSW25,AG06}, and thus, for $r=\Omega(\log n)$ and $k\geq2$, this size bound is much smaller than the bound $O(p\cdot\frac{\log^2n}{\epsilon})=O(r^{4602}\cdot\frac{\log^2n}{\epsilon})$ that we obtain for stretch $1+\epsilon$. In addition, we devise improved constructions of non-spanning tree covers that exhibit stretch-size tradeoffs.

Our new spanning tree covers give rise to significantly improved path-reporting distance oracles (PRDOs) and routing schemes. In particular, we devise the first stretch-$(1+\epsilon)$ PRDO for $K_r$-minor-free graphs. Its size is $r^{O(1)}\cdot\frac{\log^2n}{\epsilon}\cdot n$. We also devise a stretch-$O(k)$ PRDO with size $\tilde{O}(r^{2+\frac{1}{k}}\cdot n)$.

\end{abstract}

\newpage
\tableofcontents
\newpage

\pagenumbering{arabic}
\setcounter{page}{1}

\section{Introduction}

\subsection{Tree Covers with Small Stretch} \label{sec:IntSmallStretchTreeCovers}

Given an undirected weighted $n$-vertex graph $G=(V,E)$, a \textit{spanning tree cover} $\mathcal{T}=\{T_1,T_2,...,T_q\}$ for $G$ with \textit{stretch} $\alpha$ and \textit{size} $q=|\mathcal{T}|$ (for some integer $q\geq1$ and real $\alpha\geq1$) is a collection of spanning trees $T_i$ of $G$ that satisfy the following property: for every vertex pair $u,v\in V$ there exists a tree $T_i$ in $\mathcal{T}$ such that 
\[d_{T_i}(u,v)\leq\alpha\cdot d_G(u,v)~,\]
i.e., the $u$-$v$ distance in the tree $T_i$ is at most $\alpha$ times the $u$-$v$ distance in $G$.

This notion was introduced in the seminal paper by Gupta et al. \cite{GKR04}. They showed that $n$-vertex graphs with balanced recursive (vertex) separators of size $s(n)$ admit exact spanning tree covers (i.e., with stretch $1$) with size $O(\sum_{i=0}^{\log n}s(\frac{n}{2^i}))$. In particular, in planar graphs, this yields exact spanning tree covers of size $O(\sqrt{n})$. They also devised spanning tree covers with stretch $3$ and size $O(p(n)\cdot\log n)$ for graphs with balanced \textit{path} separators of size $p(n)$ (see Definition \ref{def:PathSeparable}).

The motivation of \cite{GKR04} for studying spanning tree covers was MPLS routing. Their work spurred a lot of subsequent research on this fundamental metric structure, and it was thereafter extensively explored, both for its own sake and for its numerous important algorithmic applications. These applications include (in addition to MPLS routing) path-reporting distance oracles, compact routing, distance labeling schemes, graph spanners and network design \cite{GKT03,GHR06,MN06,AGGM06,CGMZ16,ACEFN20,BFN22,CCLMST23,CCLMST24,CCLST25,CTX25}. Much of this work considered a relaxation called \textit{non-spanning tree covers}, in which the trees in $\mathcal{T}$ are not required to be spanning trees of the input graph $G=(V,E)$, but rather can use edges $(u,v)\notin E$, with weight $w(u,v)\geq d_G(u,v)$. One may also allow these trees to use Steiner vertices (that were not present in $G$). If this is the case, one typically requires the trees to be \textit{hierarchically separated} (shortly, HSTs; see Definition \ref{def:HST}), and then the tree cover is called an \textit{HST cover}.

For some applications of tree covers, non-spanning tree covers suffice. In particular, this is the case for non-path-reporting distance oracles (henceforth, DOs), for distance labeling schemes, and for metric spanners. This is, however, not the case for many other important applications of tree covers, such as path-reporting distance oracles (henceforth, PRDOs), compact routing schemes, graph spanners and MPLS routing. The importance of devising spanning tree covers, as opposed to non-spanning ones, was recently articulated by Chang et al. \cite{CCLST25}, who devised spanning tree covers for graphs whose induced metric is doubling. They also elaborate on the inherent difficulty of building spanning tree covers, as opposed to non-spanning ones (see Section 1.1 of \cite{CCLST25}).

A landmark non-spanning tree cover for general graphs was devised by Mendel and Naor \cite{MN06} (see also \cite{BLMN03} for an earlier weaker construction). They showed that for any integer parameter $k=1,2,...$, there exists an HST cover with stretch $O(k)$ and size $O(kn^{\frac{1}{k}})$. The leading constant in the $O$-notation of the stretch was improved to $12e$ by \cite{NT12}.\footnote{\cite{NT12} shows that for any $n$-point metric $X$ and a parameter $k=1,2,...$, there exists a subset $Y\subseteq X$ of size $\Omega(n^{1-\frac{1}{k}})$ that embeds into an ultrametric with distortion at most $2ek$. To obtain an HST cover, one needs to apply Theorem 4.1 of \cite{MN06} (see also Theorem \ref{thm:ExtendingUM}), that extends an HST for $Y$ to provide small stretch for all pairs in $X\times Y$. This extension blows up the stretch by an additional factor of $6$.} The HST cover of \cite{MN06,NT12} is a \textit{Ramsey} tree cover, meaning that for every vertex $v$ there is a tree $T_v$ in the cover that provides stretch $O(k)$ for all vertex pairs $(v,u)$ (that involve $v$). Abraham et al. \cite{ACEFN20} devised a spanning analogue of this construction: they showed that for any $n$-vertex graph $G=(V,E)$ and a parameter $k=1,2,...$, there exists a spanning Ramsey tree cover with stretch $O(k\log\log n)$ and size $O(kn^{\frac{1}{k}})$.

Particularly intensive recent research effort focused on building tree covers for graphs that exclude a fixed minor $K_r$, for an integer parameter $r$ \cite{BFN22,CCLMST23,CCLMST24}. Abraham and Gavoille \cite{AG06} showed that $K_r$-minor-free graphs admit balanced path separators of size
\begin{equation} \label{eq:AGPathSeparator}
    p(r)=O(r^2\cdot RS(r)\cdot(r^2+RS(r)))~,
\end{equation}
where $RS(r)$ is a very quickly growing function from Robertson-Seymour theorem \cite{RS03}. Upper bounds on $RS(r)$ were recently dramatically improved by \cite{KTW20,GSW25}, and currently the best-known estimate is $RS(r)=O(r^{2300})$ \cite{GSW25}, i.e., $p(r)=O(r^{4602})$. We denote these exponents by $c_{GSW}=2300$ and $c_{AG}=4602$. In the opposite direction, Diot and Gavoille \cite{DG10} showed that $p$-path-separable graphs exclude a $K_r$-minor, for some $r=r(p)$. To the best of our understanding, this dependency is also quite high, i.e., $r(p)$ is at least $p^{c_{GSW}}$.

Bartal et al. \cite{BFN22} showed that $p$-path-separable graphs (i.e., graphs that admit a balanced recursive separator formed by $p$ paths) admit non-spanning tree covers with stretch $1+\epsilon$ and size $O(p\cdot\frac{\log^2n}{\epsilon^2})$. They also devised another construction of non-spanning tree covers for $K_r$-minor-free graphs with stretch $O(r^4)$ and size $2^{O(r)}$. These tree covers of \cite{BFN22} are not Ramsey tree covers, and for a good reason -- based on a result by \cite{BLMN03}, Bartal et al. \cite{BFN22} showed that there are series-parallel $n$-vertex graphs for which any stretch-$k$ Ramsey tree cover requires $n^{\Omega(\frac{1}{k\log k})}$ trees. Chang et al. \cite{CCLMST23} devised a construction of exact (stretch-$1$) spanning tree covers for \textit{unweighted} $K_r$-minor-free graphs with size $r^{O(diam(G))}$, where $diam(G)$ stands for the diameter of the input graph $G$. In a follow-up paper \cite{CCLMST24}, the same authors devised a non-spanning tree cover with stretch $1+\epsilon$ and size $2^{\frac{1}{\epsilon}r^{O(r)}}$ for weighted $K_r$-minor-free graphs.

To summarize, currently the only \textit{spanning} tree cover known for weighted $p$-path-separable graphs is the tree cover of \cite{GKR04} (from quarter a century ago!) with stretch $3$ and size $O(p\log n)$. This leads to the following question.

\begin{question} \label{question:SpanTreeCoversForPathSeparable}
    Do $p$-path-separable graphs admit \textbf{spanning} tree covers with stretch smaller than $3$ and $O(p)\cdot\log^{O(1)}n$ trees? In particular, do $K_r$-minor-free graphs admit such spanning tree covers of size $p(r)\cdot\log^{O(1)}n$, for some function $p(\cdot)$?
\end{question}

To the best of our knowledge, currently no spanning tree cover with stretch $1+\epsilon$ and size $\log^{O(1)}n$ is known even for \textit{planar} graphs.

We answer Question \ref{question:SpanTreeCoversForPathSeparable} in the affirmative, and show that $p$-path-separable graphs admit spanning tree covers with stretch $1+\epsilon$, for an arbitrarily small $\epsilon>0$, and $O(p\cdot\frac{\log^2n}{\epsilon})$ trees. Note that this bound even improves (by a factor of $\frac{1}{\epsilon}$) over the bound of \cite{BFN22}; recall that the latter construction provides \textit{non-spanning} tree covers, while our construction provides spanning ones. In particular, our result implies that $K_r$-minor-free graphs admit spanning tree covers with stretch $1+\epsilon$ and size $r^{O(1)}\cdot\frac{\log^2n}{\epsilon}$. For planar graphs, $r=O(1)$, i.e., the size bound is $O(\frac{\log^2n}{\epsilon})$. See Table \ref{table:SmallStretchTreeCovers} for a concise summary of previous and new results.

\begin{table}[htbp]
\begin{tabular}{|l|l|l|l|l|l|}
\hline
\textbf{Graph Family}  & \textbf{Stretch} & \textbf{Number of} & \textbf{Span-} & \textbf{Remarks} & \textbf{Source} \\
 & & \textbf{Trees} & \textbf{-ning?} &  &  \\ \hline
\multirow{3}{*}{\begin{tabular}[c]{@{}l@{}}$p$-path-separable 
\\ graphs\end{tabular}} 
 & $3$ & $O(p\log n)$ & YES &  & \cite{GKR04} \\ \cline{2-6} 
 & $1+\epsilon$ & $O(p\cdot\frac{\log^2n}{\epsilon^2})$ & NO &  & \cite{BFN22} \\ \cline{2-6}
 & $\boldsymbol{1+\epsilon}$ & $\boldsymbol{O(p\cdot\frac{\log^2n}{\epsilon})}$ & \textbf{YES} &  & \textbf{New} \\ \cline{2-6} \hline
\multirow{5}{*}{\begin{tabular}[c]{@{}l@{}}$K_r$-minor-free\\ graphs\end{tabular}}  
 & $1$ & $r^{O(diam(G))}$ & YES & Unweighted & \cite{CCLMST23} \\ 
 &  &  &  & graphs &  \\ \cline{2-6}
 & $1+\epsilon$ & $2^{\frac{1}{\epsilon}r^{O(r)}}$ & NO &  & \cite{CCLMST24} \\ \cline{2-6} 
 & $1+\epsilon$ & $r^{O(1)}\cdot\frac{\log^2n}{\epsilon^2}$ & NO &  & \cite{BFN22} \\ \cline{2-6}
 & $\boldsymbol{1+\epsilon}$ & $\boldsymbol{r^{O(1)}\cdot\frac{\log^2n}{\epsilon}}$ & \textbf{YES} &  & \textbf{New} \\ \cline{2-6} \hline
\end{tabular}
\caption{A concise summary of previous and our results for tree covers for $p$-path-separable and $K_r$-minor-free graphs in the regime of small stretch (at most $3$). See Section \ref{sec:IntLargeStretchTreeCovers} and Table \ref{table:LargeStretchTreeCovers} for results with larger stretch.
}
\label{table:SmallStretchTreeCovers}
\end{table}

\subsection{Tree Covers that Trade Stretch for Size} \label{sec:IntLargeStretchTreeCovers}

In many of the previous studies of tree covers in $K_r$-minor-free graphs, the parameter $r$ is assumed to be constant \cite{BFN22,CCLMST23,CCLMST24}. Similarly, the parameter $p$ for $p$-path-separable graphs is often implicitly assumed to be small. As a result, many of the previous bounds feature quite high dependencies on these parameters. Most notably, this is the case for the non-spanning tree cover for $K_r$-minor-free graphs with stretch $1+\epsilon$ and size $2^{\frac{1}{\epsilon}r^{O(r)}}$ of \cite{CCLMST24}. This bound becomes meaningless already for $r=\Omega(\frac{\log\log n}{\log^{(3)}n})$ or $\epsilon=O(\frac{1}{\log n})$. This is also the case (for $r=\Omega(\log n)$) for the non-spanning tree cover of \cite{BFN22} for the same graph family, with stretch $O(r^4)$ and size $2^{O(r)}$. Note also that the size of all previous tree covers for $p$-path-separable graphs \cite{GKR04,BFN22} grows at least linearly with $p$.

To summarize, for $r=\Omega(\log n)$, currently there are only two meaningful bounds available: the spanning tree cover of \cite{GKR04} with stretch $3$ and size $O(p\log n)=O(r^{c_{AG}}\cdot\log n)$, and the non-spanning tree cover of \cite{BFN22} with stretch $1+\epsilon$ and size $O(p\cdot\frac{\log^2n}{\epsilon^2})=O(r^{c_{AG}}\cdot\frac{\log^2n}{\epsilon^2})$. Our bound from Section \ref{sec:IntSmallStretchTreeCovers} provides a \textit{spanning} tree cover with slightly better size than that of \cite{BFN22}. However, its dependence on $p$ is still $\Omega(p)$, and as a result, its dependence on $r$ suffers from the same very high exponent as in \cite{BFN22}. We note that these large dependencies on $p$ and $r$ show up also in applications of tree covers. In particular, the query times of DOs that employ these covers are also $\Omega(p)$ and $\Omega(r^{c_{AG}})$, and their size is $\Omega(n\cdot p)$ and $\Omega(n\cdot r^{c_{AG}})$, respectively. The following problem arises.

\begin{question} \label{question:LargeStretchSmallSize}
    Can one trade stretch for size and obtain spanning and non-spanning tree covers for $p$-path-separable graphs with \textbf{sublinear} size in $p$?
\end{question}

Our main technical contribution is an affirmative answer to this question. Our results in this context depend on the type of the path separator that the graph family at hand admits. We say that a graph $G$ admits a \textit{strong path separator} of size $p(\cdot)$, if any induced sub-graph $G[U]=(U,E[U])$, $U\subseteq V$, contains $p=p(|U|)$ shortest paths $Q^1,Q^2,...,Q^p$ in $G[U]$, whose removal decomposes the graph  into connected components of size at most $\frac{1}{2}|U|$. If all these paths also have a common endpoint $x$ and are all contained in a shortest paths tree $T_x$ rooted at $x$, then the separator is called a strong \textit{tree-like} (shortly, \textit{STL}) path separator. Such separators were introduced in seminal papers of Thorup \cite{Tho04} and Klein \cite{Kle02}, that showed that planar graphs admit STL path separators with $p=3$. One can also deduce from \cite{GHT84,Dji85,Dji88,AD96} that graphs of genus $g$ admit STL path separators of size $O(g)$.

Abraham and Gavoille \cite{AG06} showed that $K_r$-minor-free graphs admit \textit{weak} tree-like path separators of size $p(n)$, given by (\ref{eq:AGPathSeparator}). For a graph $G$, a union $P_0\cup P_1\cup\cdots\cup P_{\ell-1}$ of $\ell$ sets, where each $P_i$ is a union of at most $p_i$ shortest paths in $H_i=G\setminus\bigcup_{j<i}P_j$, is called a \textit{(weak) $(\ell,\pi)$-path separator}, if $p_i\leq\pi$ for every $0\leq i\leq\ell-1$. The size of the path separator is defined as $p=\sum_{i=0}^{\ell-1}p_i\leq\ell\cdot\pi$. An $(\ell,\pi)$-path separator is called \textit{tree-like}, if for every $0\leq i\leq\ell-1$, the paths in $P_i$ are all contained in a shortest paths tree $T$ in $H_i$, whose root serves as their common endpoint.

Our constructions are parametrized by a parameter $k=1,2,...$. We show that strongly tree-like $p$-path-separable graphs admit spanning (respectively, non-spanning) tree covers with stretch $O(k\log\log p)$ (resp., $O(k)$) and size $O(kp^{\frac{1}{k}}\cdot\log^2n)$. For \textit{$s$-vertex-separable graphs}, i.e., graphs that admit (balanced recursive) \textit{vertex} separators of size $s=s(n)$ (see Definition \ref{def:RecursiveSeparators}), and for graphs with treewidth $s$, this bound improves to $O(k\cdot s^{\frac{1}{k}}\cdot\log n)$. These results provide non-spanning tree covers with a constant stretch (and spanning ones with stretch $O(\log\log p)$) with \textit{sublinear size in $p$}. In particular, they are applicable to graphs with bounded genus $g$, providing stretch-$O(k\log\log g)$ spanning (respectively, stretch-$O(k)$ non-spanning) tree covers with size $O(k\cdot g^{\frac{1}{k}}\cdot\log^2n)$.

More generally, for tree-like $(\ell,\pi)$-path-separable graphs, we provide spanning (respectively, non-spanning) tree covers with stretch $O(k\log\log\pi)$ (resp., $O(k)$) and size $O(k\cdot\ell\cdot\pi^{\frac{1}{k}}\cdot\log^2n)$. In a more general case of not necessarily tree-like separators, we provide a bound of $O(k\cdot\ell\cdot\pi^{\frac{1}{k}}\cdot\log^3n)$ on the size. In the particularly important special case of $K_r$-minor-free graphs, this implies stretch-$O(k\log\log r)$ for spanning tree covers (and stretch-$O(k)$ for non-spanning ones) with size $O(k\cdot r^{2+\frac{1}{k}}\cdot\log^2n)$. A variant of our construction of non-spanning tree covers provides HST covers with the same properties. We summarize these results in Table \ref{table:LargeStretchTreeCovers}.

\begin{table}[htbp]
\begin{tabular}{|l|l|l|l|l|l|}
\hline
\textbf{Graph Family}  & \textbf{Stretch} & \textbf{Number of} & \textbf{Span-} & \textbf{Remarks} & \textbf{Source} \\
 & & \textbf{Trees} & \textbf{-ning?} &  &  \\ \hline
\multirow{3}{*}{\begin{tabular}[c]{@{}l@{}}$s$-vertex-separable\\ graphs \\ 
\end{tabular}} 
 & $1$ & $O(s\log n)$ & YES &  & \cite{GKR04} \\ \cline{2-6} 
 & $\boldsymbol{O(k)}$ & $\boldsymbol{O(k s^{\frac{1}{k}}\cdot\log n)}$ & \textbf{NO} &  & \textbf{New} \\ \cline{2-6} 
 & $\boldsymbol{O(k\log\log s)}$ & $\boldsymbol{O(k s^{\frac{1}{k}}\cdot\log n)}$ & \textbf{YES} &  & \textbf{New} \\ \hline 
\multirow{3}{*}{\begin{tabular}[c]{@{}l@{}}$p$-path-separable
\\ graphs\end{tabular}} 
 & $1+\epsilon$ & $O(p\cdot\frac{\log^2n}{\epsilon^2})$ & NO &  & \cite{BFN22} \\ \cline{2-6}
 & $\boldsymbol{1+\epsilon}$ & $\boldsymbol{O(p\cdot\frac{\log^2n}{\epsilon})}$ & \textbf{YES} &  & \textbf{New} \\ \cline{2-6}
 & $3$ & $O(p\log n)$ & YES &  & \cite{GKR04} \\ \cline{2-6} 
 \hline
 \multirow{2}{*}{\begin{tabular}[c]{@{}l@{}}Tree-like $(\ell,\pi)$-path\\-separable graphs\end{tabular}} 
 & $\boldsymbol{O(k\log\log\pi)}$ & $\boldsymbol{O(k\ell\pi^{\frac{1}{k}}\cdot\log^2n)}$ & \textbf{YES} &  & \textbf{New} \\ \cline{2-6} 
 & $\boldsymbol{O(k)}$ & $\boldsymbol{O(k\ell\pi^{\frac{1}{k}}\cdot\log^2n)}$ & \textbf{NO} & & \textbf{New} \\ 
 \hline
\multirow{5}{*}{\begin{tabular}[c]{@{}l@{}}$K_r$-minor-free\\ graphs\end{tabular}}  
 & $1+\epsilon$ & $2^{\frac{1}{\epsilon}r^{O(r)}}$ & NO &  & \cite{CCLMST24} \\ \cline{2-6} 
 & $1+\epsilon$ & $r^{c_{AG}}\cdot\frac{\log^2n}{\epsilon^2}$ & NO &  & \cite{BFN22} \\ \cline{2-6}
 & $\boldsymbol{1+\epsilon}$ & $\boldsymbol{r^{c_{AG}}\cdot\frac{\log^2n}{\epsilon}}$ & \textbf{YES} &  & \textbf{New} \\ \cline{2-6} 
 & $\boldsymbol{O(k\log\log r)}$ & $\boldsymbol{O(kr^{2+\frac{1}{k}}\cdot\log^2n)}$ & \textbf{YES} &  & \textbf{New} \\ \cline{2-6}
 & $\boldsymbol{O(k)}$ & $\boldsymbol{O(kr^{2+\frac{1}{k}}\cdot\log^2n)}$ & \textbf{NO} &  & \textbf{New} \\ \hline
\multirow{4}{*}{\begin{tabular}[c]{@{}l@{}}Graphs w. bounded \\ genus $g$ \end{tabular}}  
 & $1+\epsilon$ & $O\left(g\cdot\frac{\log^2n}{\epsilon^2}\right)$ & NO & & \cite{BFN22} \\ \cline{2-6}
 & $\boldsymbol{1+\epsilon}$ & $\boldsymbol{O\left(g\cdot\frac{\log^2n}{\epsilon}\right)}$ & \textbf{YES} & & \textbf{New} \\ \cline{2-6}
 & $\boldsymbol{O(k\log\log g)}$ & $\boldsymbol{O(kg^{\frac{1}{k}}\log^2n)}$ & \textbf{YES} &  & \textbf{New} \\ \cline{2-6}
 & $\boldsymbol{O(k)}$ & $\boldsymbol{O(kg^{\frac{1}{k}}\log^2n)}$ & \textbf{NO} &  & \textbf{New} \\\hline
\end{tabular}
\caption{A concise summary of previous and new bounds for spanning and non-spanning tree covers that trade stretch for size, for vertex-separable or path-separable graphs. 
Throughout this paper, $c_{AG}=4602$ is the currently best-known estimate (\cite{GSW25}) on the exponent of $r$ in the path separator by \cite{AG06}, for $K_r$-minor-free graphs.
The new results for tree-like path-separable graphs hold also for (not necessarily tree-like) path-separable graphs, with an additional factor of $\log n$ in the number of trees. 
Results for strongly (tree-like) path-separable graphs are obtained by substituting $\ell=1$.}
\label{table:LargeStretchTreeCovers}
\end{table}

\subsection{Applications} \label{sec:IntApplications}

As was mentioned above, our new constructions of tree covers directly give rise to numerous applications. We next outline the applications to path-reporting distance oracles (henceforth, PRDOs). In Appendix \ref{app:DLandRS} we elaborate on applications to distance labeling and routing schemes.

A \textit{PRDO} for a graph $G=(V,E)$ is a compact data structure, that given a pair of vertices $(u,v)\in V^2$ (called \textit{query}), quickly returns an $\alpha$-approximate shortest $u$-$v$ path $P_{u,v}$ in $G$ (for a parameter $\alpha\geq1$ called the \textit{stretch}). The prime parameters of PRDOs are stretch, size (number of RAM words) and query time\footnote{Typically, the query time is of the form $q(n,\alpha)+O(|P_{u,v}|)$, where $|P_{u,v}|$ stands for the number of edges in the returned path $P_{u,v}$.}.

PRDOs for general graphs were extensively studied \cite{TZ01,WN13,ENWN16,EP15,NS23,ES23,CZ24}. In the regime of constant (or small) $k$, the PRDO of Thorup and Zwick, in conjunction with an improvement by \cite{WN13}, provides a nearly-optimal stretch $2k-1$, size $O(kn^{1+\frac{1}{k}})$ and query time $O(\log k)$. A landmark PRDO for planar graphs with size $\tilde{O}_\epsilon(n)$, stretch $1+\epsilon$ and constant query time was provided by Thorup \cite{Tho04}, and independently by Klein \cite{Kle02}. The result extends directly to strongly tree-like $p$-path-separable graphs, but the query time and the size become $O(p\cdot\frac{\log n}{\epsilon})$ and $O(np\cdot\frac{\log n}{\epsilon})$, respectively. Abraham and Gavoille \cite{AG06} extended the non-path-reporting DO of \cite{Tho04,Kle02} 
to (weakly) $p$-path-separable graphs, yielding a $(1+\epsilon)$-approximate DO with the aforementioned query time and size. However, to the best of our knowledge, no \textit{path-reporting} distance oracle with stretch $1+\epsilon$ for weakly $p$-path-separable graphs, or for $K_r$-minor-free graphs, is currently known.

\begin{question}
    Does there exist a \textbf{path-reporting} distance oracle with stretch smaller than $3$ for \textbf{weakly} $p$-path-separable graphs, or for $K_r$-minor-free graphs?
\end{question}

The spanning tree cover of \cite{GKR04}, that provides stretch $3$ and size $O(p\log n)$ for weakly $p$-path-separable graphs, gives rise directly to a PRDO with the same stretch and \textit{size overhead} (i.e., size divided by $n$), and query time equal to the size overhead. Kawabarayashi et al. \cite{KKS11} devised linear-size non-path-reporting DOs for $K_r$-minor-free graphs with stretch $1+\epsilon$ and query time $O(r^{2c_{AG}}\cdot\frac{\log^2}{\epsilon^2})$.\footnote{The dependency of the query time on $r$ in the result of \cite{KKS11} is unspecified. In the proof of their Theorem 4 they indicate that it is $O(q^2)$, where $q=O(\frac{\log n}{\epsilon})$ is a parameter that depends at least linearly on the size of the path separator of \cite{AG06}, i.e., it is $\approx r^{c_{AG}}$.} Our spanning tree cover for this graph family yields a PRDO with stretch $1+\epsilon$, size overhead $O(p\cdot\frac{\log^2n}{\epsilon})$ and query time $O(p\cdot\frac{\log n}{\epsilon})$. To obtain size overhead and query time which are \textit{sublinear} in $p$, one can use our spanning tree covers with stretch $O(k\log\log p)$. In the case of strongly tree-like $p$-path-separable graphs, they yield a PRDO with size overhead $O(kp^{\frac{1}{k}}\cdot\log^2n)$ and query time $O(kp^{\frac{1}{k}}\cdot\log n)$. If the separator is strong, but not tree-like, a variant of our construction provides bounds of $O(kp^{\frac{1}{k}}\cdot\log^3n)$ and $O(kp^{\frac{1}{k}}\cdot\log^2n)$ for size overhead and query time, respectively.

We also devise a different PRDO, which uses our non-spanning tree covers in conjunction with a path-reporting pairwise $(1+\epsilon)$-preserver by Elkin and Shabat \cite{ES23}. It enables us to improve the stretch to $O(k)$, while retaining similar size overhead and query time. Specifically, the size overhead grows by a factor of $\gamma(k,p,n)=(\frac{k\log n}{\log p})^{O(1)}=\log^{O(1)}n$, while the query time stays the same. In the regime where $p=n^{\Omega(1)}$ (which is the focus of this paper), $\gamma(k,p,n)=k^{O(1)}$. These results are applicable to genus $g$ graphs, with $p$ replaced by $g$.

Our spanning tree cover with stretch $O(k\log\log\pi)$ for weakly tree-like $(\ell,\pi)$-path-separable graphs (with a total of $p$ separator paths) similarly yields a PRDO with size overhead and query time $O(k\ell\pi^{\frac{1}{k}}\cdot\log^2n)$ and $O(\ell\cdot\log n)\cdot (\log\pi\cdot\log n+k\pi^{\frac{1}{k}}\cdot\log\log\pi)$, respectively. For $K_r$-minor-free graphs, the bounds are $O(kr^{2+\frac{1}{k}}\cdot\log^2n)$ and $O(r^2\cdot\log n)\cdot (\log r\cdot\log n+kr^{\frac{1}{k}}\cdot\log\log r)$. If the separator is not tree-like, these bounds grow by an additional factor of $O(\log n)$. Similarly to the case of strongly path-separable graphs, here too we devise a more involved PRDO that has stretch $O(k)$, at the expense of increasing the size overhead by a factor of $\gamma(k,p,n)=\log^{O(1)}n$. See Table \ref{table:PRDOs} for a concise summary of new and existing results about PRDOs.

\begin{table}[htbp]
\begin{tabular}{|l|l|l|l|l|}
\hline
\textbf{Graph}  & \textbf{Stretch} & \textbf{Size} & \textbf{Query} & \textbf{Source} \\
\textbf{Family} & & \textbf{Overhead} & \textbf{Time} &  \\ 
 \hline
\multirow{4}{*}{\begin{tabular}[c]{@{}l@{}}$s$-vertex-separable\\ graphs  (or graphs \\ with treewidth $t$)\end{tabular}} 
 & $1$ & $O(s\log n)$ & $O(s\log n)$ & \cite{GKR04} \\ \cline{2-5} 
 & $\boldsymbol{O(k\log\log s)}$ & $\boldsymbol{O(k s^{\frac{1}{k}}\cdot\log n)}$ & $\boldsymbol{O(k s^{\frac{1}{k}}\cdot\log n)}$ & \textbf{New} \\ \cline{2-5}
 & $\boldsymbol{O(k)}$ & $\boldsymbol{O(k s^{\frac{1}{k}}\cdot\log n}$ & $\boldsymbol{O(k s^{\frac{1}{k}}\cdot\log n)}$ & \textbf{New} \\
 &  & $\boldsymbol{\cdot\gamma(k,s,n))}$ &  &  \\
 \hline 
\multirow{3}{*}{\begin{tabular}[c]{@{}l@{}}$p$-path-separable
\\ graphs\end{tabular}} 
 & $3$ & $O(p\log n)$ & $O(p\log n)$ & \cite{GKR04} \\ \cline{2-5}
 & $1+\epsilon$ & $O(p\cdot\frac{\log n}{\epsilon})$ \;\;\;\textit{(STL)} & $O(p\cdot\frac{\log n}{\epsilon})$ \;\;\;\textit{(STL)} & \cite{Tho04} \\ \cline{2-5}
 & $\boldsymbol{1+\epsilon}$ & $\boldsymbol{O(p\cdot\frac{\log^2n}{\epsilon})}$ & $\boldsymbol{O(p\cdot\frac{\log n}{\epsilon})}$ & \textbf{New} \\ \cline{2-5}
 \hline
 \multirow{4}{*}{\begin{tabular}[c]{@{}l@{}}Tree-like $(\ell,\pi)$-path\\-separable graphs\end{tabular}} 
 & $\boldsymbol{O(k\log\log\pi)}$ & $\boldsymbol{O(k\ell\pi^{\frac{1}{k}}\cdot\log^2n)}$ & $\boldsymbol{\ell\cdot O(\log\pi\cdot\log^2n}+$ & \textbf{New} \\
 &&& $\boldsymbol{k\pi^{\frac{1}{k}}\cdot\log n\cdot\log\log\pi)}$ & \\ \cline{2-5} 
 & $\boldsymbol{O(k)}$ & $\boldsymbol{O(k\ell\pi^{\frac{1}{k}}\cdot\log^2n}$ & $\boldsymbol{\ell\cdot O(\log\pi\cdot\log^2n}+$ & \textbf{New} \\
 & & $\boldsymbol{\cdot\gamma(k,\pi,n))}$ & $\boldsymbol{k\pi^{\frac{1}{k}}\cdot\log n\cdot\log\log\pi)}$ &  \\
 \hline
\multirow{8}{*}{\begin{tabular}[c]{@{}l@{}}$K_r$-minor-free\\ graphs\end{tabular}}  
 & $\boldsymbol{1+\epsilon}$ & $\boldsymbol{O(r^{c_{AG}}\cdot\frac{\log^2n}{\epsilon})}$ & $\boldsymbol{O(r^{c_{AG}}\cdot\frac{\log n}{\epsilon})}$ & \textbf{New} \\ \cline{2-5}
 & $3$ & $O(r^{c_{AG}}\log n)$ & $O(r^{c_{AG}}\log n)$ & \cite{GKR04} \\ \cline{2-5}
 & $O(1)$ & $2^{O(r)}\cdot\log\Lambda$ & $O(\log\log\Lambda)$ & \cite{Fil25} \\ \cline{2-5}
 & $O(r)$ & $O(r^2\cdot\log\Lambda)$ & $O(\log\log\Lambda)$ & \cite{Fil25} \\ \cline{2-5} 
 & $\boldsymbol{O(k\log\log r)}$ & $\boldsymbol{O(kr^{2+\frac{1}{k}}\cdot\log^2n)}$ & $\boldsymbol{r^2\cdot O(\log r\cdot\log^2n+}$ & \textbf{New} \\ 
 &&& $\boldsymbol{kr^{\frac{1}{k}}\cdot\log n\cdot\log\log r)}$ & \\ \cline{2-5}
 & $\boldsymbol{O(k)}$ & $\boldsymbol{O(kr^{2+\frac{1}{k}}\cdot\log^2n}$ & $\boldsymbol{r^2\cdot O(\log r\cdot\log^2n+}$ & \textbf{New} \\
 &  & $\boldsymbol{\cdot\gamma(k,r,n))}$ & $\boldsymbol{kr^{\frac{1}{k}}\cdot\log n\cdot\log\log r)}$ &  \\ \hline
 \multirow{5}{*}{\begin{tabular}[c]{@{}l@{}}Genus-$g$ graphs\end{tabular}}  
 & $\boldsymbol{1+\epsilon}$ & $\boldsymbol{O(g\cdot\frac{\log^2n}{\epsilon})}$ & $\boldsymbol{O(g\cdot\frac{\log n}{\epsilon})}$ & \textbf{New} \\ \cline{2-5}
 & $\boldsymbol{O(k\log\log g)}$ & $\boldsymbol{O(kg^{\frac{1}{k}}\cdot\log^2n)}$ & $\boldsymbol{O(\log g\cdot\log^2n+}$ & \textbf{New} \\ 
 &&& $\boldsymbol{kg^{\frac{1}{k}}\cdot\log n\cdot\log\log g)}$ & \\ \cline{2-5}
 & $\boldsymbol{O(k)}$ & $\boldsymbol{O(kg^{\frac{1}{k}}\cdot\log^2n}$ & $\boldsymbol{O(\log g\cdot\log^2n+}$ & \textbf{New} \\
 &  & $\boldsymbol{\cdot\gamma(k,g,n))}$ & $\boldsymbol{kg^{\frac{1}{k}}\cdot\log n\cdot\log\log g)}$ &  \\ \hline

\end{tabular}
\caption{A concise summary of new and existing results on PRDOs for vertex-separable and path-separable graphs with $n$ vertices and aspect ratio $\Lambda$, for any choice of an integer parameter $k\geq1$ and real parameter $\epsilon\in(0,1]$. 
The \textit{size overhead} of a PRDO is its size, divided by $n$.
For the PRDO of \cite{Tho04}, the marking \textit{(STL)} stands for \textit{strongly tree-like}, as this result holds only for strongly tree-like $p$-path-separable graphs.
The new results for tree-like path-separable graphs holds also for (not necessarily tree-like) path-separable graphs, with an additional factor of $\log n$ in the size overhead and query time.
Note that results for \textit{strongly} (tree-like) path-separable graphs are obtained by substituting $\ell=1$.
The results of \cite{Fil25} are implicit there: we use $\log\Lambda$ levels of their strong neighborhood covers with stretch $O(r)$ (respectively, $O(1)$) and size $O(r^2)$ (resp., $2^{O(r)}$).
The terms $\gamma(k,s,n)$, $\gamma(k,\pi,n)$ and $\gamma(k,r,n)$ are always at most $\log^{O(1)}n$. Moreover, if the second parameter is $n^{\Omega(1)}$, then these terms are at most $k^{O(1)}$.}
\label{table:PRDOs}
\end{table}

\subsection{A Technical Overview} \label{sec:TechnicalOverview}

A natural starting point for studying tree covers for path-separable graphs is the seminal construction of distance oracles (DOs) for planar graphs by \cite{Tho04,Kle02}. There are, in fact, two constructions in \cite{Tho04}, one of a non-path-reporting DO and another of a PRDO. Both generalize to strongly tree-like $p$-path-separable graphs.\footnote{Strictly speaking, there are parts of Thorup's construction \cite{Tho04} (such as the analysis of frames) that are specifically tailored to planar graphs. Though these parts enable to shave logarithmic factors from the size and the query time, they are not crucial for the scheme to work in general.} The non-path-reporting DO of \cite{Tho04} creates $t=O(\frac{1}{\epsilon})$ landmarks $\lambda^j_1(v),\lambda^j_2(v),...,\lambda^j_t(v)$ for every vertex $v$ and every path $Q^j$, $j\in\{1,2,...,p\}$, of the path-separator. One then argues that storing only the distances $\{d_G(v,\lambda^j_i(v))\;|\;1\leq i\leq t, 1\leq j\leq p\}$ from each vertex to each of its landmarks, and from each landmark to one endpoint of $Q^j$, enables to efficiently recover $(1+\epsilon)$-approximate distances for all vertex pairs. To provide the stretch bound, one shows that for every vertex pair $u,v$, whose shortest path $P_{u,v}$ traverses a separator path $Q^j$, there exists a pair of landmarks $\lambda^j_{uv}(u)$ and $\lambda^j_{uv}(v)$ of $u$ and $v$, respectively, both on $Q^j$, such that the path $u$-$\lambda^j_{uv}(u)$-$\lambda^j_{uv}(v)$-$v$ provides a good approximation for the shortest path $P_{u,v}$.

There are a number of problems that one encounters when trying to convert this DO to a tree cover. One of them is that these landmarks are \textit{not universal}, i.e., for a pair of distinct nearby vertices $u\neq v$, their respective sets of landmarks (on the same path $Q$ of the separator) may well be disjoint. As a result, building SPTs rooted at all landmarks produces far too many trees. Another problem is that this construction is not path-reporting, i.e., there is little hope that it can be converted into a \textit{spanning} tree cover. In particular, consider a vertex $v$, a landmark $\lambda(v)$ of $v$ on a separator path $Q$, and a vertex $u$ on a shortest $v$-$\lambda(v)$ path $P_{v,\lambda(v)}$. The set of landmarks of $u$ on $Q$ may well not contain $\lambda(v)$, and shortest paths from $u$ to its landmarks may intersect $P_{v,\lambda(v)}$. We will refer to this issue as the \textit{landmark continuity} problem.

One way around this was suggested in the pioneering work of Gupta et al. \cite{GKR04}. In their construction, every vertex $v$ connects to just one landmark $\lambda^j(v)$ on each path $Q^j$ of the separator. Specifically, $\lambda^j(v)$ is the closest vertex on $Q^j$ to $v$. This modification completely resolves the issues that were discussed above, and yields a spanning tree cover with size $O(p\log n)$. Alas, the stretch obtained via these nearest landmarks is $3$, as opposed to $1+\epsilon$ of Thorup's DO \cite{Tho04,Kle02}. Another approach was suggested by Bartal et al. \cite{BFN22}. In their construction, every vertex $v$ selects uniformly at random one of the landmarks $\lambda^j_i(v)$, $1\leq i\leq t$, for every separator path $Q^j$, $1\leq j\leq p$, and connects to it via a direct edge $(v,\lambda^j_i(v))$ (with weight $d_G(v,\lambda^j_i(v))$). This forms $p$ metric (i.e., non-spanning) trees -- one tree for each separator path. One then repeats this construction $O(\frac{\log^2n}{\epsilon^2})$ times, obtaining $O(p\cdot\frac{\log^2n}{\epsilon^2})$ metric trees. The stretch argument hinges on the observation that with high probability, for every pair $u,v\in V$ of vertices whose shortest path $P_{u,v}$ traverses a separator path $Q^j$ (for some $j$), in one of the $O(\frac{\log^2n}{\epsilon^2})$ trees built around $Q^j$ both $u$ and $v$ select the ``right" landmarks. This idea effectively converts the non-path-reporting DO of \cite{Tho04} into a non-spanning tree cover. However, it fails to address the landmark continuity problem, and thus it is inadequate for building \textit{spanning tree covers}.

The second construction of \cite{Tho04} provides \textit{path-reporting} distance oracles, and thus could a-priori serve as a better starting point for resolving this question. It builds a separate PRDO for every scale $\Delta=(1+\epsilon)^i$, $i=0,1,2,...$. The scale-$i$ PRDO caters for vertex pairs $u,v$ with $d_G(u,v)\leq\Delta$, and provides them with an additive approximation of $\epsilon\Delta$. To build this PRDO, Thorup \cite{Tho04} heavily exploited the ``tree-likeness" property of path-separators (for planar graphs). Recall that for a strong tree-like $p$-path separator $Q^1\cup Q^2\cup\cdots\cup Q^p$, all these paths are contained in an SPT $T$, rooted at a vertex $rt$. The construction of \cite{Tho04} then ``slices" the graph $G$ into sub-graphs $G_0,G_1,G_2,...$, with each $G_i$ containing vertices $v$ at distance $i\cdot\Delta\leq d_G(rt,v)\leq(i+2)\Delta$. In each slice $G_i$, one deletes all vertices with distance $>(i+2)\Delta$ from $rt$, and contracts all vertices at distance $<i\cdot\Delta$ from $rt$. Restricting the separator paths $Q^1,Q^2,...,Q^p$ to a single given slice $G_i$, one notes that their length within $G_i$ is at most $2\Delta$, and thus one can equip each of them with $\frac{2\Delta}{\epsilon\Delta}=O(\frac{1}{\epsilon})$ \textit{universal} landmarks, at pairwise distances $\epsilon\Delta$ between subsequent landmarks. Their universality is crucial, as this enables one to resolve the landmark continuity problem -- now for every vertex $v$ in $G_i$ and a landmark $\lambda^j(v)$ of $v$ on a separator path $Q^j$ (restricted to $G_i$), and a vertex $u$ on the shortest $v$-$\lambda^j(v)$ path, the vertex $u$ also needs to connect $\lambda^j(v)$. 

There are, however, a number of issues that arise if one attempts to convert this PRDO into a spanning tree cover. First, a PRDO is allowed to produce paths in a slice $G_i$, since they are also contained in $G$. On the other hand, a spanning tree cover would need to merge spanning trees of separate slices into fewer spanning trees of $G$ (otherwise, the resulting number of trees would explode). Moreover, contractions (used to obtain the $G_i$'s from $G$) are particularly problematic when building tree covers, as they alter the metric space induced by the graph -- some distances in the resulting trees may become shorter than the respective original distances. For PRDOs, this problem is circumvented by carefully selecting the slice $G_i$ in which one seeks an approximate shortest path for a given query vertex pair $u,v$. It may well be that in another slice $G_{i'}$, the distance between $u$ and $v$ is shorter than $d_G(u,v)$, but the PRDO is ``smart enough" to never fetch it. However, this type of approach is not applicable to (spanning) tree covers.

Yet another problem is that this PRDO is heavily tailored to \textit{strong tree-like} path separators. In a weak path separator, some separator paths may be shortest, not in the original graph $G$, but rather in the graph obtained from $G$ after removing some of the other separator paths. One can see that these separator paths will not necessarily become short when restricted to a slice $G_i$, and thus, not only that this approach does not yield a spanning tree cover -- it even does not yield a PRDO for weakly separable graphs. As a result, not only the problem of building spanning tree covers, but also the problem of building PRDOs, as opposed to non-path-reporting DOs, is open in weakly path-separable graphs or in $K_r$-minor-free graphs.

To summarize, both constructions of distance oracles of \cite{Tho04} appear to be inadequate for the problem of building a stretch-$(1+\epsilon)$ \textit{spanning} tree covers for path-separable graphs. Next, we sketch our construction of stretch-$(1+\epsilon)$ spanning tree covers. To facilitate presentation, we start with the simpler case of strongly $p$-path-separable graphs.

Like the construction of PRDOs due to \cite{Tho04}, our construction is a union of separate constructions for different scales $\Delta=2^i$, $i=0,1,2,...$ (we later get rid of the resulting dependence on $\log\Lambda$ via a separate weight reduction; see Theorem \ref{thm:SpanTreeCoverForPathSeparators}). For each given scale $\Delta$, we place \textit{universal landmarks} on each of the separator paths $Q^j$, $1\leq j\leq p$. Roughly speaking, these landmarks are placed at distance $\approx\epsilon\Delta$ between subsequent landmarks. Note, however, that the overall number of landmarks may be huge. Denote by $\lambda^j_1,\lambda^j_2,...$ the landmarks on the separator paths $Q^j$, for every $1\leq j\leq p$. For every landmark $\lambda^j_i$, consider a ball $B(\lambda^j_i)=B_G(\lambda^j_i,(1+\epsilon)\Delta)$. Recall that the objective of scale-$\Delta$ tree covers is to cater for vertex pairs $u,v$ with $d_G(u,v)\leq\Delta$, and thus, all vertices $v$ which are ``served" by a landmark $\lambda^j_i$ necessarily belong to the ball $B(\lambda^j_i)$ of this landmark. Note that for any two sufficiently remote landmarks $\lambda^j_i,\lambda^j_{i'}$ (with $|i-i'|$ being greater than $\rho=\frac{c}{\epsilon}$, for some small constant $c$), their respective balls are disjoint. We therefore can ``glue" together the SPTs for $B(\lambda^j_i),B(\lambda^j_{\rho+i}),B(\lambda^j_{2\rho+i}),...$ into a single spanning tree that caters for all these balls. We perform this for every $i=1,2,...,\rho$, and obtain $O(\frac{1}{\epsilon})$ spanning trees per separator path, and $O(\frac{p}{\epsilon})$ spanning trees overall, for a single recursion level. We do this in all $O(\log n)$ levels of recursion, and for each $O(\log_{1+\epsilon}\Lambda)$ different scales.

In the case of weakly path-separable graphs, the path $Q^j$, for $j>1$, is not necessarily a shortest path in $G$, but rather a shortest path in $H_j=G\setminus\bigcup_{i<j}Q^i$ (this is the case, e.g., in the path separator of \cite{AG06} for $K_r$-minor-free graphs). As a result, the balls $B_G(\lambda^j_i,(1+\epsilon)\Delta)$ and $B_G(\lambda^j_{i'},(1+\epsilon)\Delta)$ with large $|i-i'|$ are not necessarily disjoint. However, the balls $B_{H_j}(\lambda^j_i,(1+\epsilon)\Delta)$ and $B_{H_j}(\lambda^j_{i'},(1+\epsilon)\Delta)$ are disjoint. We show that this is sufficient for the stretch analysis to apply. 

Next, we overview our main construction: the construction of spanning tree covers that trade stretch for size, and provides the first tree covers (spanning or non-spanning) for general $p$-path-separable graphs with \textit{sublinear} size in $p$. To build some intuition, we first sketch a much simpler construction for $s$-vertex-separable graphs.

For this graph family, Gupta et al. \cite{GKR04} provided an exact (i.e., stretch-$1$) spanning tree cover of size $O(s(n)\log n)$. This tree cover consists of $s(n)$ SPTs for the entire graph, rooted at each vertex of the separator set $S$, for each of the $O(\log n)$ levels of the recursion. Our construction that trades stretch for size receives as an input a parameter $k$. Instead of using $s(n)$ SPTs, it invokes the metric Ramsey construction of Mendel and Naor \cite{MN06} on the set $S$. This construction provides us with a collection of $O(k\cdot s(n)^{\frac{1}{k}})$ HSTs, such that for every vertex pair $(s,v)\in S\times V$, there exists an HST in the collection that provides stretch $O(k)$ for this pair. As a result, this collection of HSTs provides stretch $O(k)$ for all the vertex pairs $u,v\in V$ whose shortest $u$-$v$ paths traverse $S$. Other vertex pairs are taken care by the recursive invocations of this construction on connected components of $G\setminus S$. Hence, the resulting HST cover provides stretch $O(k)$ and size at most $O(k\cdot s(n)^{\frac{1}{k}}\cdot\log n)$.\footnote{If $s(n)$ is large, then much better bounds apply. In particular, if $s(n)=n^\delta$, for $0<\delta\leq 1$, this is $O(\frac{k^2}{\delta}\cdot s(n)^{\frac{1}{k}})$.} To convert this into a spanning tree cover, we use the graph version of the metric Ramsey theorem due to \cite{ACEFN20}, which provides a spanning tree cover of the same size, but with stretch $O(k\log\log s(n))$.

To extend this construction to $p$-path-separable graphs, it is natural to try replacing the vertex separator $S$ by the set of landmarks of the separator paths $\{Q^j\}_{j=1}^p$. As we argued above, these landmarks need to be universal, and thus, as in the case of stretch $1+\epsilon$, our ultimate tree collection is a union of $\log\Lambda$ tree collections, one for every scale $\Delta=2^i$, $i=0,1,2,...$. On a scale $\Delta$, our tree collection provides multiplicative stretch $O(k)$ and additive stretch $O(\Delta)$ (ultimately, the additive stretch translates into a higher constant hidden in the $O$-notation of the multiplicative stretch). 

On each separator path $Q^j$, we place universal landmarks $\lambda^j_1,\lambda^j_2,...$, with distance $\approx\Delta$ between each pair of consecutive landmarks. Note, however, that the overall set of landmarks $\mathcal{L}=\{\lambda^j_i\;|\;i\geq1,\;1\leq j\leq p\}$ might be huge, and so one cannot just replace the set $S$ from the construction for small vertex separators by this set. It is tempting to consider together landmarks $\lambda^j_i,\lambda^j_{ck+i},\lambda^j_{2ck+i},...$, for a sufficiently large constant $c$, and to use the observation that their balls $B(\lambda^j_{i}), B(\lambda^j_{ck+i}),B(\lambda^j_{2ck+i})$ are disjoint. A natural attempt is to connect an imaginary vertex $\hat{\lambda}^j_i$ to all of them via zero-weight edges, for every $1\leq i\leq ck$. The resulting set of imaginary landmarks $\{\hat{\lambda}^j_1,\hat{\lambda}^j_2,...,\hat{\lambda}^j_{ck}\;|\;1\leq j\leq p\}$ is small, and one could attempt to use it instead of the set $S$.

Unfortunately, this attempt fails, as adding zero-weight edges ruins the metric. One can do it for distance oracles, but not for tree covers -- some vertex pairs $u,v$ may now end up being closer in the tree collection obtained from the metric Ramsey theorems of \cite{MN06,ACEFN20} than in $G$.

To overcome this hurdle, we define a graph $L$ whose vertex set is the set $\mathcal{L}$ of landmarks, and there is an edge between two landmarks iff their respective balls intersect. We then build a \textit{clustered coloring} $\varphi$ for this graph. A clustered coloring $\varphi$ of a graph $L$ with $Z$ colors and \textit{block size} $T$ is an assignment of colors from $\{1,2,...,Z\}$ to the vertices of $L$, in such a way that each connected component of vertices colored by the same color has size at most $T$. This notion was studied, e.g., in \cite{KMRV97,ADOV03,EF05}. See \cite{W12} for an extensive survey.

First, note that given a $(Z,T)$-clustered coloring of the graph $L$ as above, one indeed can invoke metric Ramsey theorems in each connected component of a fixed color class, and merge the resulting tree collections into a single tree collection (of size $O(kT^{\frac{1}{k}})$) for this color class. Doing it separately for every color class results in the overall tree cover of size $O(kT^{\frac{1}{k}}\cdot Z)$. Our problem is thus reduced to building a good clustered coloring (with small $Z$ and $T$) for the landmark graph $L$. We argue that the landmark graphs of $p$-path-separable graphs admit very good clustered colorings. In particular, in the case of strongly $p$-path-separable graphs, they admit clustered coloring with $Z=O(\log n)$ colors and block size $T=O(p\log p)$.

Intuitively, the reason for this is that the landmark graph $L$ has a very small expansion. It is not hard to see that for any landmark $\lambda$ and a positive integer $R$, the $R$-ball around $\lambda$ in $L$ may contain at most $O(R\cdot p)$ landmarks. Indeed, it can contain $O(R)$ landmarks on every separator path $Q^j$, as otherwise it would contain two landmarks that are too far from one another. This, however, would be a contradiction to the assumption that they belong to the same $R$-ball. It follows that one can carve disjoint balls of radius $O(\log p)$ each, that are at unweighted distance at least $2$ from one another (in $L$), and whose union covers a constant fraction of all landmarks. Each such ball contains $O(p\log p)$ landmarks, i.e., $T=O(p\log p)$. The set of these balls forms the first color class, and then we carve the balls of the second color, etc. Overall, $Z=O(\log n)$ color classes cover all the landmarks, providing us with the desired coloring, which translates to the spanning and non-spanning tree collections that gracefully trade between stretch $k$ and size $\tilde{O}(kp^{\frac{1}{k}})$.\footnote{The $\tilde{O}$-notation here hides factors polylogarithmic in $n$.} In the case of weakly path-separable graphs, the construction is more involved, and the resulting bounds are somewhat weaker.

Finally, we devise a weight reduction that enables us to replace a factor of $\log\Lambda$ by $\log n$. A related weight reduction was devised by \cite{KS97}. Below we argue that it is not applicable in the context of tree covers. The reduction of \cite{KS97} is typically used in the context of spanners and hopsets \cite{C98,C00,EN19,EM21}. It defines graphs $G_1,G_2,...,G_\lambda$, with $\lambda\approx\log\Lambda$, one graph for each distance scale $[2^{i-1},2^i]$. The aspect ratio of each of these graphs is polynomial in $n$. In addition, their numbers of vertices $n_1,n_2,..,n_\lambda$ sum up to $O(n\log n)$. Applying a black-box spanner/hopset construction to it, one creates spanners/hopsets $H_1,H_2,...,H_\lambda$ and returns their union $H$ as the ultimate output. For each $i=1,2,...,\lambda$, $|H_i|$ is typically bounded by $\approx O(n_i^{1+\frac{1}{k}})$, for a positive integer parameter $k$, and so $|H|\approx O(\sum_{i=1}^\lambda n_i^{1+\frac{1}{k}})=O(n^{1+\frac{1}{k}}\log n)$. Note, however, that in the context of tree covers, naively one would have to compute tree covers $\mathcal{T}_1,\mathcal{T}_2,...,\mathcal{T}_\lambda$ for $G_1,G_2,...,G_\lambda$, respectively, and return their union $\mathcal{T}=\bigcup_{i=1}^\lambda\mathcal{T}_i$. However, even if each $|\mathcal{T}_i|$ is bounded by $\tilde{O}(n_i^{\frac{1}{k}})$, it is no longer true that $\sum_{i=1}^\lambda|\mathcal{T}_i|=\tilde{O}(\sum_{i=1}^\lambda n_i^{\frac{1}{k}})$ is bounded by $\tilde{O}(n^{\frac{1}{k}})$ (for example, if $n_i\approx\frac{n}{\lambda}$ for every $i$, then $\sum_{i=1}^\lambda n_i^{\frac{1}{k}}=\lambda^{1-\frac{1}{k}}n^{\frac{1}{k}}$).

Our reduction, on the other hand, considers together tree covers $\mathcal{T}_1,\mathcal{T}_{1+t},\mathcal{T}_{1+2t},...$, where $t\approx\log\frac{n^2}{\epsilon}$. Note that there is a huge gap of $\frac{n^2}{\epsilon}$ between two consecutive scales in this subsequence. Any simple path in $G_{1+it}$, for any $i=0,1,2,...$, has negligible length in comparison to the minimum distance in the next scale graph $G_{1+(i+1)t}$. As a result, we can glue together the tree covers $\mathcal{T}_1,\mathcal{T}_{1+t},\mathcal{T}_{1+2t},...$ into one single tree cover $\hat{\mathcal{T}}_1$ of size bounded by $\max\{|\mathcal{T}_{1+it}|\;|\;i=0,1,2,...\}$. We then do the same for each $j=2,3,...,t$, i.e., build a single tree cover $\hat{\mathcal{T}}_j$ for all scales $G_j,G_{j+t},G_{j+2t},...$. As a result, we obtain a tree cover $\hat{\mathcal{T}}=\bigcup_{j=1}^t\hat{\mathcal{T}}_j$ of size just $t=O(\log n)$-times larger than the size of a single-scale tree cover $\mathcal{T}_i$, eliminating the dependence on $\log\Lambda$.

\subsection{Related work}

Extending previous results of Arya et al. \cite{ADMSS95} regarding Euclidean metrics, Chan et al. \cite{CGMZ16} and later Bartal et al. \cite{BFN22} explored \textit{doubling metrics}, i.e., metrics with doubling constant at most $\lambda$, for a parameter $\lambda$. Improving upon a result of \cite{CGMZ16}, \cite{BFN22} showed that for any positive integer $k$, doubling metrics admit tree covers with stretch $O(k)$ and with $O(\lambda^{\frac{1}{k}}\log\lambda\cdot\log k)$ trees. They also devised a construction of tree covers with stretch $1+\epsilon$, for an arbitrarily small $0<\epsilon<1$, and $(\frac{1}{\epsilon})^{O(\log\lambda)}$ trees. Chang et al. \cite{CCLST25} devised tree covers for \textit{doubling graphs}, i.e., graphs whose induced shortest-path metric is doubling. The tree covers of \cite{CGMZ16,BFN22} can be viewed as non-spanning tree covers for doubling graphs. The spanning tree cover of \cite{CCLST25} for doubling graphs has stretch $1+\epsilon$ and uses $(\frac{1}{\epsilon})^{\tilde{O}(\log\lambda)}$ trees. Lower bounds for tree covers with a small number of trees were shown in \cite{CTX25}.

\subsection{Organization}

In Section \ref{sec:Preliminaries} we provide necessary definitions regarding metric spaces, spanners and emulators, and graph separators. In Section \ref{sec:TreeCoversBasic}, we define the different types of tree covers and show some basic properties and connections between them. Then, in Section \ref{sec:RamseyThms} we develop our main toolkit for constructing tree covers, by generalizing some metric Ramsey theorems. Section \ref{sec:TreeCoversForSeparableGraphs} contains the main contribution of the paper, i.e., constructions of spanning and non-spanning tree covers for vertex-separable and path-separable graphs (for vertex-separable graphs, see Section \ref{sec:VertexSeparableTreeCovers}, and for path-separable graphs, see Sections \ref{sec:PathSeparableSpanSmallStretchTreeCovers} and \ref{sec:PathSeparableLargeStretchTreeCovers}). Next, in Section \ref{sec:SpannersAndEmulators}, we use our new tree covers to obtain path-reporting spanners and emulators for path-separable graphs. 
In Appendix \ref{app:ComparisonGeneralGraphs} we compare our new tree covers with known results for tree covers in general graphs. Then, Appendix \ref{app:AspectRatioReduction} presents our weight reduction that enables us to eliminate the dependence on the aspect ratio in our tree covers (in Appendices \ref{app:PRSpannerAspectRatioReduction1} and \ref{app:PRSpannerAspectRatioReduction2} we also utilize this reduction to obtain improved query times for our PRDOs). Appendix \ref{app:ImprovedStretch} presents non-path-reporting distance oracles that provide an improved stretch. Finally, in Appendix \ref{app:DLandRS} we provide new results for distance labeling schemes (Section \ref{sec:DistanceLabeling}) and routing schemes (Section \ref{sec:RoutingSchemes}), for path-separable and vertex-separable graphs.

\section{Preliminaries} \label{sec:Preliminaries}

Unless stated otherwise, all logarithms are in base $2$.

Given an undirected graph $G=(V,E)$, we denote by $V(G)$ and $E(G)$ the sets of $G$'s vertices and edges, respectively. When the graph $G$ is clear from the context, we denote $n=|V(G)|$ and $m=|E(G)|$. We also denote $\binom{V}{2}=\{\{u,v\}\;|\;u,v\in V,\;u\neq v\}$. 

A \textbf{minor} of $G$ is a graph that is obtained from $G$ by a sequence of vertex deletions, edge deletions and edge contractions.

If the graph $G$ is weighted, with positive weights on the edges, we denote by $d_G(u,v)$ the distance between two vertices $u,v\in V$, i.e., the weight of the shortest path between them. If $u,v$ are on a path $P$, we denote by $P[u,v]$ the sub-path of $P$ between $u$ and $v$. The \textbf{diameter} of $G$ is $diam(G)=\max_{u,v\in V}d_G(u,v)$, and its \textbf{aspect ratio} is 
\[\Lambda=\frac{\max_{u,v\in V}d_G(u,v)}{\min_{u\neq v\in V}d_G(u,v)}=\frac{diam(G)}{\min_{e\in E}w(e)}~.\]

Given a rooted tree $T$ and two of its vertices $u,v$, we denote by $LCA_T(u,v)$ the lowest common ancestor of $u,v$ in the tree. That is, the sub-tree of $T$ rooted at $LCA_T(u,v)$ contains $u,v$, but none of the children of $LCA_T(u,v)$ satisfies the same property.


\subsection{Ultrametrics and HSTs} \label{sec:UMHST}

Given a metric space $(X,d)$, if $d$ satisfies $d(x,z)\leq\max\{d(x,y),d(y,z)\}$ for every $x,y,z\in X$, then $d$ is called an \textbf{ultrametric} on $X$, and $(X,d)$ is called an \textbf{ultrametric space}.

\begin{definition} \label{def:HST}
A \textbf{hierarchically (well) separated tree}, or \textbf{HST},\footnote{This is the definition of a $1$-HST. See \cite{BLMN03} for the more general definition of a $k$-HST, for a parameter $k$.} is a rooted tree $T=(V,E)$ with labels\footnote{The original definition by Bartal in \cite{B96} uses weights on the \textit{edges} instead of labels on the vertices. We use a different equivalent notion, that was given by Bartal et al. in another paper \cite{BLMN03}.} $\ell:V\rightarrow\mathbb{R}_{\geq0}$, such that if $v\in V$ is a child of $u\in V$, then $\ell(v)\leq\ell(u)$, and for every leaf $v\in V$, $\ell(v)=0$. 
\end{definition}

Let $x,y$ be two vertices in a rooted tree (e.g., in an HST), and let $LCA(x,y)$ denote their lowest common ancestor. Given an HST, let $L$ be its set of leaves, and define the function $d(x,y)=\ell(LCA(x,y))$. It is not hard to see that $(L,d)$ is an ultrametric space. 
In \cite{BLMN03}, Bartal et al. proved the following fact

\begin{fact} \label{fact:UM<->HST}
\textit{Every} finite ultrametric $(X,\rho)$ can be represented by an HST $T=(V,E)$ with labels $\ell:V\rightarrow\mathbb{R}_{\geq0}$, where the leaves of $T$ are exactly the points of $X$, and $\rho(x,y)=\ell(LCA_T(x,y))$ for every $x,y\in X$.
\end{fact}

For convenience, we always assume that an HST $T$ contains labels $\ell$, and we denote $\rho_T(u,v)=\ell(LCA_T(u,v))$, for every $u,v\in L(T)$ (where $L(T)$ is the set of leaves of $T$).

Elkin and Shabat \cite{ES23} showed how to use a result by Gupta \cite{Gupta01} to construct a tree $T=(X,E)$ such that the shortest path metric in $T$ equals $\rho$, up to a stretch of at most $8$ (see also Konjevod et al. \cite{KRS01}, Theorem 5.1, for a related argument). For convenience, we state this result here.

\begin{lemma}[\cite{ES23}] \label{lemma:FromUMToTree}
Let $(X,\rho)$ be an ultrametric. Then, there is a tree $T=(X,E)$ such that for every $u,v\in X$,
\[\rho(u,v)\leq d_T(u,v)\leq8\rho(u,v)~.\]
\end{lemma}

Lemma \ref{lemma:FromUMToTree} can be thought of as an approximation of an ultrametric by a \textit{tree metric} (i.e., a metric space $(V,d_T)$, where $V$ are the vertices of a tree $T$, and $d_T$ is the shortest path metric of $T$).

\subsection{Spanners and Emulators} \label{sec:PreSpannersAndEmulators}

Consider the fundamental objects of graph \textit{spanners} and \textit{emulators}.

\begin{definition} \label{def:SpannersAndEmulators}
Given an undirected weighted graph $G=(V,E)$, an \textbf{$\alpha$-emulator} is a weighted graph $G'=(V,E')$ (not necessarily $E'\subseteq E$) such that for every $u,v\in V$,
\[d_G(u,v)\leq d_{G'}(u,v)\leq\alpha\cdot d_G(u,v)~.\]
If $E'\subseteq E$, we say that $G'$ is an \textbf{$\alpha$-spanner}.
We say that the emulator/spanner $G'$ has \textbf{stretch} at most $\alpha$ and size $|E'|$.
\end{definition}

We sometimes identify an emulator or a spanner with its set of edges (thus calling $E'$ itself an emulator/spanner).


In this paper we provide results for the stronger notions of \textit{path-reporting} spanners/emulators. The following definitions were introduced in \cite{ES23}.

\begin{definition}[(Pairwise) Path-Reporting Emulators and Spanners] \label{def:PathReportingStructures}
Given an undirected weighted graph $G=(V,E)$ and a subset $\mathcal{P}\subseteq V^2$ of vertex pairs, a \textbf{$\mathcal{P}$-pairwise path-reporting $\alpha$-emulator} with query time $q$ is a pair $(S,D)$, where $S\subseteq\binom{V}{2}$ is called the set of \textbf{underlying edges}, equipped with weights $w(u,v)=d_G(u,v)$ for every $\{u,v\}\in S$, and $D$ is an oracle that receives a query $(u,v)\in V^2$ as an input, and can provide either only a distance estimate $\hat{d}(u,v)$ within time $q$, or a distance estimate $\hat{d}(u,v)$ and a $u$-$v$ path $P_{u,v}\subseteq S$, with weight $w(P_{u,v})=\hat{d}(u,v)$, within time $q+O(|P_{u,v}|)$. The distance estimate must satisfy $d_G(u,v)\leq\hat{d}(u,v)$, and if $(u,v)\in\mathcal{P}$, also
\[\hat{d}(u,v)\leq\alpha\cdot d_G(u,v)~.\]
The size of the path-reporting emulator $(S,D)$ is defined as the maximum between $|S|$ and the storage size of $D$.
If $S\subseteq E$, then $(S,D)$ is called a \textbf{$\mathcal{P}$-pairwise path-reporting $\alpha$-spanner}. In the case that $\mathcal{P}=V^2$ we omit the word ``pairwise", i.e., we call $(S,D)$ a \textbf{path-reporting $\alpha$-emulator}, or \textbf{path-reporting $\alpha$-spanner}.
\end{definition}

Path-reporting emulators and spanners are a more restricted form of classic emulators and spanners. On the other hand, these notions can be viewed as special cases of \textit{distance oracles} and \textit{path-reporting distance oracles}.

\begin{definition} \label{def:DistanceOracles}
    Given an undirected weighted graph $G=(V,E)$, a \textbf{distance oracle} with stretch $\alpha$ is a data structure $D$, such that when given a pair of vertices $(u,v)\in V^2$, called query, $D$ outputs a distance estimate $\hat{d}(u,v)$, that satisfies
    \[d_G(u,v)\leq\hat{d}(u,v)\leq\alpha\cdot d_G(u,v)~.\]
    The maximum time it takes for $D$ to output an estimate is called the query time of $D$. The size of $D$ is defined as the total storage space of the data structure $D$.

    A \textbf{path-reporting distance oracle}, or shortly, a \textbf{PRDO}, is a distance oracle that, given a query $(u,v)\in V^2$, can provide either only a distance estimate $\hat{d}(u,v)$ as above, or a distance estimate $\hat{d}(u,v)$ and a $u$-$v$ path $P_{u,v}\subseteq E$ in $G$ of weight $\hat{d}(u,v)$. The query time of a PRDO $D$ is the minimal value $q$ such that, upon a query, $D$ outputs a distance estimate in $q$ time, and outputs a distance estimate and a corresponding $u$-$v$ path in $q+O(|P_{u,v}|)$ time.
\end{definition}

\subsection{Vertex Separators and Path Separators} \label{sec:SeparatorsPreliminaries}

Below we define various notions of graph \textit{separators}.

\begin{definition} \label{def:VertexSeparator}
Given an $n$-vertex undirected graph $G=(V,E)$, a \textbf{balanced vertex separator}, or shortly, a \textbf{separator}, is a subset $S\subseteq V$ such that each connected component of the graph $G'=G[V\setminus S]$ contains at most $\frac{n}{2}$ vertices.
\end{definition}

A well-known bound, that connects the graph \textit{treewidth} and its separators, is provided below (see, e.g., \cite{BGHK95}, Lemma 6).

\begin{fact} \label{fact:Separators}
Any graph $G=(V,E)$ with treewidth at most $t$ has a separator of size at most $t+1$.
\end{fact}


For several families of graphs (e.g., planar or constant treewidth graphs), one can show that not only the graph itself has a small separator, but so does each of its sub-graphs. This motivates the following definition.

\begin{definition} \label{def:RecursiveSeparators}
Given an undirected graph $G=(V,E)$, we say that $G$ is \textbf{$s$-vertex-separable} for some function $s=s(\theta)$, if every sub-graph of $G$ with $\theta$ vertices has a separator of size at most $s(\theta)$.
\end{definition}

Next, we define path separators.

\begin{definition}[Extended\footnote{Only Items (1)-(2) appear in \cite{AG06}.} definitions from \cite{AG06}] \label{def:PathSeparable}
  Suppose that every sub-graph $G'$ of a graph $G$ has a balanced vertex separator of the form $P_0\cup P_1\cup\cdots\cup P_{\ell-1}$, where $P_i$ is a union of $p_i$ shortest paths in the graph $H_i=G'\setminus\bigcup_{j<i}P_j$ (i.e., $G'$ with the vertices of $\bigcup_{j<i}P_j$ removed).
  \begin{enumerate}
      \item We say that $G$ is \textbf{(weakly) $p$-path-separable} if $\sum_{i=0}^{\ell-1}p_i\leq p$.
      \item We say that $G$ is \textbf{strongly $p$-path-separable} if $\ell=1$, and thus the separator $P_0$ consists of at most $p$ shortest paths in $G$.
      \item We say that $G$ is \textbf{(weakly) $(\ell,\pi)$-path-separable} if $p_i\leq\pi$ for every $0\leq i<\ell$.
      \item We say that $G$ is \textbf{(weakly) tree-like $(\ell,\pi)$-path-separable} (or \textbf{ (weakly) tree-like $p$-path-separable}) if it is $(\ell,\pi)$-path-separable (or $p$-path-separable), and the paths that form $P_i$ share a common endpoint $r_i\in V$ (and thus $P_i$ is a tree rooted at $r_i$), for every $0\leq i<\ell$.
  \end{enumerate}
\end{definition}

\begin{theorem}  \label{thm:PathSeparableFamilies}
    The following graph families are path-separable.
    \begin{enumerate}
        \item By \cite{LT79,Tho04}, any planar graph is strongly tree-like $3$-path-separable.
        \item By \cite{AG06,GSW25}, any $H$-minor-free graph is both tree-like $p$-path-separable and tree-like $(\ell,\pi)$-path-separable, where $\ell\leq|E(H)|$, and
        $p=\pi=O(RS(H)^2\cdot|E(H)|)$, where $RS(H)$ is the Robertson-Seymour function \cite{RS86b,RS03}. By \cite{GSW25}, $RS(H)=O(|V(H)|^{2300})$.
        \item By \cite{AD96}, any graph with orientable or non-orientable genus at most $g$ is strongly tree-like $(2g+2)$-path-separable.
    \end{enumerate}
\end{theorem}


\section{Tree Covers: Definitions and Basic Properties} \label{sec:TreeCoversBasic}

The notion of \textit{tree covers} of a graph plays a crucial role in the construction of a variety of graph metric structures, such as spanners, emulators, distance oracles, distance labeling and routing schemes. The use of tree covers for these purposes appears, both implicitly and explicitly, in many papers, such as \cite{AP92,TZ01,TZ01-spaa,KLMS22}. In each of these results, the definition of a tree cover of a graph $G=(V,E)$ may slightly vary. In particular, they differ on whether the trees are sub-graphs of $G$, or whether their vertex set is the entire $V$ or rather just a subset of $V$. In our paper, we use a generalized definition of tree covers that captures all of these different scenarios. In fact, we first define \textit{forest} covers, rather than tree covers, and also devise a version of forest covers (called \textit{pairwise} forest covers) that satisfies the stretch guarantee only for some of the pairs in $V^2$.

\begin{definition}[(Pairwise) Forest Covers and Tree Covers] \label{def:PairwiseForestCover}
Given an undirected weighted graph $G=(V,E)$, a real number $\alpha\geq1$, and a subset $\mathcal{P}\subseteq V^2$ of vertex pairs, a \textbf{$\mathcal{P}$-pairwise forest cover} with stretch $\alpha$ is a set $\mathcal{F}$ of weighted\footnote{In case of a forest edge that does not exist in the input graph $G$, we allow any positive edge weight, including $\infty$ (this is necessary for graphs that consist of more than one connected component).} forests, where $V(F)\subseteq V$ for every $F\in\mathcal{F}$, such that $d_G(u,v)\leq\min_{F\in\mathcal{F}}d_F(u,v)$ for every $u,v\in V$, and if $(u,v)\in\mathcal{P}$,
\[min_{F\in\mathcal{F}}d_F(u,v)\leq\alpha\cdot d_G(u,v)~.\]

The (maximum) \textbf{overlap} of $\mathcal{F}$ is $\max_{v\in V}|\{F\in\mathcal{F}\;|\;v\in V(F)\}|$, and its \textbf{average overlap} is
\[\frac{1}{|V|}\sum_{v\in V}|\{F\in\mathcal{F}\;|\;v\in V(F)\}|=\frac{1}{|V|}\sum_{F\in\mathcal{F}}|V(F)|~.\]
The sum $\sum_{F\in\mathcal{F}}|V(F)|$, is called the \textbf{covering} of $\mathcal{F}$, and $|\mathcal{F}|$ is its \textbf{size}.

In the case that $\mathcal{P}=V^2$ we omit the word ``pairwise", i.e., we call $\mathcal{F}$ a \textbf{forest cover}.
If for all $F\in\mathcal{F}$, the forest $F$ is a tree, we call $\mathcal{F}$ a \textbf{(pairwise) tree cover}.

\end{definition}

Note that the forests of a pairwise forest cover do not have to be sub-graphs of the given graph $G=(V,E)$. However, an important special case is when they are sub-graphs, i.e., $E(F)\subseteq E$ for every $F\in\mathcal{F}$. We call such a forest cover a \textbf{spanning (pairwise) forest cover}, or \textbf{spanning (pairwise) tree cover} when the forests are trees.

We say that a tree cover is a \textbf{Ramsey tree cover} if every vertex $u\in V$ has a \textit{home tree} $T=home(u)\in\mathcal{T}$, such that for every $v\in V$, 
\[d_T(u,v)\leq\alpha\cdot d_G(u,v)~,\]
where $\alpha$ is the stretch of the tree cover.

We emphasize that every tree cover is also a forest cover with the same properties, since a tree is a special case of a forest. In the other direction, a forest cover can be perceived as a tree cover, possibly with a larger number of trees. Observation \ref{obs:FromForestCoverToTreeCover} below elaborates on these connections between forest covers and tree covers.

\begin{observation}[Gluing Trees] \label{obs:FromForestCoverToTreeCover}
    Let $\mathcal{F}$ be a $\mathcal{P}$-pairwise forest cover of $G=(V,E)$, where $\mathcal{P}\subseteq V^2$. The set $\mathcal{T}$ of all trees in the forests of $\mathcal{F}$ is a $\mathcal{P}$-pairwise tree cover with the same stretch, overlap and average overlap as $\mathcal{F}$ (but not necessarily with the same size). If $\mathcal{F}$ is spanning, so is $\mathcal{T}$.

    In addition, there is a $\mathcal{P}$-pairwise tree cover $\bar{\mathcal{T}}$, which is not necessarily spanning (even if $\mathcal{F}$ is spanning), with equal or smaller stretch, overlap and average overlap, such that $|\bar{\mathcal{T}}|=|\mathcal{F}|$.
    
    Lastly, if $G$ is a connected graph and $\mathcal{F}$ is spanning, there is a spanning $\mathcal{P}$-pairwise tree cover $\bar{\mathcal{T}}_{span}$ of $G$ with $|\bar{\mathcal{T}}_{span}|=|\mathcal{F}|$ and equal or smaller stretch as $\mathcal{F}$ (with no guarantee on the overlap).
\end{observation}

\begin{proof}

Consider the tree cover $\mathcal{T}$ that is the collection of all trees in all forests of $\mathcal{F}$. For every $u,v\in V$ and every forest $F\in\mathcal{F}$ we have $d_F(u,v)=\min_{\text{tree }T\text{ in }F}d_T(u,v)$, and therefore $\mathcal{T}$ has the same stretch as of $\mathcal{F}$. Moreover, since for every $v\in V$ at most one of the trees in a forest $F\in\mathcal{F}$ contains $v$, we get 
\[|\{T\in\mathcal{T}\;|\;v\in V(T)\}|=|\{F\in\mathcal{F}\;|\;v\in V(F)\}|~.\] 
Thus, the overlap and average overlap of $\mathcal{T}$ are the same as of $\mathcal{F}$.

Next, construct a tree cover $\bar{\mathcal{T}}$ as follows. Consider a forest $F\in\mathcal{F}$, consisting of the trees $T^0,T^1,...,T^f$, and let $v_i$ be an arbitrary vertex in $T^i$, for every $0\leq i\leq f$. Then, for every $i=1,2,...,f$, add to $F$ the edge $\{v_0,v_i\}$ with weight $d_G(v_0,v_i)$. Denote the resulting tree by $T_F$, and note that for every $u,v\in V$ we have $d_G(u,v)\leq d_{T_F}(u,v)\leq d_F(u,v)$. Thus, the stretch of the tree cover $\bar{\mathcal{T}}=\{T_F\}_{F\in\mathcal{F}}$ is at most the stretch of $\mathcal{F}$. Moreover, note that $V(T_F)=V(F)$ for every $F\in\mathcal{F}$, and therefore $\bar{\mathcal{T}}$ has the same overlap and average overlap as $\mathcal{F}$. In addition, we clearly have $|\bar{\mathcal{T}}|=|\mathcal{F}|$, by the definition of $\bar{\mathcal{T}}$.

Suppose now that $G$ is connected. Given a forest $F\in\mathcal{F}$, we now define the tree $T_F$ by iteratively adding edges from $E$ to $F$, as long as no cycles are formed. The resulting graph $T_F$ has no cycles and is connected -- since for every edge $\{x,y\}$ on every path in $G$, either $\{x,y\}$ is in $T_F$ or there is an $x$-$y$ path in $T_F$. Thus, $T_F$ is a tree such that $V(T_F)\subseteq V$ and $E(F)\subseteq E(T_F)\subseteq E$. We define a tree cover $\bar{\mathcal{T}}_{span}=\{T_F\}_{F\in\mathcal{F}}$, which clearly satisfies $|\bar{\mathcal{T}}_{span}|=|\mathcal{F}|$. Note that $d_G(u,v)\leq d_{T_F}(u,v)\leq d_F(u,v)$, for every $u,v\in V$ and $F\in\mathcal{F}$. This shows that the stretch of $\mathcal{T}$ is bounded by the stretch of $\mathcal{F}$.

\end{proof}

The next lemma argues that forest covers for disjoint vertex sets can be merged into a single forest cover.

\begin{lemma} \label{lemma:UnifyDisjointForestCovers}
    Let $G=(V,E)$ be a graph with connected components $C_1,C_2,C_3,...$, and suppose that every $C_j$ admits a (spanning) forest cover $\mathcal{F}_j$ with at most $K$ forests, stretch at most $\alpha$, 
    overlap at most $t$ and average overlap at most $s$. Then, $G$ admits a (spanning) forest cover $\tilde{\mathcal{F}}$ with at most $K$ forests, 
    overlap at most $t$ and average overlap at most $s$, such that for every $j$, if $u,v\in V$ are in $C_j$, then $\min_{F\in\tilde{\mathcal{F}}}d_F(u,v)=\min_{F\in\mathcal{F}_j}d_F(u,v)$. In particular, the stretch of $\tilde{\mathcal{F}}$ is at most $\alpha$.
\end{lemma}

\begin{proof}

For every $1\leq h\leq K$, let $F^h_j$ be the $h$-th forest of $\mathcal{F}_j$, given some arbitrary ordering of these forests (if there are less than $h$ forests in $\mathcal{F}_j$, define $F^h_j=\emptyset$). Since for a fixed $h$, the forests $\{F^h_j\}_j$ are pairwise disjoint, their union $\tilde{F}^h=\bigcup_{j}F^h_j$ is still a (spanning) forest. We now denote 
\[\tilde{\mathcal{F}}=\left\{\tilde{F}^h\;|\;1\leq h\leq K\right\}~.\]
This is a collection of at most $K$ forests in $G$. Since the components $\{C_j\}$ are pairwise disjoint, it is immediate to verify that the overlap and average overlap of $\tilde{\mathcal{F}}$ are still at most $t$ and $s$, respectively.

If $u,v\in V$ are in $C_j$, then $d_{\tilde{F}^h}(u,v)=d_{F^h_j}(u,v)$ for every $1\leq h\leq K$, as $F^h_j$ is the only forest in $\tilde{F}^h$ that is in $C_j$ (if $F^h_j$ does not contain $u,v$, then $d_{\tilde{F}^h}(u,v)=d_{F^h_j}(u,v)=\infty$). Thus,
\[\min_{F\in\tilde{\mathcal{F}}}d_F(u,v)=\min_{1\leq h\leq K}d_{\tilde{F}^h}(u,v)=\min_{1\leq h\leq K}d_{F^h_j}(u,v)=\min_{F\in\mathcal{F}_j}d_F(u,v)~.\]
    
\end{proof}

As a variation of (non-spanning) tree covers, we also consider the following notion of \textit{HST covers}, where the input graph is covered by a collection of ultrametrics, that can be thought of as HSTs (see Definition \ref{def:HST}).

\begin{definition}[(pairwise) HST Cover] \label{def:HSTCover}
Given an undirected weighted graph $G=(V,E)$, a real number $\alpha\geq1$, and a subset $\mathcal{P}\subseteq V^2$ of vertex pairs, a \textbf{$\mathcal{P}$-pairwise HST cover} with stretch $\alpha$ is a set $\mathcal{T}$ of HSTs, each $T\in\mathcal{T}$ with a set $L(T)\subseteq V$ of leaves, such that $d_G(u,v)\leq\min_{T\in\mathcal{T}}\rho_T(u,v)$ for every $u,v\in V$, and if $(u,v)\in\mathcal{P}$, then
\[\min_{T\in\mathcal{T}}\rho_T(u,v)\leq\alpha\cdot d_G(u,v)~.\]

The (maximum) \textbf{overlap} of $\mathcal{T}$ is $\max_{v\in V}|\{T\in\mathcal{T}\;|\;v\in L(T)\}|$, its \textbf{average overlap} is $\frac{1}{|V|}\sum_{T\in\mathcal{T}}|L(T)|$, and its \textbf{covering} is $\sum_{T\in\mathcal{T}}|L(T)|$.
If $\mathcal{P}=V^2$ we call $\mathcal{T}$ an \textbf{HST cover}.
\end{definition}

Note that using Lemma \ref{lemma:FromUMToTree} (and recalling that every HST represents an ultrametric), it is possible to convert any HST cover with stretch $\alpha$ into a (not necessarily spanning) tree cover with stretch $8\alpha$, and with the same maximum and average overlap.

\section{Metric Ramsey Theorems for Trees} \label{sec:RamseyThms}

In our constructions of spanning tree covers, we use a result by \cite{ACEFN20}, which is described in the following theorem.

\begin{theorem}[\cite{ACEFN20}, Theorem 6] \label{thm:RamseyTree}
Let $G=(V,E)$ be a connected undirected weighted graph, let $A\subseteq V$ be a subset of size $t$, and let $k\geq1$ be some parameter. There exists a subset $S\subseteq A$ of size at least $t^{1-\frac{1}{k}}$, and a spanning tree $T$ of $G$, such that for every $u\in S$ and $v\in V$,
\[d_T(u,v)\leq\hat{c}\cdot k\log\log(t+4)\cdot d_G(u,v)~,\footnote{The stretch expression in Theorem \ref{thm:RamseyTree} is originally presented as $O(k\log\log t)$. However, here we specify the leading constant $\hat{c}$, and we also replace $t$ by $t+4$, so that the theorem holds also for $t\leq2$ (otherwise, for $0\leq t\leq2$, $\log\log t$ is either zero or undefined).}\]
for some universal constant $\hat{c}$.
\end{theorem}

Theorem \ref{thm:RamseyTree} is true for any \textit{connected} graph $G$. If $G$ is not connected, the spanning tree $T$ is replaced by a spanning \textit{forest} of $G$, as we prove in the next lemma. Here, $\hat{c}$ is the same constant from Theorem \ref{thm:RamseyTree}.

\begin{lemma} \label{lemma:RamseyForest}
Let $G=(V,E)$ be an undirected weighted graph, not necessarily connected, let $A\subseteq V$ be a subset of size $t$, and let $k\geq1$ be some parameter. There exists a subset $S\subseteq A$ of size at least $t^{1-\frac{1}{k}}$, and a spanning forest $F=(V,E_F)$ of $G$, such that for every $u\in S$ and $v\in V$,
\begin{equation} \label{eq:RamseyForest}
    d_F(u,v)\leq\hat{c}\cdot k\log\log(t+4)\cdot d_G(u,v)~.
\end{equation}
\end{lemma}

\begin{proof}

Denote the connected components of $G$ by $C_1,C_2,...,C_\ell$, and by $n_i$ the number of vertices in $C_i$. In addition, let $A_i$ be the intersection of $A$ with the component $C_i$, and let $t_i=|A_i|$. We apply Theorem \ref{thm:RamseyTree} on every connected component $C_i$, and obtain, for every $i=1,2,...,\ell$, a subset $S_i\subseteq A_i$ of size at least $t_i^{1-\frac{1}{k}}$ and a spanning tree $T_i$ of $C_i$, such that for every $u\in S_i$ and $v$ in $C_i$,
\[d_{T_i}(u,v)\leq\hat{c}\cdot k\log\log(t_i+4)\cdot d_{C_i}(u,v)~.\]

We claim that the union $S=\bigcup_{i=1}^\ell S_i$, and the forest $F$ which is the union of $\{T_i\}_{i=1}^\ell$, are the desired subset and forest. Indeed, the size of $S$ is
\[|S|=\sum_{i=1}^\ell |S_i|\geq\sum_{i=1}^\ell t_i^{1-\frac{1}{k}}\geq\left(\sum_{i=1}^\ell t_i\right)^{1-\frac{1}{k}}=t^{1-\frac{1}{k}}~,\]
where we used here the fact that $x^\gamma+y^\gamma\geq(x+y)^\gamma$, for any numbers $x,y\geq0$ and $\gamma\leq1$. In addition, for every $u\in S$ and $v\in V$, if $u,v$ are not in the same connected component $C_i$, then $d_F(u,v)=\infty=d_G(u,v)$. Therefore, Inequality (\ref{eq:RamseyForest}) is satisfied. Otherwise, suppose that $u,v$ are in the same component $C_i$. Then by Theorem \ref{thm:RamseyTree}, we have
\[d_F(u,v)\leq d_{T_i}(u,v)\leq\hat{c}\cdot k\log\log(t_i+4)\cdot d_{C_i}(u,v)\leq\hat{c}\cdot k\log\log(t+4)\cdot d_G(u,v)~.\]

\end{proof}

Next, we state an analogous version of Theorem \ref{thm:RamseyTree}, for non-spanning forests and for HSTs. Its proof uses a simple combination of Theorem 1.2 from \cite{NT12} and Theorem 4.1 from \cite{MN06}. For convenience, we explicitly state these theorems here.

\begin{theorem}[Theorem 1.2 from \cite{NT12}] \label{thm:RamseyUM}
For every $n$-point metric space and $k\geq1$, there exists a subset of size $n^{1-\frac{1}{k}}$ that can be embedded into an ultrametric with distortion $2ek$.\footnote{A weaker previous bound of this form was proved by \cite{MN06}.}
\end{theorem}

\begin{theorem}[Theorem 4.1 from \cite{MN06}] \label{thm:ExtendingUM}
Let $(X,d)$ be a finite metric space, and $\alpha\geq1$. Fix a subset $\emptyset\neq Y\subseteq X$, and assume that there exists an ultrametric $\rho$ on $Y$ such that for every $x,y\in Y$, $d(x,y)\leq\rho(x,y)\leq\alpha\cdot d(x,y)$. Then there exists an ultrametric $\tilde{\rho}$ defined on all $X$ such that for every $x,y\in X$ we have $d(x,y)\leq\tilde{\rho}(x,y)$, and if $x\in X$ and $y\in Y$ then $\tilde{\rho}(x,y)\leq6\alpha\cdot d(x,y)$.
\end{theorem}

Using Theorems \ref{thm:RamseyUM} and \ref{thm:ExtendingUM}, we generalize Theorem \ref{thm:RamseyUM} into an analogous version of Theorem \ref{thm:RamseyTree}. Recall that for an HST $T$ with labels $\ell$ we denote $\rho_T(u,v)=\ell(LCA_T(u,v))$.

\begin{theorem} \label{thm:RamseyNotSubTree}
Let $G=(V,E)$ be an undirected weighted graph, not necessarily connected. Let $A\subseteq V$ be a subset of size $t$, and let $k\geq1$ be some parameter. There exists a subset $S\subseteq A$ of size at least $t^{1-\frac{1}{k}}$, an HST $T_1=(V_1,E_1)$ with leaves $V\subseteq V_1$, and a tree $T_2=(V,E_2)$ (not necessarily $E_2\subseteq E$), such that
\begin{enumerate}
    \item $d_G(u,v)\leq\rho_{T_1}(u,v)$ and $d_G(u,v)\leq d_{T_2}(u,v)$ for every $u,v\in V$, and
    \item $\rho_{T_1}(u,v)\leq12ek\cdot d_G(u,v)$ and $d_{T_2}(u,v)\leq(32ek+1)d_G(u,v)$ for every $u\in S$ and $v\in V$.
\end{enumerate}
\end{theorem}

\begin{proof}

Consider the $t$-point metric space $(A,d_G)$ (i.e., the points are the vertices of the subset $A$, and the distances are the shortest path distances in the graph $G$). By Theorem \ref{thm:RamseyUM}, there is a subset $S\subseteq A$ of size $t^{1-\frac{1}{k}}$ that can be embedded into an ultrametric with distortion $2ek$. For simplicity, we identify the image of this embedding with the vertices of $S$, and conclude that there is an ultrametric $\rho$ on $S$ such that for every $u,v\in S$,
$d_G(u,v)\leq\rho(u,v)\leq2ek\cdot d_G(u,v)$.
Now, using Theorem \ref{thm:ExtendingUM}, we obtain an ultrametric $\tilde{\rho}$ on $V$, such that
\begin{enumerate}
    \item $d_G(u,v)\leq\tilde{\rho}(u,v)$ for every $u,v\in V$, and
    \item $\tilde{\rho}(u,v)\leq6\cdot 2ek\cdot d_G(u,v)=12ek\cdot d_G(u,v)$ for every $u\in S$ and $v\in V$.
\end{enumerate}

By Fact \ref{fact:UM<->HST}, the ultrametric $\tilde{\rho}$ can be converted into an HST $T_1$ that satisfies the assertion of the theorem.
To construct the tree $T_2$, we recall that $\rho$ is an ultrametric over $S\subseteq A$ with $d_G(u,v)\leq\rho(u,v)\leq2ek\cdot d_G(u,v)$ for every $u,v\in S$. We use Lemma \ref{lemma:FromUMToTree} on $\rho$, to obtain a tree $T'_2=(S,E'_2)$, not necessarily $E'_2\subseteq E$, such that for every $u,v\in S$,
\[d_G(u,v)\leq d_{T'_2}(u,v)\leq8\cdot2ek\cdot d_G(u,v)=16ek\cdot d_G(u,v)~.\]

Now, for every $v\in V\setminus S$, we add an edge between $v$ to its closest vertex in $S$, denoted by $s_v$, with weight $d_G(v,s_v)$. Let $T_2$ be the resulting tree. Note that
\begin{enumerate}
    \item For every $u,v\in V$, we have
    \begin{eqnarray*}
        d_G(u,v)&\leq&d_G(u,s_u)+d_G(s_u,s_v)+d_G(s_v,v)\\
        &\leq&d_{T_2}(u,s_u)+d_{T_2}(s_u,s_v)+d_{T_2}(s_v,v)=d_{T_2}(u,v)~,
    \end{eqnarray*}
    and
    \item For $u\in S$ and $v\in V$, we have
    \begin{eqnarray*}
        d_{T_2}(u,v)=d_{T_2}(u,s_v)+d_G(s_v,v)
        &\leq&16ek\cdot d_G(u,s_v)+d_G(s_v,v)\\
        &\leq&16ek\cdot[d_G(u,v)+d_G(v,s_v)]+d_G(s_v,v)\\
        &\leq&16ek\cdot[d_G(u,v)+d_G(u,v)]+d_G(u,v)\\
        &=&(32ek+1)d_G(u,v)~,
    \end{eqnarray*}
\end{enumerate}
as desired.

\end{proof}

\section{Tree Covers for Separable Graphs} \label{sec:TreeCoversForSeparableGraphs}

\subsection{Pairwise Covers}

Our first step towards constructing tree covers for vertex-separable and path-separable graphs is to construct a \textit{pairwise} tree cover (see Definition \ref{def:PairwiseForestCover}), for a designated set $\mathcal{P}_A$ of pairs, that is defined by a vertex subset $A$. Specifically, $\mathcal{P}_A$ is the set of all vertex pairs $(u,v)\in V^2$, such that some shortest $u$-$v$ path intersects the subset $A$. In the next lemma, we formalize our result about pairwise tree covers.

\begin{lemma} \label{lemma:PairwiseSpanTreeCover}
Let $G=(V,E)$ be an undirected weighted graph on $n$ vertices, and fix some $A\subseteq V$ with size $|A|=t$. Define the set $\mathcal{P}_A\subseteq V^2$ as the set of all pairs $(u,v)$ such that there is a shortest $u$-$v$ path that contains a vertex of $A$. Given an integer parameter $k\geq1$, the graph $G$ admits the following $\mathcal{P}_A$-pairwise covers.

\begin{enumerate}
    \item A $\mathcal{P}_A$-pairwise spanning forest cover $\mathcal{F}_A$ with stretch $\hat{c}\cdot k\log\log(t+4)$, where $\hat{c}$ is the constant from Theorem \ref{thm:RamseyTree}.

    \item A $\mathcal{P}_A$-pairwise (not necessarily spanning) tree cover $\mathcal{T}_A$ with stretch $32ek+1$.

    \item A $\mathcal{P}_A$-pairwise HST cover $\mathcal{H}_A$ with stretch $12ek$.
\end{enumerate}

Each of these covers has size $O\left(k\cdot t^{\frac{1}{k}}\right)$, and their maximum covering is
\[\max\left\{\sum_{F\in\mathcal{F}_A}|V(F)|,\sum_{T\in\mathcal{T}_A}|V(T)|,\sum_{T\in\mathcal{H}_A}|L(T)|\right\}=O\left(k(n-t)t^{\frac{1}{k}}+t^{1+\frac{1}{k}}\right)=O\left(k\cdot nt^\frac{1}{k}\right)~.\]

\end{lemma}

\begin{proof}

Our construction of the pairwise covers $\mathcal{F}_A,\mathcal{T}_A,\mathcal{H}_A$, with respect to a subset $A\subseteq V$, is recursive over $n,t$, while we fix the parameter $k$. We call the subset $A$ the \textit{demand} subset. Along with the recursive construction of $\mathcal{F}_A,\mathcal{T}_A,\mathcal{H}_A$, we will recursively define upper bounds $\Sigma(n,t),\alpha(t),q(t)$ for their covering, stretch and size, respectively, for an $n$-vertex graph $G$ and a demand subset of size $|A|=t$ (the upper bounds $\alpha(t)$ and $q(t)$ will hold for any number of vertices $n$, as we prove in the sequel). Our upper bounds $\Sigma(n,t),\alpha(t)$ and $q(t)$ will be monotonous non-decreasing functions. In our recursive construction, we consider the case where the graph $G$ is not necessarily connected.

For the base case, if $t=0$, thus $A=\mathcal{P}_A=\emptyset$, and $\mathcal{F}_A=\mathcal{T}_A=\mathcal{H}_A=\emptyset$ are trivial covers. Thus, we set the upper bounds $\alpha(0)=1$, $\Sigma(n,0)=0$ (for every $n$) and $q(1)=0$.

Given $t>0$, an $n$-vertex graph $G=(V,E)$ and a subset $A\subseteq V$ of size $|A|=t$, let $S\subseteq A$ be a subset of size at least $t^{1-\frac{1}{k}}$ and let $F$ be a spanning forest of $G$, such that $S,F$ satisfy the assertion of Lemma \ref{lemma:RamseyForest}. Alternatively, to form $\mathcal{T}_A$ and $\mathcal{H}_A$, we use Theorem \ref{thm:RamseyNotSubTree} to obtain the subset $S\subseteq A$ and corresponding tree $T$ and HST $T^{HST}$. Consider the graph $G'=G[V\setminus S]$, which is obtained by removing the vertices of $S$ from the graph $G$, together with their adjacent edges. Define also $A'=A\setminus S$. Denote the number of vertices in $G'$ by $n'$ and denote $|A'|=t'$. Then, we have
\begin{equation} \label{eq:NumberOfVertices}
    n'=n-|S|\leq n-t^{1-\frac{1}{k}},\;\;\;t'=t-|S|\leq t-t^{1-\frac{1}{k}}~.
\end{equation}

We recursively construct a $\mathcal{P}_{A'}$-pairwise forest cover $\mathcal{F}_{A'}$ (respectively, tree cover $\mathcal{T}_{A'}$ and HST cover $\mathcal{H}_{A'}$) for $G'$ with size at most $q(t')$, stretch at most $\alpha(t')$ and covering at most $\Sigma(n',t')$. Here, $\mathcal{P}_{A'}$ is the set of all $(u,v)\in V^2$ such that some shortest $u$-$v$ path in $G'$ intersects $A'$. We define
\begin{equation} \label{eq:recTreeCover}
\begin{split}
    \mathcal{F}_A&=\mathcal{F}_{A'}\cup\{F\}~,\\
    \mathcal{T}_A&=\mathcal{T}_{A'}\cup\{T\}~,\\
    \mathcal{H}_A&=\mathcal{H}_{A'}\cup\{T^{HST}\}~.
\end{split}
\end{equation}

We start by analyzing the stretch of these new pairwise covers. Fix a pair $(u,v)\in\mathcal{P}_A$, and note that in particular, $u,v$ are in the same connected component of $G$. We consider two cases. In the first case, let us assume that there is a shortest path between $u,v$ in $G$, which passes through a vertex $s\in S$. Denote by $Q^F_{u,s}$ and $Q^F_{s,v}$ the unique paths in the spanning forest $F$ between $u,s$ and between $s,v$, respectively. Similarly, let $Q^T_{u,s}$ and $Q^T_{s,v}$ be the unique paths in $T$ between $u,s$ and $s,v$, respectively. For each $\mu\in\{F,T\}$, the concatenation $Q^\mu_{u,s}\circ Q^\mu_{s,v}$ is a (not necessarily simple) path in $\mu$ between $u,v$. Thus, it contains $P^\mu$ -- the unique path in $\mu$ between $u,v$. By the properties of the set $S$, the forest $F$ from Lemma \ref{lemma:RamseyForest}, and the tree $T$ from Theorem \ref{thm:RamseyNotSubTree}, we get
\begin{equation} \label{eq:StrongerObservation}
\begin{split}
    d_\mu(u,v)&=w(P^\mu)
\leq w(Q^\mu_{u,s})+w(Q^\mu_{s,v})\\
&=d_\mu(u,s)+d_\mu(v,s)\\
&\leq\alpha_\mu\cdot d_G(u,s)+\alpha_\mu\cdot d_G(v,s)\\
&=\alpha_\mu\cdot d_G(u,v)~.
\end{split}
\end{equation}
where $\alpha_\mu$ is $\hat{c}\cdot k\log\log(t+4)$ in the case $\mu=F$, and $32ek+1$ in the case $\mu=T$.

Regarding $\mathcal{H}_A$, recall that $\rho_{T^{HST}}(u,v)$ denotes the label of the LCA of $u,v$ in the HST $T^{HST}$, and that it is an ultrametric. In this case we obtain, by Theorem \ref{thm:RamseyNotSubTree},
\begin{equation} \label{eq:StrongerObservationHST}
\begin{split}
    \rho_{T^{HST}}(u,v)&\leq\max\{\rho_{T^{HST}}(u,s),\rho_{T^{HST}}(s,v)\}\\
    &\leq12ek\cdot\max\{d_G(u,s),d_G(s,v)\}\\
    &\leq12ek\cdot d_G(u,v)~.
\end{split}
\end{equation}

We next consider the second case, where there is no shortest path between $u,v$ in $G$ that passes through $S$. Then, any such shortest path is fully contained in the graph $G'$, i.e., $d_{G'}(u,v)=d_G(u,v)$. Also, by the assumption that $(u,v)\in\mathcal{P}_A$, we conclude that there is a shortest $u$-$v$ path that contains a vertex of $A\setminus S=A'$. Hence, there is a forest $F'\in\mathcal{F}_{A'}$, a tree $T'\in\mathcal{T}_{A'}$ and an HST $T''\in\mathcal{H}_{A'}$ such that
\[\max\{d_{F'}(u,v),d_{T'}(u,v),\rho_{T''}(u,v)\}\leq\alpha(t')\cdot d_{G'}(u,v)\stackrel{(\ref{eq:NumberOfVertices})}{\leq}\alpha\left(t-t^{1-\frac{1}{k}}\right)\cdot d_G(u,v)~.\]
Thus, the following serves as an upper bound for the stretch of the forest cover $\mathcal{F}_A$.
\begin{equation} \label{eq:StretchRec}
    \alpha(t)=\max\{\hat{c}\cdot k\log\log(t+4),\alpha\left(t-t^{1-\frac{1}{k}}\right)\}~.
\end{equation}
A simple inductive proof over $t$ shows that $\alpha(t)\leq\hat{c}\cdot k\log\log(t+4)$: for $t=0$, we already saw that $\alpha(0)=1\leq\hat{c}\cdot k\log\log(t+4)$. For $t>0$, our induction hypothesis implies that $\alpha\left(t-t^{1-\frac{1}{k}}\right)\leq\hat{c}\cdot k\log\log(t+4)$, and by Inequality (\ref{eq:StretchRec}),
$\alpha(t)\leq\max\{\hat{c}\cdot k\log\log(t+4),\alpha\left(t-t^{1-\frac{1}{k}}\right)\}\leq\hat{c}\cdot k\log\log(t+4)$,
as desired. For the pairwise tree cover $\mathcal{T}_A$ and HST cover $\mathcal{H}_A$, a similar recursive relation of $\alpha(t)=\max\left\{32ek+1,\alpha\left(t-t^{1-\frac{1}{k}}\right)\right\}$ and $\alpha(t)=\max\left\{12ek,\alpha\left(t-t^{1-\frac{1}{k}}\right)\right\}$, respectively, gives $\alpha(t)\leq32ek+1$ and $\alpha(t)\leq12ek$, respectively, in the same way.

Next, we bound the size of the covers $\mathcal{F}_A,\mathcal{T}_A,\mathcal{H}_A$. Recall their definition in (\ref{eq:recTreeCover}), which implies that their sizes are $|\mathcal{F}_{A'}|+1,|\mathcal{T}_{A'}|+1,|\mathcal{H}_{A'}|+1$, respectively. In addition, note that $q(t')\leq q(t-t^{1-\frac{1}{k}})$ (by Inequality (\ref{eq:NumberOfVertices})). Thus, we recursively define the upper bound $q(t)$ as
\begin{equation} \label{eq:QueryTimeRec}
    q(t)=1+q\left(t-t^{1-\frac{1}{k}}\right)~.
\end{equation}

To solve this recursive relation, let us denote $x_k(t)=1-t^{-\frac{1}{k}}$ and $f_k(t)=\left\lceil t^{\frac{1}{k}}\right\rceil$. Suppose that we perform $f_k(t)$ recursive calls of our construction on the input graph $G$. After each recursive call, the size of the demand subset $A$ decreases from some $t_0$ (at the beginning $t_0=t$) to $x_k(t_0)t_0\leq x_k(t)t_0$. Thus, after $f_k(t)$ calls, this size is at most
$(x_k(t))^{f_k(t)}\cdot t\leq\left(1-\frac{1}{t^{\frac{1}{k}}}\right)^{t^{\frac{1}{k}}}\cdot t\leq\frac{t}{e}$.
By the same argument, after another $f_k\left(\frac{t}{e}\right)$ recursive calls, our demand subset has at most $\frac{t}{e^2}$ vertices, and after another $f_k\left(\frac{t}{e^2}\right)$ calls, at most $\frac{t}{e^3}$ vertices. We proceed with this process until the demand subset is empty, and then our recursive construction terminates. The number of recursive calls required for this purpose is at most 
\begin{equation} \label{eq:RecDepth}
    \sum_{i=0}^{\left\lceil\ln t\right\rceil}f_k\left(\frac{t}{e^i}\right)\leq\left\lceil\ln t\right\rceil+\sum_{i=0}^\infty\left(\frac{t}{e^i}\right)^{\frac{1}{k}}=\left\lceil\ln t\right\rceil+\frac{t^{\frac{1}{k}}}{1-e^{-\frac{1}{k}}}\leq\left\lceil\ln t\right\rceil+2k\cdot t^{\frac{1}{k}}=O(k\cdot t^{\frac{1}{k}})~,
\end{equation}
where we used the fact that $k\cdot t^{\frac{1}{k}}\geq\ln t$ for every $k$. Thus, $q(t)=O(k\cdot t^{\frac{1}{k}})$. We conclude that the number of recursive calls, which is also the size of $\mathcal{F}_A,\mathcal{T}_A,\mathcal{H}_A$, is at most $q(t)=O(k\cdot t^{\frac{1}{k}})$.

Next, we bound the covering of the pairwise covers $\mathcal{F}_A,\mathcal{T}_A,\mathcal{H}_A$. The forest $F$, the tree $T$ and the HST $T^{HST}$ that were added to the covers in Equation (\ref{eq:recTreeCover}) contain at most $n$ vertices (or in the case of $T^{HST}$, at most $n$ leaves). Thus, by (\ref{eq:NumberOfVertices}), we recursively define
\begin{equation} \label{eq:SizeRec}
    \Sigma(n,t)=n+\Sigma\left(n-t^{1-\frac{1}{k}},t-t^{1-\frac{1}{k}}\right)~.
\end{equation}

Using this equation, we prove by induction over $t$ that $\Sigma(n,t)\leq(n-t)\cdot q(t)+t^{1+\frac{1}{k}}$. Recall that by (\ref{eq:RecDepth}), we have $q(t)=O(k\cdot t^{\frac{1}{k}})$. For $t=0$, the equation holds since $q(0)=\Sigma(n,0)=0$.
For $t>0$, recall that by Equation (\ref{eq:QueryTimeRec}), $q(t)-1=q(t-t^{1-\frac{1}{k}})$. We use this fact, Equation (\ref{eq:SizeRec}), and the induction hypothesis, and we obtain
\begin{eqnarray*}
    \Sigma(n,t)&=&n+\Sigma\left(n-t^{1-\frac{1}{k}},t-t^{1-\frac{1}{k}}\right)\\
    &\leq&n+\left(n-t^{1-\frac{1}{k}}-(t-t^{1-\frac{1}{k}})\right)\cdot q(t-t^{1-\frac{1}{k}})+(t-t^{1-\frac{1}{k}})^{1+\frac{1}{k}}\\
    &=&n+(n-t)\cdot(q(t)-1)+(x_k(t)t)^{1+\frac{1}{k}}\\
    &\leq&n+(n-t)\cdot q(t)-n+t+x_k(t)\cdot t^{1+\frac{1}{k}}\\
    &=&n+(n-t)\cdot q(t)-n+t+t^{1+\frac{1}{k}}-t
    =(n-t)\cdot q(t)+t^{1+\frac{1}{k}}~.\\
\end{eqnarray*}
Hence, our pairwise covers has covering at most
$(n-t)\cdot q(t)+t^{1+\frac{1}{k}}=O\left(k(n-t)t^{\frac{1}{k}}+t^{1+\frac{1}{k}}\right)$.

\end{proof}

In case that $A=V$, we have the following corollary for spanning tree covers (here $t=n$).

\begin{corollary} \label{cor:BetterAverageOverlap}
Let $G=(V,E)$ be an undirected weighted graph on $n$ vertices. Given an integer parameter $k\geq1$, there exists a spanning forest cover $\mathcal{F}$ with $O(kn^{\frac{1}{k}})$ forests, stretch $\hat{c}\cdot k\log\log n$ and average overlap $O(n^{\frac{1}{k}})$.

By Observation \ref{obs:FromForestCoverToTreeCover}, there is a tree cover $\mathcal{T}$ with stretch $\hat{c}\cdot k\log\log n$ and average overlap $O(n^{\frac{1}{k}})$.
\end{corollary}

We note that it was known \cite{ACEFN20} that for any connected graph $G$ there exists a spanning tree cover $\mathcal{T}$ with $O(kn^{\frac{1}{k}})$ trees and stretch $O(k\log\log n)$ (this tree cover can be obtained from Lemma \ref{lemma:PairwiseSpanTreeCover} and Observation \ref{obs:FromForestCoverToTreeCover}). However, to the best of our knowledge, there was no known construction of a spanning tree cover for general graphs with stretch $O(k\log\log n)$ and \textit{average overlap} $O(n^{\frac{1}{k}})$.

\subsection{Tree Covers for Vertex-Separable Graphs} \label{sec:VertexSeparableTreeCovers}

In this section we use Lemma \ref{lemma:PairwiseSpanTreeCover}, to build tree covers for vertex-separable graphs (see Definition \ref{def:RecursiveSeparators}).

\begin{theorem} \label{thm:SmallTreewidthSpanTreeCover}
Let $G=(V,E)$ be an $n$-vertex undirected weighted $s$-vertex-separable graph, for some non-decreasing function $s=s(\theta)$. Let $k\geq1$ be an integer parameter, and denote $S(n,k)=\sum_{i=0}^{\log n}s\left(\frac{n}{2^i}\right)^{\frac{1}{k}}$. There exists a spanning forest cover for $G$ with stretch at most $\hat{c}\cdot k\log\log s(n)$ and size $O\left(k\cdot S(n,k)\right)$.
By Observation \ref{obs:FromForestCoverToTreeCover}, if $G$ is connected, there exists a spanning tree cover with the same stretch and size.

Moreover, there is a (not necessarily spanning) tree cover $\mathcal{T}$ and an HST cover $\mathcal{H}$ with the same size $O\left(k\cdot S(n,k)\right)$ and stretch $32ek+1$ and $12ek$, respectively.
\end{theorem}

\begin{proof}

We start by proving the existence of the desired spanning forest cover.

Let $A\subseteq V$ be a separator of the given graph $G=(V,E)$, of size $|A|\leq s(n)$. That is, if $C_1,C_2,...,C_\ell$ are the connected components of the graph $G'=G[V\setminus A]$, and $n_i$ is the number of vertices in $C_i$, then for every $i$, $n_i\leq\frac{n}{2}$.

Given the graph $G$ and its separator $A$, we construct the $\mathcal{P}_A$-pairwise forest cover $\mathcal{F}_A$ from Item (1) in Lemma \ref{lemma:PairwiseSpanTreeCover}. Recall that $\mathcal{F}_A$ has $O\left(k\cdot s(n)^{\frac{1}{k}}\right)$ forests, stretch $\hat{c}\cdot k\log\log s(n)$ and covering $O(kn\cdot s(n)^{\frac{1}{k}})$. Note that 
any component $C_i$ is $s$-vertex-separable. Recursively, we construct a spanning forest cover $\mathcal{F}_i$ of size $O(k\cdot S(\frac{n}{2},k))$ and stretch at most $\hat{c}\cdot k\log\log s(n)$, for every such component $C_i$. For the base case, if $n_i=1$, we let $\mathcal{F}_i$ be the trivial forest cover, that has one forest, stretch $1$ and covering $1$. By Lemma \ref{lemma:UnifyDisjointForestCovers}, there is a forest cover $\mathcal{F}'$ for the graph $G'$ with $O(k\cdot S(\frac{n}{2},k))$ forests and stretch at most $\hat{c}\cdot k\log\log s(n)$.

We define our new spanning forest cover as
\begin{equation} \label{eq:nonFullDef}
  \mathcal{F}=\mathcal{F}_A\cup\mathcal{F}'~.
\end{equation}

Note that the number of forests in $\mathcal{F}$ is proportional to
\begin{eqnarray*}
    k\cdot s(n)^\frac{1}{k}+k\cdot S\left(\frac{n}{2},k\right)&=&k\cdot s(n)^\frac{1}{k}+k\sum_{i=0}^{\log n-1}s\left(\frac{n}{2^{i+1}}\right)^{\frac{1}{k}}\\
    &=&k\cdot s(n)^\frac{1}{k}+k\sum_{i=1}^{\log n}s\left(\frac{n}{2^{i}}\right)^{\frac{1}{k}}\\
    &=&k\sum_{i=0}^{\log n}s\left(\frac{n}{2^{i}}\right)^{\frac{1}{k}}=k\cdot S(n,k)~.
\end{eqnarray*}

We argue that the stretch of $\mathcal{F}$ is at most $\hat{c}\cdot k\log\log s(n)$. Indeed, if a given pair $(u,v)\in V^2$ is in $\mathcal{P}_A$, that is, there is a shortest $u$-$v$ path that passes through a vertex of $A$, then by Lemma \ref{lemma:PairwiseSpanTreeCover}, there is a forest $F\in\mathcal{F}_A$ such that
\[d_F(u,v)\leq\hat{c}\cdot k\log\log s(n)\cdot d_G(u,v)~.\]
Otherwise, if any shortest $u$-$v$ path does not contain a vertex of $A$, then in particular, 
$d_G(u,v)=d_{G'}(u,v)$, for the graph $G'=G[V\setminus A]$. Thus, the spanning forest cover $\mathcal{F}'$ contains a forest $F$ with
\[d_{F}(u,v)
\leq\hat{c}\cdot k\log\log s(n)\cdot d_G(u,v)~.\]
Figure \ref{fig:Separatorspanning tree cover} depicts the different scenarios for the pair $(u,v)$, that were discussed above.

\begin{center}
\begin{figure}[!ht]
    \centering
    \includegraphics[width=9cm, height=6cm]{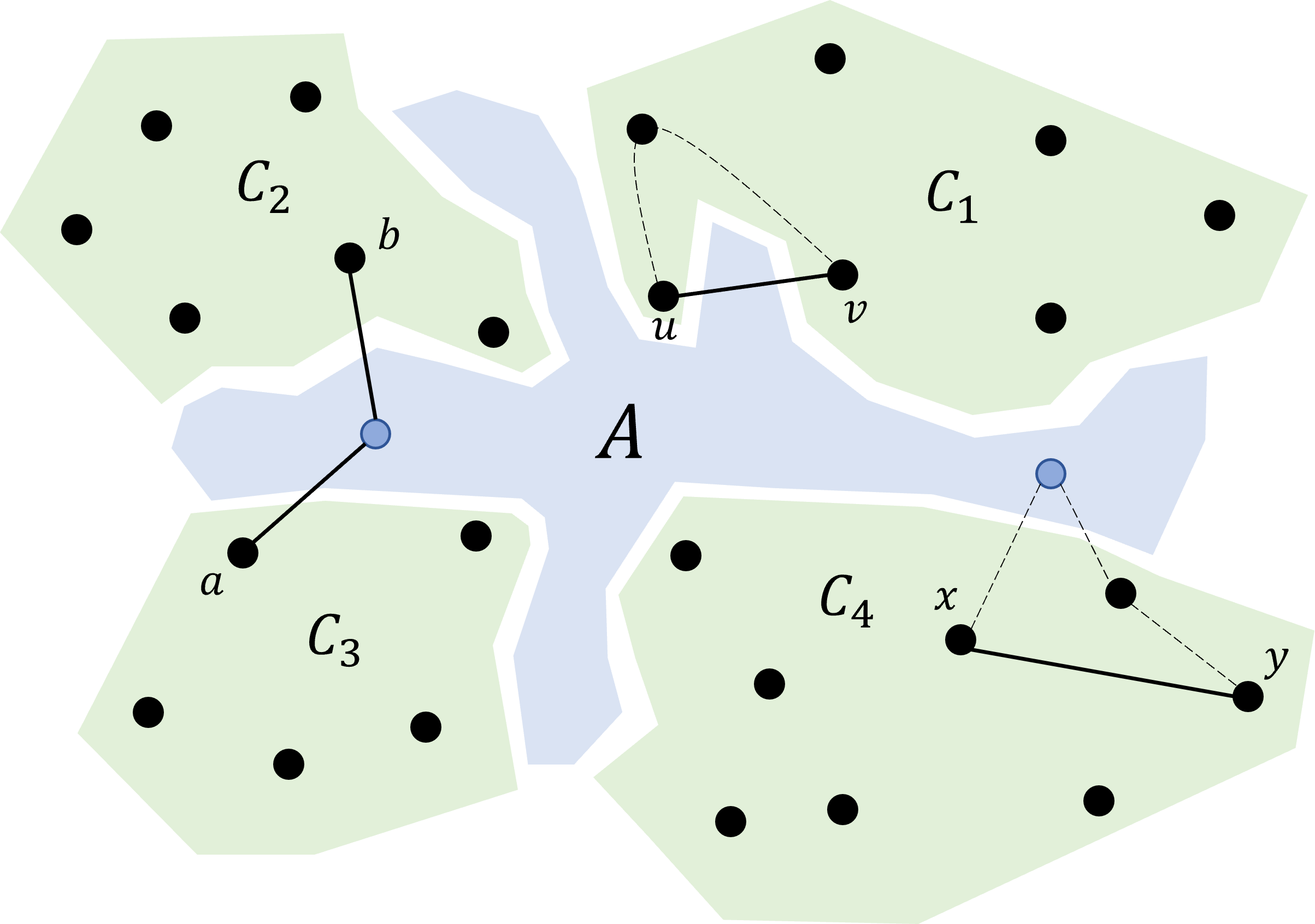}
    \caption{The separator $A$ divides the graph $G$ into connected components $C_i$, each containing at most half of the vertices in $G$. Then, the spanning forest cover $\mathcal{F}_A$ from Lemma \ref{lemma:PairwiseSpanTreeCover} is constructed for the demand set $A$. In addition, a recursive construction of a spanning forest cover $\mathcal{F}_i$ is performed for each component $C_i$. For every two vertices, we consider the shortest path between them in the forest of $\mathcal{F}_A$ that provides the lowest stretch between them. In case that the two vertices are in the same component $C_i$, we also consider the shortest path in the forest of $\mathcal{F}_i$ that provides the lowest stretch. The weights of these paths, up to stretch $O(k\log\log s(n))$, are the weights of the shortest path that intersects $A$, and the one that does not intersect $A$, respectively. In the figure, the shortest $a$-$b$ path intersects $A$, and $a,b$ are not in the same component. Hence, the spanning forest cover $\mathcal{F}_A$ contains a forest that provides stretch $O(k\log\log s(n))$ for $a,b$. For $x,y$, their shortest path does not intersect $A$, and they are in the same component $C_4$. Thus, there is a forest in $\mathcal{F}_4$ that provides stretch $O(k\log\log s(n))$ for $x,y$. Note that the case of $u,v$ might also occur. They are in the same component $C_1$, but their shortest path intersects $A$ (for this reason, the forest that provides low stretch for $u,v$ is in $\mathcal{F}_A$, even though the vertices are in the same component).}
    \label{fig:Separatorspanning tree cover}
\end{figure}
\end{center}

Next, we describe how to construct the (not necessarily spanning) tree cover $\mathcal{T}$ and the HST cover $\mathcal{H}$. This construction is very similar to that of the spanning forest cover $\mathcal{F}$, with the following differences.

For the same separator $A$ of $G$, with size at most $s(n)$, let $\mathcal{T}_A,\mathcal{H}_A$ be the $\mathcal{P}_A$-pairwise tree cover and $\mathcal{P}_A$-pairwise HST cover from Lemma \ref{lemma:PairwiseSpanTreeCover}, respectively. Then, we define the tree cover $\mathcal{T}$ using (\ref{eq:nonFullDef}), but with $\mathcal{T}_A$ instead of $\mathcal{F}_A$. Here, we form the forest cover $\mathcal{F}'$ in the same way from the recursive tree covers of the connected components of $G'=G[V\setminus A]$, and then use Observation \ref{obs:FromForestCoverToTreeCover} to obtain a tree cover $\mathcal{T}'$ with the same size and stretch as $\mathcal{F}'$. The union $\mathcal{T}=\mathcal{T}_A\cup\mathcal{T}'$ defines the desired (not necessarily spanning) tree cover. 

The construction of the HST cover with small number of HSTs follows the same approach as in Lemma \ref{lemma:UnifyDisjointForestCovers}. We explicitly describe this construction below.

For any connected component $C_i$ of $G'=G[V\setminus A]$, assume that we have an HST cover $\mathcal{T}_i$ with at most $J=O(k\cdot S(\frac{n}{2},k))$ trees, and enumerate them by $T_i^1,T_i^2,T_i^3...$. Note that for each $j$, the HSTs $T_1^j,T_2^j,...$ have disjoint sets of leaves, as for every $i$, the set of leaves $L(T_i^j)$ is a subset of $C_i$. For any $j\geq1$, we unify the HSTs $T_1^j,T_2^j,...$ (the ones of them that exist) into a single HST $T^j$ as follows. We create a new root vertex $r_j$, and set it as a parent of the roots $r_1^j,r_2^j,...$ of $T_1^j,T_2^j,...$, respectively. The label of $r_j$ is chosen as $\infty$, or any number large enough such that the HST constraint is satisfied.
Then, $\{T^j\}_{j=1}^J\cup\mathcal{H}_A$ is an HST cover with $J+O(k\cdot s(n)^{\frac{1}{k}})=O(k\cdot S(n,k))$ HSTs. Similarly to the proof above, we can prove that the stretch of this HST cover is still $12ek$.
    
\end{proof}

\begin{corollary} \label{cor:TWN^delta}
Let $G=(V,E)$ be an $s$-vertex-separable graph with $n$ vertices, where $s=s(\theta)=\theta^{\delta}$, for some $\delta\in(0,1]$, and let $k\geq1$ be an integer parameter. Then, $G$ admits a spanning tree cover, a non-spanning tree cover, and an HST cover, with stretch at most $\hat{c}\cdot k\log\log s(n)$, $32ek+1$, and $12ek$, respectively, all with size $O\left(\frac{k^2}{\delta}n^{\frac{\delta}{k}}\right)=O\left(\frac{k^2\log n}{\log s(n)}\cdot s(n)^{\frac{1}{k}}\right)$.
\end{corollary}

\begin{proof}

Using the notations of Theorem \ref{thm:SmallTreewidthSpanTreeCover}, we have
\[S(n,k)=\sum_{i=0}^{\log n}s\left(\frac{n}{2^i}\right)^{\frac{1}{k}}<\sum_{i=0}^{\infty}\left(\frac{n}{2^i}\right)^{\frac{\delta}{k}}=n^{\frac{\delta}{k}}\cdot\frac{1}{1-2^{-\delta/k}}\leq\frac{2k}{\delta}n^{\frac{\delta}{k}}~,\]
where the last step follows from the fact that $2^{-x}\leq1-\frac{x}{2}$ for every $x\in[0,1]$. Hence, Theorem \ref{thm:SmallTreewidthSpanTreeCover} provides tree covers with the desired properties.
    
\end{proof}

Corollary \ref{cor:TWN^delta} improves upon the bound of \cite{ACEFN20} on the number of trees in spanning tree covers for $\theta^\delta$-vertex-separable graphs, for any constant $\delta<1$, whenever $k=o\left(\frac{\log n}{\log\log n}\right)$. 

In the next corollary we consider a wider family of graphs, that admit a \textit{flat} bound $t(n)$ on their treewidth. Such an $n$-vertex graph $G=(V,E)$ has the property that for any vertex subset $U\subseteq V$ (possibly with $|U|\ll n$), the induced sub-graph $G[U]$ of $U$ is guaranteed to have treewidth $t(n)$, and therefore a separator of size at most $t(n)+1$ (as opposed to $t(|U|)$ or $t(|U|)+1$, respectively). In this case, we apply Theorem \ref{thm:SmallTreewidthSpanTreeCover} with $s(\theta)\equiv t(n)$, and obtain the following corollary.

\begin{corollary} \label{cor:FlatTW}
Let $G=(V,E)$ be an $n$-vertex graph with treewidth $t=t(n)$, and let $k\geq1$ be an integer. $G$ admits a spanning tree cover, a non-spanning tree cover, and an HST cover, with stretch $\hat{c}\cdot k\log\log t$, $32ek+1$, and $12ek$, all with size $O(k\cdot t^{\frac{1}{k}}\cdot\log n)$.
\end{corollary}

For spanning tree covers, Corollary \ref{cor:FlatTW} improves upon the aforementioned bound of \cite{ACEFN20} for treewidth $t(n)<\frac{n}{(c\log n)^k}$, for a sufficiently large constant $c$.


\subsection{Low Stretch Spanning Tree Covers for Path-Separable Graphs} \label{sec:PathSeparableSpanSmallStretchTreeCovers}

In this section, we devise spanning tree covers with stretch $1+\epsilon$ (i.e., \textit{near-exact} spanning tree covers), for path-separable graphs (see Definition \ref{def:PathSeparable}). For this purpose, similarly to previous work by \cite{Tho04,Kle02,GKR04,BFN22}, we extract a small number of special vertices called \textit{landmarks} from each separator path, and grow shortest paths trees (SPTs) rooted at each of these landmarks. To form a tree cover out of all these SPTs, previous work had to settle for non-spanning tree covers \cite{BFN22}, or to use only one landmark on each separator path \cite{GKR04}, which results in stretch $3$. Instead, we describe a delicate way to \textit{glue} these SPTs together, and obtain a spanning tree cover with stretch $1+\epsilon$. Similarly to previous work \cite{Tho04}, to keep the number of landmarks controllable, we consider each time a different \textit{distance scale} $\Delta>0$, and provide a tree cover for vertex pairs $u,v\in V$ such that $d_G(u,v)\approx\Delta$.

We begin with the following useful lemma, that finds a small number of landmarks on a given path $Q$, that divides $Q$ to sub-paths of bounded weight.

\begin{lemma} \label{lemma:FindLandmarks}
    Let $Q$ be a path in an undirected weighted graph $G=(V,E)$, and denote its endpoints by $x,y$. For every positive parameter $\eta>0$, there is a sequence $x=q_1,q_2,...,q_t=y$ of vertices on $Q$, called $\boldsymbol{\eta}$-\textbf{landmarks}, that satisfies the following properties (here, $Q[a,b]$ denotes the sub-path of $Q$ between $a$ and $b$, for any two vertices $a,b$ on $Q$).
\begin{enumerate}
    \item For every $1\leq i<t$, the sub-path $Q[q_i,q_{i+1}]$ is a path of weight at most $\eta$, or a single edge.
    \item For any sub-path $Q'$ of $Q$, there are at most $\frac{2w(Q')}{\eta}+1$ $\eta$-landmarks on $Q'$.
\end{enumerate}
\end{lemma}

\begin{proof}

To construct the sequence $\{q_i\}$, let $q_1=x$, and iteratively for any \textbf{odd} $i\geq1$ define $q_{i+1}$ as the farthest vertex on $Q[q_i,y]$ from $q_i$, such that $w(Q[q_i,q_{i+1}])\leq\eta$ (possibly $q_{i+1}=q_i$). Then define $q_{i+2}$ as the successor of $q_{i+1}$ on $Q$ (towards $y$). Note that $w(Q[q_i,q_{i+2}])>\eta$, whenever $q_{i+2}$ exists. If $q_{i+2}$ does not exist, i.e., $q_{i+1}=y$, stop the process and denote $t=i+1$.

Clearly, Item (1) holds, as for odd $i$, the weight of $Q[q_i,q_{i+1}]$ is by definition at most $\eta$, and for even $i$, it is a single edge. For Item (2), let $Q'$ be a sub-path of $Q$, and let $\{q_i,q_{i+2},...,q_j\}$ be all the $\eta$-landmarks on $Q'$ with odd index. By our construction, the weight of the sub-path between any two consecutive landmarks in this set is \textit{larger} than $\eta$, and therefore there are at most $\frac{w(Q')}{\eta}$ landmarks in this set. It follows that the total number of $\eta$-landmarks on $Q'$, including odd and even indices, is at most $\frac{2w(Q')}{\eta}+1$.
    
\end{proof}

We next employ the landmarks of Lemma \ref{lemma:FindLandmarks} to build spanning tree covers with small number of trees and low stretch. For this end, we prove the following lemma, that constructs a separate forest cover for each scale $\Delta>0$ by growing $\Delta$-bounded radius trees around each landmark, joining some of them to the same forest, and continuing recursively.

\begin{lemma} \label{lemma:OneScaleSimpleTreeCover}
    Let $G=(V,E)$ be an undirected weighted (weakly) $p$-path-separable graph, and let $\Delta,\epsilon$ be positive parameters. There is a spanning forest cover $\mathcal{F}_\Delta$ of at most $\left(\left\lceil\frac{4}{\epsilon}\right\rceil+3\right)p\cdot\log n$ forests, such that for every $u,v\in V$ with $d_G(u,v)\leq\Delta$ there is a forest $F=F(u,v)\in\mathcal{F}_\Delta$ that satisfies
    \[d_F(u,v)\leq d_G(u,v)+\epsilon\Delta~.\]
\end{lemma}

\pagebreak
\begin{proof}

By Definition \ref{def:PathSeparable}, the graph $G$ admits a separator $S=P_0\cup P_1\cup\cdots\cup P_{\ell-1}$, where every $P_i$ is a union of shortest paths in the graph $H_i=G\setminus\bigcup_{i'<i}P_{i'}$. Each connected component $C_j$ of the graph $G\setminus S$ contains at most $\frac{n}{2}$ vertices. Recursively, we construct a spanning forest cover $\mathcal{F}^j_\Delta$ with at most $\left(\left\lceil\frac{4}{\epsilon}\right\rceil+3\right)p\log\frac{n}{2}$ forests for each component $C_j$, where $\mathcal{F}^j_\Delta$ satisfies the assertion of the lemma for $C_j$.

To obtain a small number of forests, we use Lemma \ref{lemma:UnifyDisjointForestCovers} on the spanning forest covers $\{\mathcal{F}^j_\Delta\}$, and obtain a forest cover $\tilde{\mathcal{F}}_\Delta$ for $G\setminus S$ with at most $\left(\left\lceil\frac{4}{\epsilon}\right\rceil+3\right)p\log\frac{n}{2}$ forests. The forest cover $\tilde{\mathcal{F}}_\Delta$ satisfies that whenever $d_{G\setminus S}(u,v)\leq\Delta$, there is a forest $F\in\tilde{\mathcal{F}}_\Delta$ with $d_F(u,v)\leq d_{G\setminus S}(u,v)+\epsilon\Delta$ (since in this case $u,v$ are in the same component $C_j$).


Next, fix $i\in\{0,1,...,\ell-1\}$, and let $Q$ be one of the paths that form $P_i$. Denote the endpoints of $Q$ by $x$ and $y$, and let $x=q_1,q_2,...,q_{t_Q}=y$ be $\eta$-\textit{landmarks} on $Q$, for $\eta=\epsilon\Delta$ (from Lemma \ref{lemma:FindLandmarks}):
\begin{enumerate}
    \item For every $s<t_Q$, the sub-path $Q[q_s,q_{s+1}]$ is a path of weight at most $\epsilon\Delta$, or a single edge.
    \item For any sub-path $Q'$ of $Q$, there are at most $\frac{2w(Q')}{\epsilon\Delta}+1$ landmarks in $Q'$.
\end{enumerate}

For every landmark $q_s$ on $Q$, let $T(Q,s)$ be the shortest path tree rooted at $q_s$ in the graph $H_i=G\setminus\bigcup_{i'<i}P_{i'}$, consisting only of vertices of distance at most $(1+\frac{\epsilon}{2})\Delta$ from $q_s$. That is, the vertex set $V(T(Q,s))$ of this tree is the \textit{cluster} 
\[C(Q,s)=\{v\in V(H_i)\;|\;d_{H_i}(q_s,v)\leq(1+\frac{\epsilon}{2})\Delta\}~.\]

We claim that for any two indices $s_1,s_2\leq t_Q$ such that $s_2>s_1+\frac{4}{\epsilon}+2$ the trees $T(Q,s_1),T(Q_,s_2)$ are vertex disjoint, i.e., $C(Q,s_1)\cap C(Q,s_2)=\emptyset$. Indeed, if there was $v\in C(Q,s_1)\cap C(Q,s_2)$, then
\[d_{H_i}(q_{s_1},q_{s_2})\leq d_{H_i}(q_{s_1},v)+d_{H_i}(v,q_{s_2})\leq2(1+\frac{\epsilon}{2})\Delta=(2+\epsilon)\Delta~,\]
and since $Q$ is a shortest path in $H_i$, we conclude that $w(Q[q_{s_1},q_{s_2}])\leq(2+\epsilon)\Delta$. By Item (2) above, it means that the number of landmarks in $Q[q_{s_1},q_{s_2}]$ is at most $\frac{2(2+\epsilon)}{\epsilon}+1=\frac{4}{\epsilon}+3$, i.e., $s_2\leq s_1+\frac{4}{\epsilon}+2$, in contradiction. 

This implies that for every $\sigma=1,2,...,\min\{\lceil\frac{4}{\epsilon}\rceil+3,t_Q\}$, the trees 
\[\left\{T\left(Q,s\right)\;\bigg|\;1\leq s\leq t_Q\;,\;s\equiv \sigma\mod{\left(\left\lceil\frac{4}{\epsilon}\right\rceil+3\right)}\right\}\]
are pairwise vertex disjoint. Thus, the union $F(Q,\sigma)$ of these trees is a spanning forest in $G$. See illustration in Figure \ref{fig:SinglePathForestCover}. We can finally define the resulting spanning forest $\mathcal{F}_\Delta$ as follows.
\[\mathcal{F}_\Delta=\left\{F(Q,\sigma)\;\bigg|\;Q\text{ in }S,\;\sigma=1,2,...,\min\left\{\left\lceil\frac{4}{\epsilon}\right\rceil+3,t_Q\right\}\right\}\cup\tilde{\mathcal{F}}_\Delta~,\]
where ``$Q$ in $S$" means that $Q$ is one of the paths that form $P_i$, for some $i\in\{0,1,2,...,\ell-1\}$. Recall that $S=P_0\cup P_1\cup\cdots\cup P_{\ell-1}$, and that $\tilde{\mathcal{F}}_\Delta$ is a collection of at most $\left(\left\lceil\frac{4}{\epsilon}\right\rceil+3\right)p\log\frac{n}{2}$ forests, sub-graphs of $G$, such that if $u,v\in V$ satisfy $d_{G\setminus S}(u,v)\leq\Delta$, then there is $F\in\tilde{\mathcal{F}}_\Delta$ with $d_F(u,v)\leq d_{G\setminus S}(u,v)+\epsilon\Delta$.

\begin{center}
\begin{figure}[!ht]
    \centering
    \includegraphics[width=12cm, height=6cm]{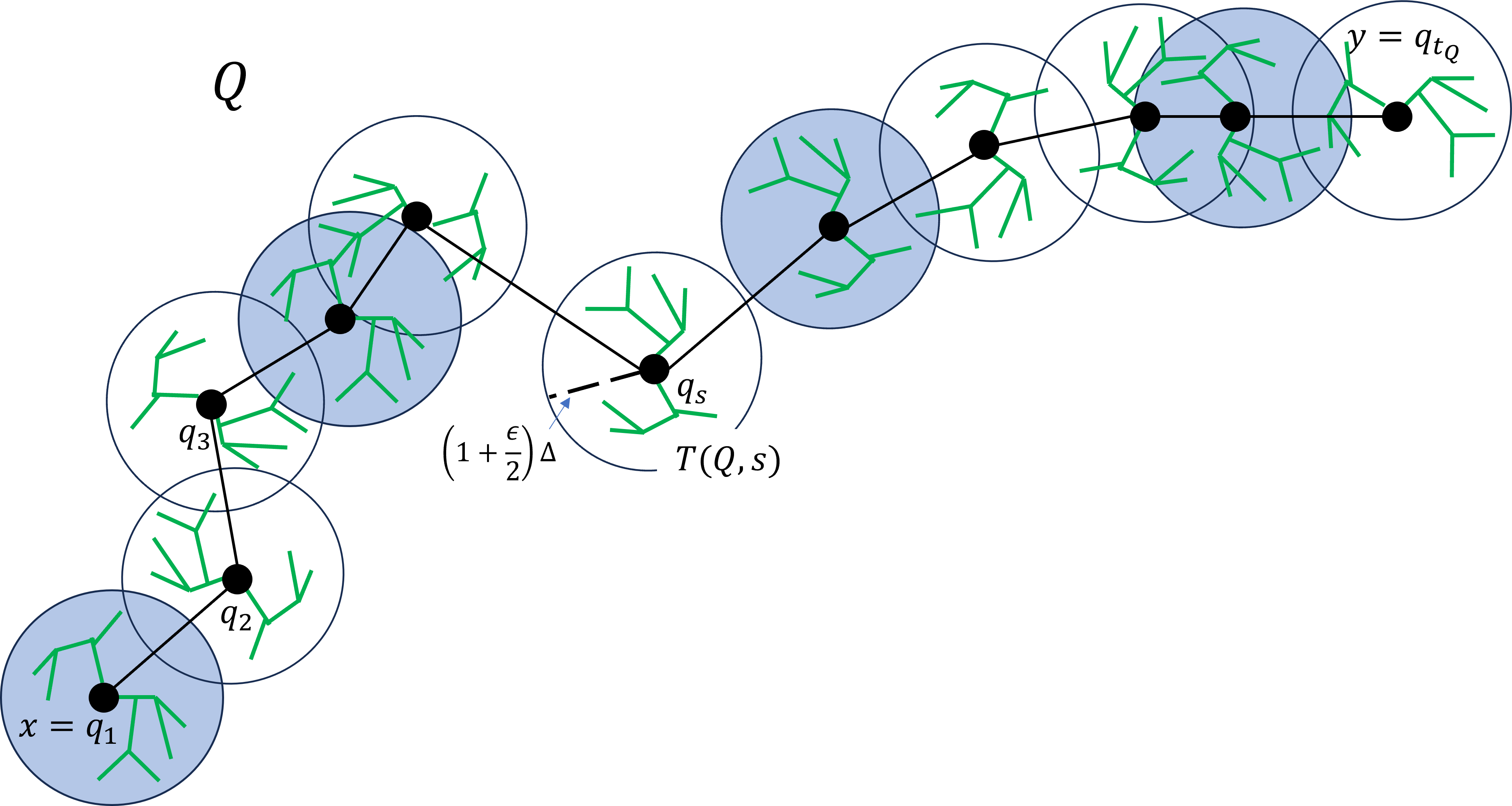}
    \caption{The path $Q$ with its $\epsilon\Delta$-landmarks $x=q_1,q_2,...,q_{t_Q}=y$ marked on it. Around every landmark $q_s$, we grow a shortest path tree $T(Q,s)$ in the graph $H_i$ for its cluster $C(Q,s)=\{v\in V(H_i)\;|\;d_{H_i}(q_s,v)\leq(1+\frac{\epsilon}{2})\Delta\}$. Clusters of remote landmarks, i.e., of which the indices differ by more than $\frac{4}{\epsilon}+2$, are disjoint. Starting with $C(Q,\sigma)$, for any index $1\leq\sigma\leq\frac{4}{\epsilon}+3$, we consider a sequence of clusters with remote landmarks (colored by blue) and denote the union of their trees $T(Q,s)$ by $F(Q,\sigma)$. Note that $F(Q,\sigma)$ is a forest. For the blue clusters in the picture, $\sigma=1$. Doing so for every $1\leq\sigma\leq\frac{4}{\epsilon}+3$ and for every separator path $Q$, we obtain a forest cover ($\mathcal{F}_\Delta\setminus\tilde{\mathcal{F}}_\Delta$) with $\left(\left\lceil\frac{4}{\epsilon}\right\rceil+3\right)p$ forests.}
    \label{fig:SinglePathForestCover}
\end{figure}
\end{center}

Following its definition, the number of forests in $\mathcal{F}_\Delta$ is at most 
\[\left(\left\lceil\frac{4}{\epsilon}\right\rceil+3\right)p+\left(\left\lceil\frac{4}{\epsilon}\right\rceil+3\right)p\log\frac{n}{2}=\left(\left\lceil\frac{4}{\epsilon}\right\rceil+3\right)p\log n,\] 
and they are all sub-graphs (i.e., spanning forests) of $G$. 

Fix some $u,v\in V$ with $d_G(u,v)\leq\Delta$, and let $P_{u,v}$ be a shortest $u$-$v$ path in $G$. If $P_{u,v}$ does not intersect the separator $S$, then $P_{u,v}$ is contained in $G\setminus S$. In this case, we know that there is $F\in\tilde{\mathcal{F}}_\Delta\subseteq\mathcal{F}_\Delta$ with 
\[d_F(u,v)\leq d_{G\setminus S}(u,v)+\epsilon\Delta=d_G(u,v)+\epsilon\Delta~.\]

Otherwise, the path $P_{u,v}$ intersects $S$. Let $i$ be the smallest index such that $P_{u,v}$ intersects a path $Q$ from $P_i\subseteq S$. Let $z$ be a vertex in which $P_{u,v}$ and $Q$ intersect, and let $q_s$ be the $\epsilon\Delta$-landmark on $Q$ that is the closest to $z$. If $z=q_s$, then clearly $d_{H_i}(z,q_s)=0$. If $z\neq q_s$, then $z$ must be either on the sub-path $Q[q_{s-1},q_s]$ or on the sub-path $Q[q_s,q_{s+1}]$. We assume the former, as the latter is analyzed analogously. Since $z\in Q[q_{s-1},q_s]$ and does not equal to $q_s$ or $q_{s-1}$, we conclude by Item (1) of Lemma \ref{lemma:FindLandmarks} that $w(Q[q_{s-1},q_s])\leq\epsilon\Delta$. The vertex $q_s$ is closer to $z$ than $q_{s-1}$, and therefore $d_{H_i}(z,q_s)= d_Q(z,q_s)\leq\frac{\epsilon}{2}\Delta$. This shows that in any case, $d_{H_i}(z,q_s)\leq\frac{\epsilon}{2}\Delta$.

Next, note that
\[d_{H_i}(u,q_s)\leq d_{H_i}(u,z)+d_{H_i}(z,q_s)\leq d_{H_i}(u,z)+\frac{\epsilon}{2}\Delta\leq(1+\frac{\epsilon}{2})\Delta~,\]
and similarly,
\[d_{H_i}(v,q_s)\leq d_{H_i}(v,z)+\frac{\epsilon}{2}\Delta\leq(1+\frac{\epsilon}{2})\Delta~.\]
That means that $u$ and $v$ are vertices in the tree $T(Q,s)$, and that
\[d_{T(Q,s)}(u,v)\leq d_{H_i}(u,q_s)+d_{H_i}(q_s,v)\leq d_{H_i}(u,z)+\frac{\epsilon}{2}\Delta+d_{H_i}(v,z)+\frac{\epsilon}{2}\Delta=d_G(u,v)+\epsilon\Delta~.\]
The last step here is due to the fact that $i$ is the smallest index such that $P_{u,v}$ intersects a separator path $Q$ from $P_i$. Therefore, $P_{u,v}$ is fully contained in $H_i=G\setminus\bigcup_{i'<i}P_{i'}$, and $d_{H_i}(u,z)+d_{H_i}(z,v)=d_G(u,z)+d_G(z,v)=d_G(u,v)$.

Lastly, let $\sigma$ be the unique index in $\{1,2,...,\min\{\lceil\frac{4}{\epsilon}\rceil+3,t_Q\}\}$ such that $\sigma\equiv s\mod{(\lceil\frac{4}{\epsilon}\rceil+3)}$. By the definition of the forest $F(Q,\sigma)\in\mathcal{F}_\Delta$, it contains $T(Q,s)$ as a sub-graph. We conclude that, as desired, there is a forest $F=F(Q,\sigma)\in\mathcal{F}_\Delta$ such that
$d_F(u,v)\leq d_G(u,v)+\epsilon\Delta~.$
    
\end{proof}

The union of the forest covers $\mathcal{F}_\Delta$, for every $\Delta\in\{2^i\}_{i=0}^{\lceil\log_2\Lambda\rceil}$, provides a spanning forest cover with stretch $1+\epsilon$ and $O(\frac{p}{\epsilon}\cdot\log n\cdot\log\Lambda)$ forests, where $\Lambda=\frac{\max_{u,v\in V}d_G(u,v)}{\min_{u\neq v\in V}d_G(u,v)}$ is the aspect ratio of $G$. To get rid of the dependence in $\Lambda$, and replace $\log\Lambda$ by $O(\log\frac{n}{\epsilon})$, we present a weight reduction in Appendix \ref{app:AspectRatioReduction}. As a result, we obtain the following theorem.

\begin{theorem} \label{thm:SpanTreeCoverForPathSeparators}
    Let $G=(V,E)$ be an $n$-vertex undirected weighted $p$-path-separable graph. For every $\epsilon\in(0,1]$, there is a spanning forest cover $\mathcal{F}$ for $G$ with stretch $1+\epsilon$ and size $O(p\cdot\frac{\log^2n}{\epsilon})$.

    If $G$ is connected, then by Observation \ref{obs:FromForestCoverToTreeCover}, there is also a spanning tree cover $\mathcal{T}$ for $G$ with stretch $1+\epsilon$ and $O(p\cdot\frac{\log^2n}{\epsilon})$ trees.
\end{theorem}

\subsection{Tree Covers with Fewer Trees for Path-Separable Graphs} \label{sec:PathSeparableLargeStretchTreeCovers}

In Section \ref{sec:PathSeparableSpanSmallStretchTreeCovers}, we constructed spanning tree covers for $p$-path-separable graphs, with a very low stretch of $1+\epsilon$, but with size that is at least linear in $p$. In this section, we devise (spanning and non-spanning) tree covers with larger stretch $\approx O(k)$ and with fewer trees $\approx O(p^\frac{1}{k})$ (for any choice of an integer parameter $k>0$). Note that for path-separable graphs, the separator size may be much higher than the number of paths of which it consists, and thus applying Theorem \ref{thm:SmallTreewidthSpanTreeCover} directly provides unsatisfactory bounds.

Note that in the proof of Lemma \ref{lemma:OneScaleSimpleTreeCover}, for approximating a shortest $u$-$v$ path that intersects a separator path $Q$, we used paths from $u$ and $v$ to a close $\eta$-landmark $q_s$ on $Q$. For this end, we had to include in the forest cover, for every landmark $q_s$, a shortest path tree rooted in $q_s$. This results in a large number of trees in the forest cover. In order to use fewer forests, we use our machinery from Lemma \ref{lemma:PairwiseSpanTreeCover}, albeit with the following observation, which was in fact proved in Inequalities (\ref{eq:StrongerObservation}) and (\ref{eq:StrongerObservationHST}) in the proof of Lemma \ref{lemma:PairwiseSpanTreeCover}.

\begin{observation} \label{obs:StrongPairwise}
    Given an undirected weighted graph $G=(V,E)$, a subset $A\subseteq V$ of size $t$, and an integer parameter $k\geq1$, the forest cover $\mathcal{F}_A$, the tree cover $\mathcal{T}_A$ and the HST cover $\mathcal{H}_A$ from Lemma \ref{lemma:PairwiseSpanTreeCover} satisfy, for every $u,v\in V$ and $s\in A$, 
    \begin{eqnarray*}
        \min_{F\in\mathcal{F}_A}d_F(u,v)&\leq&\hat{c}\cdot k\log\log(t+4)\cdot(d_G(u,s)+d_G(s,v))~,\\
        \min_{T\in\mathcal{T}_A}d_T(u,v)&\leq&(32ek+1)\cdot(d_G(u,s)+d_G(s,v))~,\\
        \min_{T\in\mathcal{H}_A}\rho_T(u,v)&\leq&12ek\cdot\max\{d_G(u,s),d_G(s,v)\}~.\\
    \end{eqnarray*}
\end{observation}

Recall that the covers $\mathcal{F}_A$, $\mathcal{T}_A$ and $\mathcal{H}_A$ have a relatively small size of $O(kt^{\frac{1}{k}})$. Thus, to utilize Observation \ref{obs:StrongPairwise}, we have to form the demand sets $A$ by using $\eta$-landmarks on different separator paths -- otherwise, the number of forests will be $\Omega(p)$. This makes the gluing of disjoint trees to a single forest (see the definition of $F(Q,\sigma)$ in the proof of Lemma \ref{lemma:OneScaleSimpleTreeCover}) more challenging. To encounter this problem, we first make an abstraction of it, using the following definition.

\begin{definition} \label{def:PathClustersGraph}
    Let $\Delta,\eta>0$ be positive real parameters. Suppose that the graph $G=(V,E)$ has a path separator $S=P_0\cup P_1\cup\cdots\cup P_{\ell-1}$ (see Definition \ref{def:PathSeparable}), and denote $H_i=G\setminus\bigcup_{j<i}P_j$ for every $i=0,1,...,\ell-1$. The \textbf{path-clusters graph} $L=(V_L,E_L)$ of $G$ with respect to $S,\Delta$ and $\eta$ is defined as follows. For every $i=0,1,...,\ell-1$ and a path $Q$ in $P_i$, let $q_1,q_2,...,q_{t_Q}$ be $\eta$-landmarks on $Q$ (see Lemma \ref{lemma:FindLandmarks}). For every $s=1,2,...,t_Q$, define the cluster $C(Q,s)=\{v\in V(H_i)\;|\;d_{H_i}(q_s,v)\leq\Delta+\frac{1}{2}\eta\}$. Then, the set of vertices $V_L$ is defined as the set of these clusters $C(Q,s)$, for every $Q$ and $1\leq s\leq t_Q$, and there is an edge between $C(Q,s)$ and $(Q',s')$ if and only if $C(Q,s)\cap C(Q',s')\neq\emptyset$.

    Given a cluster $C=C(Q,s)\in V_L$, we denote by $q(C)$ the landmark $q_s$ on $Q$, which is the \textbf{center} of this cluster. See Figure \ref{fig:PathClustersGraph} for an illustration.
\end{definition}

\begin{center}
\begin{figure}[!ht]
    \centering
    \includegraphics[width=16cm, height=8cm]{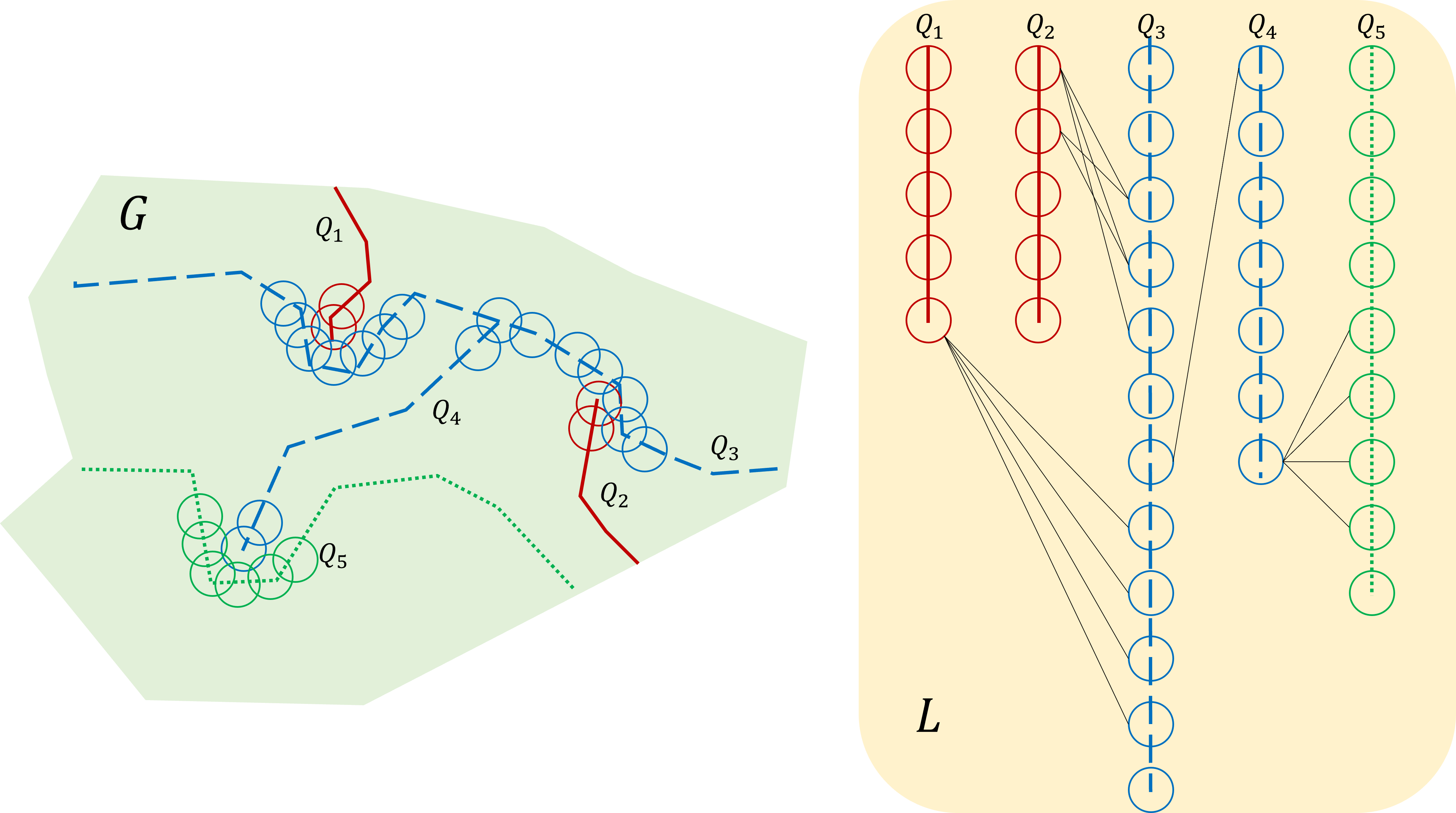}
    \caption{The graph $G$ is depicted on the left, with its path separator. Here, the path separator consists of two paths $Q_1,Q_2$ in $P_0$ (depicted by continuous red lines), two paths $Q_3,Q_4$ in $P_1$ (depicted by dashed blue lines) and a path $Q_5$ in $P_2$ (depicted by a dotted green line). Recall that, for example, $Q_3$ is a shortest path in the graph $H_1=G\setminus P_0$, and not necessarily in $G$. On these paths, the clusters around the $\eta$-landmarks are constructed, and they serve as the vertices of the path-clusters graph $L$ on the right. There are edges between two vertices of $L$ if and only if they intersect as clusters in $G$.}
    \label{fig:PathClustersGraph}
\end{figure}
\end{center}

\pagebreak
The clusters $C(Q,s)$, that are defined in Definition \ref{def:PathClustersGraph} and in the proof of Lemma \ref{lemma:OneScaleSimpleTreeCover} (with $\eta=\epsilon\Delta$), have the property that paths of length at most $\Delta$, that intersect the path separator, are fully contained in one of these clusters, and are well-approximated by a path that goes through the corresponding $\eta$-landmark. Thus, unifying several clusters and applying Observation \ref{obs:StrongPairwise} on the set of the corresponding landmarks takes care of all the $\Delta$-paths that are contained in these clusters. The remaining challenge is to find a way to divide all the clusters into small sets, such that we will be able to unify the clusters in each set, construct a forest cover to the resulting graph, and then glue forest covers of many different sets together. In terms of the path-clusters graph, this can be formalized using the notion of \textit{clustered coloring}. This notion has appeared in the literature with several names or formulations, and was mostly studied from a graph-theoretical viewpoint.

\begin{definition}[\cite{ADOV03}] \label{def:ClusteredColoring}
    Given an undirected graph $L=(V,E)$, a \textbf{clustered coloring} of $L$ with $Z$ colors and block size\footnote{In \cite{ADOV03}, the block size was called the \textit{clustering} of the clustered coloring. However, we find the term \textit{block size} more suitable. E.g., note that when the block size is $\Gamma=1$, a clustered coloring is simply a standard proper coloring of a graph.} $\Gamma$ is a vertex coloring of $L$, in which every monochromatic connected component has at most $\Gamma$ vertices. More formally, it is a function $f:V\rightarrow\{1,2,...,Z\}$ such that for every color $1\leq c\leq Z$, every connected component in the graph $L[f^{-1}(c)]$ has at most $\Gamma$ vertices. Here, $f^{-1}(c)=\{v\in V\;|\;f(v)=c\}$ is the set of vertices colored by $c$, and $L[f^{-1}(c)]$ is the induced graph on $f^{-1}(c)$.
\end{definition}

Next, we argue that for producing tree and forest covers with small stretch and size, one can employ clustered colorings with small parameters for the cluster graphs $L=L(S)$ obtained from path separators $S$ on different levels of the recursion.

\begin{lemma} \label{lemma:ReductionToLooseColoring}
    Fix two real parameters $\eta,\Delta>0$, an integer parameter $k>0$, and a graph $G=(V,E)$ on $n$ vertices. Suppose that every sub-graph $G'$ of $G$ has a path separator $S$, such that the path-clusters graph of $G'$ with respect to $S,\Delta,\eta$ has a clustered coloring with $Z$ colors and block size $\Gamma$. Then, $G$ admits a spanning forest cover $\mathcal{F}_\Delta$, a (not necessarily spanning) tree cover $\mathcal{T}_\Delta$, and an HST cover $\mathcal{H}_\Delta$, all of which of size $O\left(Z\cdot k\cdot \Gamma^{\frac{1}{k}}\cdot\log n\right)$, such that for every $u,v\in V$ with $d_G(u,v)\leq\Delta$,
    \begin{eqnarray*}
        \min_{F\in\mathcal{F}_\Delta}d_F(u,v)&\leq&O(k\log\log (\Gamma+4))\cdot(d_G(u,v)+\eta)~,\\
        \min_{T\in\mathcal{T}_\Delta}d_T(u,v)&\leq&(32ek+1)\cdot(d_G(u,v)+\eta)~,\\
        \min_{T\in\mathcal{H}_\Delta}\rho_T(u,v)&\leq&12ek\cdot(d_G(u,v)+\eta)~.\\
    \end{eqnarray*}
\end{lemma}

\begin{proof}

We prove the lemma by induction over $n$ (while the other parameters are fixed). For a constant $n$, the claim is trivial. For a general $n$, let $S=P_0\cup P_1\cup\cdots\cup P_{\ell-1}$ be a path separator of $G$, and let $C_1,C_2,C_3,...$ be the connected components of $G\setminus S$, all of which with at most $\frac{n}{2}$ vertices. By the induction hypothesis, every component $C_j$ has a spanning forest cover $\mathcal{F}^j_\Delta$, a (not necessarily spanning) tree cover $\mathcal{T}^j_\Delta$, and an HST cover $\mathcal{H}^j_\Delta$, all with size $O\left(Z\cdot k\cdot \Gamma^{\frac{1}{k}}\cdot\log\frac{n}{2}\right)$, such that if $d_{C_j}(u,v)\leq\Delta$, then 
\begin{eqnarray*}
        \min_{F\in\mathcal{F}^j_\Delta}d_F(u,v)&\leq&O(k\log\log (\Gamma+4))\cdot(d_{C_j}(u,v)+\eta)~,\\
        \min_{T\in\mathcal{T}^j_\Delta}d_T(u,v)&\leq&(32ek+1)\cdot(d_{C_j}(u,v)+\eta)~,\\
        \min_{T\in\mathcal{H}^j_\Delta}\rho_T(u,v)&\leq&12ek\cdot(d_{C_j}(u,v)+\eta)~.\\
    \end{eqnarray*}

By Lemma \ref{lemma:UnifyDisjointForestCovers}, there is also a forest cover $\tilde{\mathcal{F}}_\Delta$, a (not necessarily spanning) tree cover $\tilde{\mathcal{T}}_\Delta$, and an HST cover\footnote{Lemma \ref{lemma:UnifyDisjointForestCovers} only considers forest cover, however, it is immediate to generalize it to hold also for HST covers. For this purpose, under some arbitrary enumeration of the HSTs in the HST covers of each connected component, one can merge the $j$'th HST of each of them, for every $j\geq0$, by simply setting all their roots as children of a new root with label $diam(G)$, or $\infty$. A more explicit construction appears at the end of the proof of Theorem \ref{thm:SmallTreewidthSpanTreeCover}.} $\tilde{\mathcal{H}}_\Delta$, all with size $O\left(Z\cdot k\cdot \Gamma^{\frac{1}{k}}\cdot\log\frac{n}{2}\right)$, such that whenever $d_{G\setminus S}(u,v)\leq\Delta$,
\begin{equation} \label{eq:StretchNotIntersect}
\begin{split}
    \min_{F\in\tilde{\mathcal{F}}_\Delta}d_F(u,v)&\leq O(k\log\log (\Gamma+4))\cdot(d_{G\setminus S}(u,v)+\eta)~,\\
    \min_{T\in\tilde{\mathcal{T}}_\Delta}d_T(u,v)&\leq(32ek+1)\cdot(d_{G\setminus S}(u,v)+\eta)~,\\
        \min_{T\in\tilde{\mathcal{H}}_\Delta}\rho_T(u,v)&\leq12ek\cdot(d_{G\setminus S}(u,v)+\eta)~.\\
\end{split}
\end{equation}

Next, denote by $L$ the path-clusters graph of $G$ with respect to $S,\Delta,\eta$. Fix some color $c\in\{1,2,...,Z\}$. In the graph $L[f^{-1}(c)]$ that is induced by the clusters (vertices of $L$) colored by $c$, we consider a connected component $L'$. Note that the number of clusters in $L'$ is at most $\Gamma$, by the properties of the clustered coloring $f$. 
Consider the graph $G_{L'}$ that is induced by the union of all clusters in $L'$, i.e., $G_{L'}=G[\bigcup\{C\in L'\}]$. We now apply Lemma \ref{lemma:PairwiseSpanTreeCover} on $G_{L'}$ with the demand set $A_{L'}=\{q(C)\;|\;C\in L'\}$ (recall that $q(C)$ is the center of the cluster $C$), and obtain a spanning forest cover $\mathcal{F}_{L'}$, a (not necessarily spanning) tree cover $\mathcal{T}_{L'}$, and an HST cover $\mathcal{H}_{L'}$, all with size $O(k\cdot \Gamma^{\frac{1}{k}})$, as $|A_{L'}|=|L'|\leq \Gamma$. By Observation \ref{obs:StrongPairwise}, for every two vertices $u,v$ in $G_{L'}$, and any $x\in A_{L'}$,
\begin{eqnarray*}
    \min_{F\in\mathcal{F}_{L'}}d_F(u,v)&\leq& O(k\log\log(\Gamma+4))\cdot(d_{G_{L'}}(u,x)+d_{G_{L'}}(x,v))~,\\
    \min_{T\in\mathcal{T}_{L'}}d_F(u,v)&\leq& (32ek+1)\cdot(d_{G_{L'}}(u,x)+d_{G_{L'}}(x,v))~,\\
    \min_{T\in\mathcal{H}_{L'}}d_F(u,v)&\leq& 12ek\cdot\max\{d_{G_{L'}}(u,x),d_{G_{L'}}(x,v)\}~.\\
\end{eqnarray*}

Consider all connected components $L'$ of $L[f^{-1}(c)]$. By the definition of the path-clusters graph $L$, clusters do not intersect in $G$ if they are not neighbors in $L$. This implies that the sub-graphs $G_{L'}=G[\bigcup\{C\in L'\}]$ of $G$, for every connected component $L'$ of $L[f^{-1}(c)]$, are pairwise vertex disjoint, and so are the covers $\mathcal{F}_{L'},\mathcal{T}_{L'}$ and $\mathcal{H}_{L'}$. By Lemma \ref{lemma:UnifyDisjointForestCovers},\footnote{Technically, to apply Lemma \ref{lemma:UnifyDisjointForestCovers}, we have to consider a sub-graph $G^c$ of $G$, which is the disjoint union of the graphs $G_{L'}$, \textit{without} any edges that connect the different graphs $G_{L'}$.} the graph $G$ admits a spanning forest cover $\mathcal{F}^c$, a (not necessarily spanning) tree cover $\mathcal{T}^c$, and an HST cover $\mathcal{H}^C$, all with size $O(k\cdot \Gamma^{\frac{1}{k}})$, such that for every two vertices $u,v$ in the same graph $G_{L'}$ (for a connected component $L'$ of $L[f^{-1}(c)]$), and every $x\in A_{L'}$,
\begin{equation} \label{eq:ObservationStrongPairwise}
\begin{split}
    \min_{F\in\mathcal{F}^c}d_F(u,v)&\leq O(k\log\log(\Gamma+4))\cdot(d_{G_{L'}}(u,x)+d_{G_{L'}}(x,v))~,\\
    \min_{T\in\mathcal{T}^c}d_F(u,v)&\leq(32ek+1)\cdot(d_{G_{L'}}(u,x)+d_{G_{L'}}(x,v))~,\\
    \min_{T\in\mathcal{H}^c}d_F(u,v)&\leq12ek\cdot\max\{d_{G_{L'}}(u,x),d_{G_{L'}}(x,v)\}~.\\
\end{split}
\end{equation}

We finally can define the desired covers as
\begin{eqnarray*}
    \mathcal{F}_\Delta&=&\left(\bigcup_{c=1}^Z\mathcal{F}^c\right)\cup\tilde{\mathcal{F}}_\Delta~,\\
    \mathcal{T}_\Delta&=&\left(\bigcup_{c=1}^Z\mathcal{T}^c\right)\cup\tilde{\mathcal{T}}_\Delta~,\\
    \mathcal{H}_\Delta&=&\left(\bigcup_{c=1}^Z\mathcal{H}^c\right)\cup\tilde{\mathcal{H}}_\Delta~.\\
\end{eqnarray*}

The size of each of the covers defined above is bounded by
\[\sum_{c=1}^Z O\left(k\cdot \Gamma^{\frac{1}{k}}\right)+O\left(Z\cdot k\cdot \Gamma^{\frac{1}{k}}\cdot\log\frac{n}{2}\right)=O\left(Z\cdot k\cdot \Gamma^{\frac{1}{k}}\right)\cdot\left(1+\log\frac{n}{2}\right)=O\left(Z\cdot k\cdot \Gamma^{\frac{1}{k}}\cdot\log n\right)~,\]
as desired.

Fix two vertices $u,v\in V$ with $d_G(u,v)\leq\Delta$. If there is a shortest $u$-$v$ path $P_{u,v}$ that does not intersect $S$, then we have $d_{G\setminus S}(u,v)=d_G(u,v)\leq\Delta$, and by (\ref{eq:StretchNotIntersect}), we get
\begin{eqnarray*}
    \min_{F\in\mathcal{F}_\Delta}d_F(u,v)&\leq&\min_{F\in\tilde{\mathcal{F}}_\Delta}d_F(u,v)\leq O(k\log\log (\Gamma+4))\cdot(d_{G}(u,v)+\eta)~,\\
    \min_{T\in\mathcal{T}_\Delta}d_T(u,v)&\leq&\min_{T\in\tilde{\mathcal{T}}_\Delta}d_T(u,v)\leq(32ek+1)\cdot(d_{G}(u,v)+\eta)~,\\
        \min_{T\in\tilde{\mathcal{H}}_\Delta}\rho_T(u,v)&\leq&\min_{T\in\tilde{\mathcal{H}}_\Delta}\rho_T(u,v)\leq12ek\cdot(d_{G}(u,v)+\eta)~.\\
\end{eqnarray*}

Otherwise, we fix a shortest (with respect to $G$) $u$-$v$ path $P_{u,v}$, which must intersect $S$. Let $i$ be the smallest index such that $P_{u,v}$ intersects a path $Q$ from $P_i\subseteq S$. Let $z$ be a vertex in which $P_{u,v}$ and $Q$ intersect, and let $q_s$ be the $\eta$-landmark on $Q$ that is the closest to $z$. In the same way as in the proof of Lemma \ref{lemma:OneScaleSimpleTreeCover}, it is easy to prove that 
$d_{H_i}(z,q_s)\leq\frac{1}{2}\eta$ (using Lemma \ref{lemma:FindLandmarks}). Recall that $H_i=G\setminus\bigcup_{j<i}P_j$. This implies that $z$ is in the cluster $C=C(Q,s)$, and in fact, so are all the vertices $y$ of $P_{u,v}$, since
\[d_{H_i}(y,q_s)\leq d_{H_i}(y,z)+d_{H_i}(z,q_s)\leq d_{H_i}(y,z)+\frac{1}{2}\eta\leq\Delta+\frac{1}{2}\eta~.\]
We also conclude that the distance between $u$ and $v$ in the sub-graph induced by $C$ is at most
\begin{equation} \label{eq:DistInCluster}
    d_C(u,q_s)+d_C(q_s,v)\leq d_{H_i}(u,z)+\frac{1}{2}\eta+d_{H_i}(z,v)+\frac{1}{2}\eta=d_G(u,v)+\eta~,
\end{equation}
where the step $d_{H_i}(u,z)+d_{H_i}(z,v)=d_G(u,v)$ holds since $z$ is on the shortest path $P_{u,v}$, which is fully contained in $H_i$.

Suppose that the cluster $C$ is colored by $c$ in the clustered coloring $f$, i.e., $f(C)=c$, and is in a connected component $L'$ of $L[f^{-1}(c)]$. Then, in the graph $G_{L'}=G[\bigcup\{C'\in L'\}]$ we have
\[d_{G_{L'}}(u,q_s)+d_{G_{L'}}(q_s,v)\stackrel{C\in L'}{\leq}d_C(u,q_s)+d_C(q_s,v)\stackrel{(\ref{eq:DistInCluster})}{\leq}d_G(u,v)+\eta~.\]

Using this result, applying Inequality (\ref{eq:ObservationStrongPairwise}) with $x=q_s\in A_{L'}$ gives
\begin{eqnarray*}
    \min_{F\in\mathcal{F}_\Delta}d_F(u,v)&\leq&\min_{F\in\mathcal{F}^c}d_F(u,v)\leq\min_{F\in\mathcal{F}_{L'}}d_F(u,v)\leq O(k\log\log(\Gamma+4))\cdot(d_G(u,v)+\eta)~,\\
    \min_{T\in\mathcal{T}_\Delta}d_T(u,v)&\leq&\min_{T\in\mathcal{T}^c}d_T(u,v)\leq\min_{T\in\mathcal{T}_{L'}}d_T(u,v)\leq(32ek+1)\cdot(d_G(u,v)+\eta)~,\\
    \min_{T\in\mathcal{H}_\Delta}\rho_T(u,v)&\leq&\min_{T\in\mathcal{H}^c}\rho_T(u,v)\leq\min_{T\in\mathcal{H}_{L'}}\rho_T(u,v)\leq12ek\cdot(d_G(u,v)+\eta)~.
\end{eqnarray*}

\end{proof}

Lemma \ref{lemma:ReductionToLooseColoring} reduces the problem of finding a forest cover (or a tree cover, or an HST cover) in path-separable graphs to finding a clustered coloring of its corresponding path-clusters graph. In fact, Lemma \ref{lemma:OneScaleSimpleTreeCover} can be viewed as an application of Lemma \ref{lemma:ReductionToLooseColoring}, where one colors the clusters of each one of the $p$ paths alternately by $O(\frac{1}{\epsilon})$ distinct colors, using distinct colors for different paths. This results in a clustered coloring with $O(\frac{p}{\epsilon})$ colors and block size $1$. Then, each demand set $A_{L'}$ consists of only one landmark. Thus, applying Lemma \ref{lemma:PairwiseSpanTreeCover} (with $k=1$) actually provides a shortest path tree rooted at this landmark (w.l.o.g.). Setting $\eta=\epsilon\Delta$, we obtain Lemma \ref{lemma:OneScaleSimpleTreeCover}.

Next, we show that for $(\ell,\pi)$-path-separable graphs, there always exists a clustered coloring with $O(\ell\cdot\log n)$ colors and block size in $O(\pi\log\pi)$. Using Lemma \ref{lemma:ReductionToLooseColoring}, this will lead to a forest cover with larger stretch than in Theorem \ref{thm:SpanTreeCoverForPathSeparators}, but with a significantly smaller number of forests.

\begin{lemma} \label{lemma:EllPiColoring}
    Let $G=(V,E)$ be an $n$-vertex $(\ell,\pi)$-path-separable graph, and let $S=P_0\cup P_1\cup\cdots\cup P_{\ell-1}$ be a path separator such that $P_i$ consists of at most $\pi$ shortest paths in $H_i=G\setminus\bigcup_{j<i}P_j$. Fix two parameters $\Delta\geq\eta>0$ and denote $\epsilon=\frac{\eta}{\Delta}$. Then, the path-clusters graph of $G$ with respect to $S,\Delta$ and $\eta$ has a clustered coloring with $O(\ell\cdot\log n)$ colors and block size $O(\frac{\pi}{\epsilon}\log\frac{\pi}{\epsilon})$.
\end{lemma}

\begin{proof}

Denote by $L=(V_L,E_L)$ the path-clusters graph of $G$ with respect to $S,\Delta$ and $\eta$. Recall that the vertices of $L$ are all the clusters of the form $C(Q,s)=\{v\in V(H_i)\;|\;d_{H_i}(q_s,v)\leq\Delta+\frac{1}{2}\eta\}$, for some $0\leq i\leq\ell-1$, where $q_s$ is the $s$'th $\eta$-landmark (see Lemma \ref{lemma:FindLandmarks}) on the path $Q$ from $P_i$. Denote by $L^i$ the sub-graph of $L$ that is induced by such clusters $C(Q,s)$, where $Q$ is a path from $P_i$. Note that the $L^i$'s are pairwise vertex disjoint (for $i\neq i'$ and paths $Q$ from $P_i$ and $Q'$ from $P_{i'}$, respective clusters $C(Q,s),C(Q',s')$ may intersect in $G$ as subsets of $V$, but $C(Q,s)$ is a vertex of $L^i$, and $C(Q',s')$ is a vertex of $L^{i'}$). We construct a clustered coloring of $L$ by focusing on a different $L^i$ each time, and using a distinct palette of $O(\log n)$ colors for each of them. This results in $O(\ell\cdot\log n)$ colors overall.

Fix an index $i\in\{0,1,...,\ell-1\}$. Given a cluster $C$ in $L^i$ and an integer $r\geq0$, we define the $r$-ball 
\[B_{L^i}(C,r)=\{C'\in V(L^i)\;|\;d_{L^i}(C,C')\leq r\}~.\]
Note that the graph $L^i$ is unweighted, thus the notation $d_{L^i}(C,C')$ stands for the unweighted distance (i.e., the smallest number of edges on a path) between $C$ and $C'$ in $L^i$.

Next, we argue that for every cluster $C$ and radius $r>0$, we have the following bound on the size of its ball with radius $r$.
\begin{equation} \label{eq:BallSizeBound}
    |B_{L^i}(C,r)|\leq\left(\frac{12}{\epsilon}\cdot r+1\right)\cdot\pi
\end{equation}

Indeed, given a path $Q$ in $P_i$ and its $\eta$-landmarks $q_1,q_2,...,q_t$, suppose that there is an index $s$ such that $C(Q,s)\in B_{L^i}(C,r)$. Let $(C(Q,s)=C_0,C_1,...,C_l=C)$ be a path in $L^i$ from $C(Q,s)$ to $C$, such that $l\leq r$. Note that the radius of each cluster in this path, as a sub-graph of $H_i$, is at most $\Delta+\frac{1}{2}\eta$, and therefore there is a path in $H_i$ from $q_s$ to $q(C)$ -- the center of $C$ -- with weight at most $(1+(l-1)\cdot2+1)\cdot(\Delta+\frac{1}{2}\eta)=2l\cdot(\Delta+\frac{1}{2}\eta)\leq r\cdot(2\Delta+\eta)$. Let $s_1$ be the lowest index such that $C(Q,s_1)\in B_{L^i}(C,r)$, and let $s_2$ be the highest index such that $C(Q,s_2)\in B_{L^i}(C,r)$. Using the bound we proved for the distances from the landmarks $q_{s_1},q_{s_2}$ to $q(C)$, we conclude that $d_{H_i}(q_{s_1},q_{s_2})\leq2r\cdot(2\Delta+\eta)$. See Figure \ref{fig:BallSizeBound} for an illustration. By Lemma \ref{lemma:FindLandmarks}, the number of landmarks on $Q[q_{s_1},q_{s_2}]$ is at most
\[\frac{2w(Q[q_{s_1},q_{s_2}])}{\eta}+1=\frac{2d_{H_i}(q_{s_1},q_{s_2})}{\eta}+1\leq4r\cdot(\frac{2\Delta}{\eta}+1)+1<\frac{12}{\epsilon}\cdot r+1~.\]
Here we used the fact that $\frac{\Delta}{\eta}=\frac{1}{\epsilon}\geq1$, and that $Q$ is a shortest path in $H_i$, and therefore $Q[q_{s_1},q_{s_2}]$ is the shortest path between $q_{s_1},q_{s_2}$. We conclude that there are at most $\frac{12}{\epsilon}\cdot r+1$ clusters in $B_{L^i}(C,r)$ that are centered in landmarks of $Q$, and therefore there are at most $(\frac{12}{\epsilon}\cdot r+1)\cdot\pi$ clusters in $B_{L^i}(C,r)$ (since there are at most $\pi$ paths $Q$ in $P_i$).

\begin{center}
\begin{figure}[!ht]
    \centering
    \includegraphics[width=14cm, height=8cm]{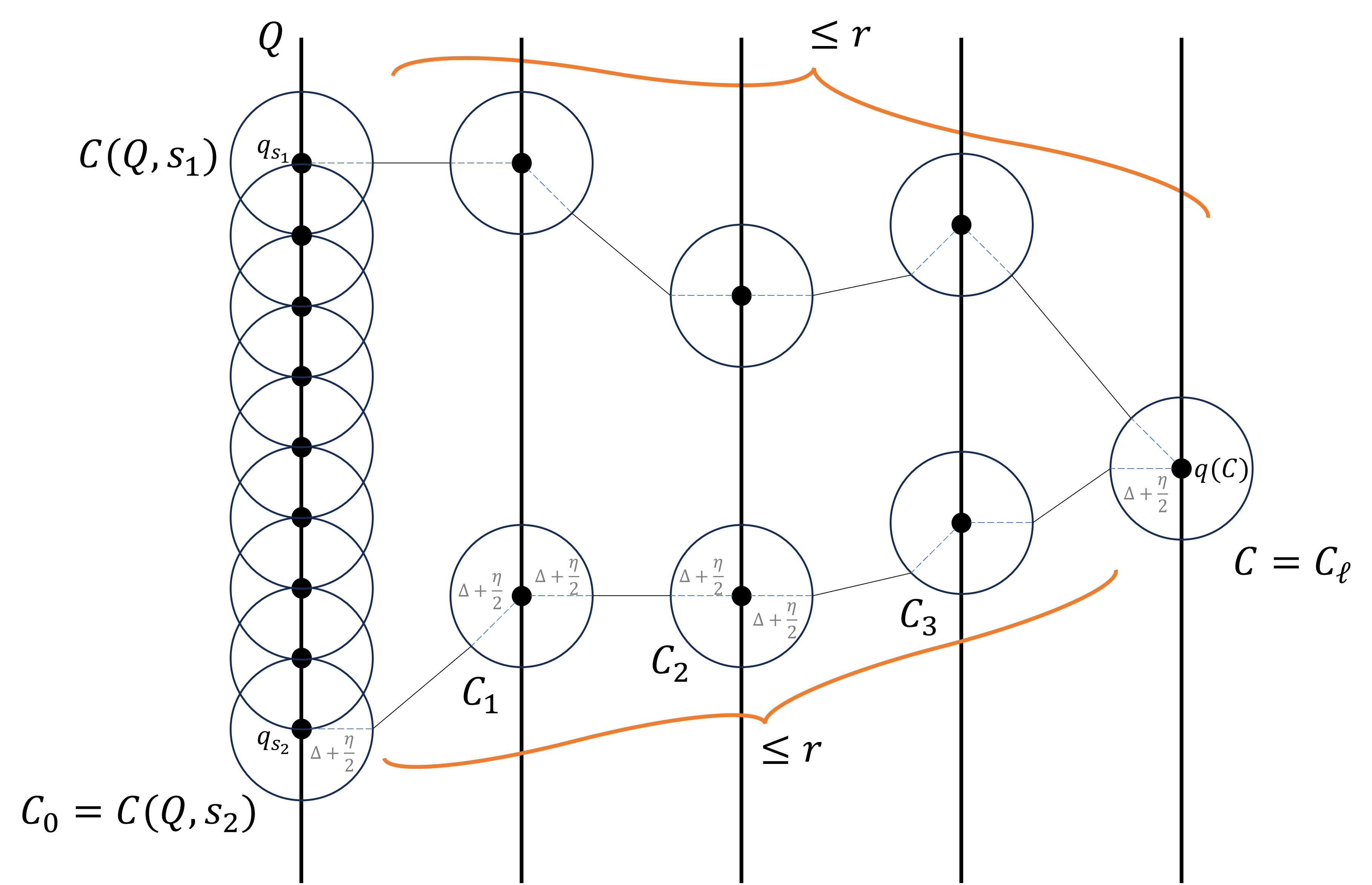}
    \caption{The cluster $C_0=C(Q,s_2)$ is in the $r$-ball of $C$ in $L^i$. Thus, there is a path of at most $r$ edges between $C_0$ and $C$. Every cluster $C_j$ on this path has a radius of at most $\Delta+\frac{1}{2}\eta$ in $H_i$. Thus, this path translates into a path of weight at most $2r\cdot(\Delta+\frac{1}{2}\eta)=r\cdot(2\Delta+\eta)$ in $H_i$ between $q_s$ and $q(C)$, the center of $C$.}
    \label{fig:BallSizeBound}
\end{figure}
\end{center}

We proceed with a standard ball-carving technique \cite{LS93,Bar96,FRT03}. We iterate over the colors $\sigma=1,2,3,...$, starting with $\sigma=1$. Fix an arbitrary cluster $C\in V(L^i)$ and consider the product
\begin{eqnarray*}
    \frac{|B_{L^i}(C,1)|}{|B_{L^i}(C,0)|}\cdot\frac{|B_{L^i}(C,2)|}{|B_{L^i}(C,1)|}\cdots\frac{|B_{L^i}(C,\log_2(\frac{\pi}{\epsilon})+1)|}{|B_{L^i}(C,\log_2(\frac{\pi}{\epsilon}))|}&=&\left|B_{L^i}\left(C,\log_2\left(\frac{\pi}{\epsilon}\right)+1\right)\right|\\
    &\stackrel{(\ref{eq:BallSizeBound})}{\leq}&\left(\frac{12}{\epsilon}\left(\log_2\left(\frac{\pi}{\epsilon}\right)+1\right)+1\right)\cdot\pi~.
\end{eqnarray*}
For simplicity, we assume here that $\log_2(\frac{\pi}{\epsilon})$ is an integer number. At least one of the terms in this product must be bounded by the geometric mean, i.e., there is $r\in\{0,1,...,\log_2(\frac{\pi}{\epsilon})\}$ with 
\[\frac{|B_{L^i}(C,r+1)|}{|B_{L^i}(C,r)|}\leq\left[\left(\frac{12}{\epsilon}\left(\log_2\left(\frac{\pi}{\epsilon}\right)+1\right)+1\right)\cdot\pi\right]^{\frac{1}{\log_2(\frac{\pi}{\epsilon})+1}}\leq13~.\]
Here we used the bound $\left[\left(\frac{12}{\epsilon}\left(\log_2\left(\frac{\pi}{\epsilon}\right)+1\right)+1\right)\cdot\pi\right]^{\frac{1}{\log_2(\frac{\pi}{\epsilon})+1}}\leq\left[\frac{12\pi}{\epsilon}\left(\log_2\left(\frac{\pi}{\epsilon}\right)+1\right)+\frac{\pi}{\epsilon}\right]^{\frac{1}{\log_2(\frac{\pi}{\epsilon})+1}}=\left[12x\left(\log_2x+1\right)+x\right]^{\frac{1}{\log_2x+1}}$, where $x=\frac{\pi}{\epsilon}$. Analyzing the latter function of $x$, one can show that it is decreasing for every $x\geq1$, thus its maximum value $13$ is obtained when $x=1$.

We color the clusters in $B_{L^i}(C,r)$ by the color $\sigma$ and remove all the clusters of the padded ball $B_{L^i}(C,r+1)$ from the graph $L^i$. After the deletion of these clusters, inequality (\ref{eq:BallSizeBound}) still holds, and we repeat the process with the remaining clusters of $L^i$ (each time coloring clusters by the same color $\sigma$). When all the clusters of $L^i$ are removed, we move on to the next color $\sigma+1$ and the graph $L^i$ \textit{without the clusters that were already colored}. This completes the construction of $f$.

To formalize notations, denote by $L^i_{\sigma,j}$ the remaining graph from $L^i$ after removing the $j$'th ball on the $\sigma$'th iteration. The arbitrary cluster that is fixed as center is denoted by $C_{\sigma,j}$ and the corresponding radius $r$, such that $|B_{L^i_{\sigma,j}}(C_{\sigma,j},r+1)|/|B_{L^i_{\sigma,j}}(C_{\sigma,j},r)|\leq13$, is denoted by $r_{\sigma,j}$. That is, the clusters of $B_{L^i_{\sigma,j}}(C_{\sigma,j},r_{\sigma,j})$ are colored by $\sigma$, and the clusters of the padded ball $B_{L^i_{\sigma,j}}(C_{\sigma,j},r_{\sigma,j}+1)$ are removed from $L^i_{\sigma,j}$ to obtain $L^i_{\sigma,j+1}$. 

If $\gamma_\sigma$ balls were removed in the $\sigma$'th iteration, then the padded balls $\{B_{L^i_{\sigma,j}}(C_{\sigma,j},r_{\sigma,j}+1)\}_{j=0}^{\gamma_\sigma-1}$ 
cover all of $V(L^i_{\sigma,0})=V(L^i)$. We bound the number of clusters colored by the color $\sigma$ as follows.
\[\sum_{j=0}^{\gamma_\sigma-1}|B_{L^i_{\sigma,j}}(C_{\sigma,j},r_{\sigma,j})|\geq\sum_{j=1}^{\gamma_\sigma-1}\frac{|B_{L^i_{\sigma,j}}(C_{\sigma,j},r_{\sigma,j}+1)|}{13}\geq\frac{1}{13}|V(L^i_{\sigma,0})|~.\]

That is, each time we color a fraction of at least $\frac{1}{13}$ of the remaining uncolored clusters in $L^i$. Hence, the number of colors until all clusters of $L^i$ are colored is at most $\log_{13/12}|V(L^i)|=O(\log n)$. To color all clusters in $L$, the coloring $f$ uses a total of $O(\ell\cdot\log n)$ colors.

Finally, notice that each ball $B_{L^i_{\sigma,j}}(C_{\sigma,j},r_{\sigma,j})$ that we colored by $\sigma$, for every $\sigma$ and $j$, is a connected component in $L[f^{-1}(\sigma)]$, since neighboring vertices of this ball are immediately removed from $L^i_{\sigma,j}$ (as they are in the padded ball $B_{L^i_{\sigma,j}}(C_{\sigma,j},r_{\sigma,j}+1)$), and are not colored by $\sigma$. Using Inequality (\ref{eq:BallSizeBound}), we conclude that these connected components have at most 
\[\left(\frac{12}{\epsilon}r_{\sigma,j}+1\right)\pi\stackrel{r_{\sigma,j}\leq\log_2\pi}{\leq}\left(\frac{12}{\epsilon}\log_2\left(\frac{\pi}{\epsilon}\right)+1\right)\pi=O\left(\frac{\pi}{\epsilon}\log\frac{\pi}{\epsilon}\right)\]
clusters each, and this is a bound on the block size of the clustered coloring $f$.

\end{proof}

In the next lemma we consider \textit{tree-like} $(\ell,\pi)$-path-separable graphs, where all the path separators in a single union $P_i$ are rooted at the same vertex. We devise a better clustered coloring of the path-clusters graph for this important case (that includes, e.g., $K_r$-minor-free graphs).

\begin{lemma} \label{lemma:TreelikeColoring}
    Let $G=(V,E)$ be an $n$-vertex tree-like $(\ell,\pi)$-path-separable graph, and let $S=P_0\cup P_1\cup\cdots\cup P_{\ell-1}$ be a path separator such that $P_i$ consists of at most $\pi$ shortest paths in $H_i=G\setminus\bigcup_{j<i}P_j$ with a common root $r_i$. Fix two parameters $\Delta\geq\eta>0$ and denote $\epsilon=\frac{\eta}{\Delta}$. Then, the path-clusters graph of $G$ with respect to $S,\Delta$ and $\eta$ has a clustered coloring with $5\ell$ colors and block size $\left\lceil\frac{2}{\epsilon}\right\rceil\cdot\pi$.
\end{lemma}

\begin{proof}

Similarly to the proof of Lemma \ref{lemma:EllPiColoring}, given the path-clusters graph $L$ of $G$ with respect to $S,\Delta$ and $\eta$, we define a coloring $f$ for each sub-graph $L^i$ of $L$, that contains the clusters that are centered on a path from $P_i$.

Fix an index $i\in\{0,1,...,\ell-1\}$. For every shortest path $Q$ in $P_i$, let $q_1,q_2,...,q_t$ be the $\eta$-landmarks of $Q$ from Lemma \ref{lemma:FindLandmarks}. Recall that the sub-graph $L^i$ consists of clusters of the form $C(Q,s)=\{v\in V(H_i)\;|\;d_{H_i}(q_s,v)\leq\Delta+\frac{1}{2}\eta\}$, for a shortest path $Q$ in $P_i$.

We define the color of any cluster $C(Q,s)$ in $L^i$ using the distance of $q_s$, the center of this cluster, to the common root $r_i$ (which is an endpoint of $Q$).
\[f(C(Q,s))=\left\lfloor\frac{d_{H_i}(r_i,q_s)}{\Delta}\right\rfloor\mod{5}~.\]
That is, the color that $f$ sets to the cluster $C(Q,s)$ is the remainder when dividing $\left\lfloor\frac{d_{H_i}(r_i,q_s)}{\Delta}\right\rfloor$ by $5$. Clearly, the clustered coloring $f$ uses at most $5$ different colors for $L^i$, i.e., at most $5\ell$ colors for the entire path-clusters graph $L$.

Fix two separator paths $Q_1,Q_2$ in $P_i$, a landmark $q_{s_1}$ on $Q_1$ and a landmark $q_{s_2}$ on $Q_2$. Suppose that $C(Q_1,s_1)$ and $C(Q_2,s_2)$ intersect each other, and let $v\in C(Q_1,s_1)\cap C(Q_2,s_2)$. Recall that $Q_1,Q_2$ have a common endpoint $r_i$. Then, by the triangle inequality,
\begin{eqnarray*}
    |d_{H_i}(r_i,q_{s_1})-d_{H_i}(r_i,q_{s_2})|
    &\leq&d_{H_i}(q_{s_1},q_{s_2})
    \leq d_{H_i}(q_{s_1},v)+d_{H_i}(v,q_{s_2})\\
    &\leq&\Delta+\frac{1}{2}\eta+\Delta+\frac{1}{2}\eta=2\Delta+\eta~.
\end{eqnarray*}
Therefore, we have
\[\left|\frac{d_{H_i}(r_i,q_{s_1})}{\Delta}-\frac{d_{H_i}(r_i,q_{s_2})}{\Delta}\right|\leq2+\frac{\eta}{\Delta}=2+\epsilon\leq3~,\]
and thus $\left|\left\lfloor\frac{d_{H_i}(r_i,q_{s_1})}{\Delta}\right\rfloor-\left\lfloor\frac{d_{H_i}(r_i,q_{s_2})}{\Delta}\right\rfloor\right|\leq4$.
Hence, by the definition of $f$, if the two clusters have the same color $f(C(Q_1,s_1))=f(C(Q_2,s_2))$, i.e., $\left\lfloor\frac{d_{H_i}(r_i,q_{s_1})}{\Delta}\right\rfloor=\left\lfloor\frac{d_{H_i}(r_i,q_{s_2})}{\Delta}\right\rfloor\mod{5}$ then it must be that $\left\lfloor\frac{d_{H_i}(r_i,q_{s_1})}{\Delta}\right\rfloor=\left\lfloor\frac{d_{H_i}(r_i,q_{s_2})}{\Delta}\right\rfloor$. This implies that for any color $\sigma\in\{0,1,...,4\}$, and any connected component $L'$ in $L^i[f^{-1}(\sigma)]$, all the clusters $C(Q,s)$ in $L'$ have the same value of $\left\lfloor\frac{d_{H_i}(r_i,q_s)}{\Delta}\right\rfloor$.

Fix a color $\sigma\in\{0,1,...,4\}$ and a connected component $L'$ of $L^i[f^{-1}(\sigma)]$. Given a path $Q$ in $P_i$, let $s_1$ be the lowest index such that $C(Q,s_1)$ is in $L'$, and let $s_2$ be the highest index such that $C(Q,s_2)$ is in $L'$. The argument above shows that $\left\lfloor\frac{d_{H_i}(r_i,q_{s_1})}{\Delta}\right\rfloor=\left\lfloor\frac{d_{H_i}(r_i,q_{s_2})}{\Delta}\right\rfloor$, and therefore
\[0\leq w(Q[q_{s_1},q_{s_2}])=|d_{H_i}(r_i,q_{s_1})-d_{H_i}(r_i,q_{s_2})|<\Delta~,\]
where we used the fact that $r_i$ is an endpoint of $Q$, and $Q$ is a shortest path in $H_i$. By Lemma \ref{lemma:FindLandmarks}, the number of landmarks in $Q[q_{s_1},q_{s_2}]$ is at most $\frac{2w(Q[q_{s_1},q_{s_2}])}{\eta}+1<\frac{2\Delta}{\eta}+1=\frac{2}{\epsilon}+1$, and therefore the number of the clusters in $L'$ that are centered on $Q$ is at most $\left\lceil\frac{2}{\epsilon}\right\rceil$. This shows that the component $L'$ has at most $\left\lceil\frac{2}{\epsilon}\right\rceil\cdot\pi$ clusters overall, i.e., the clustered coloring $f$ has block size at most $\left\lceil\frac{2}{\epsilon}\right\rceil\cdot\pi$.

\end{proof}

We now apply the reduction from Lemma \ref{lemma:ReductionToLooseColoring}, to obtain the following corollaries.

\begin{corollary} \label{cor:EllPiOneScale}
    Let $G=(V,E)$ be an $n$-vertex $(\ell,\pi)$-path-separable graph, let $\Delta>0$ and $0<\epsilon\leq1$ be real parameters, and let $k>0$ be an integer parameter. Then, $G$ admits a spanning forest cover $\mathcal{F}_\Delta$, a (not necessarily spanning) tree cover $\mathcal{T}_\Delta$, and an HST cover $\mathcal{H}_\Delta$, each of them with size $O(k\cdot\ell\cdot(\frac{\pi}{\epsilon}\log\frac{\pi}{\epsilon})^{\frac{1}{k}}\cdot\log^2n)$, and for every $u,v\in V$ with $d_G(u,v)\leq\Delta$,
    \begin{eqnarray*}
        \min_{F\in\mathcal{F}_\Delta}d_F(u,v)&\leq& O(k\log\log(\pi+4))\cdot(d_G(u,v)+\epsilon\Delta)~,\\
        \min_{T\in\mathcal{T}_\Delta}d_T(u,v)&\leq&(32ek+1)\cdot(d_G(u,v)+\epsilon\Delta)~,\\
        \min_{T\in\mathcal{H}_\Delta}\rho_T(u,v)&\leq&12ek\cdot(d_G(u,v)+\epsilon\Delta)~.
    \end{eqnarray*}

    If $G$ is \textbf{tree-like} $(\ell,\pi)$-path-separable, then the size of these covers is $O(k\cdot\ell\cdot(\frac{\pi}{\epsilon})^{\frac{1}{k}}\cdot\log n)$.
    
\end{corollary}

The covers from Corollary \ref{cor:EllPiOneScale} provide low stretch for pairs of vertices $u,v\in V$ such that $d_G(u,v)\approx\Delta$, for any given $\Delta>0$. To obtain low stretch for \textit{all} pairs of vertices in the graph, we apply them on different scales $\Delta$. As in Section \ref{sec:PathSeparableSpanSmallStretchTreeCovers}, if we consider every scale $\Delta=2^i$, for $i=0,1,...,\log_2\Lambda$, we obtain covers with size $O(k\cdot\ell\cdot(\frac{\pi}{\epsilon})^{\frac{1}{k}}\cdot\log n\cdot\log\Lambda)$ (for the tree-like case) and low stretch ($O(k\log\log(\pi+4))$, $(32e+\epsilon)k+1$ and $(12e+\epsilon)k$ for a spanning forest cover, not necessarily spanning tree cover and HST cover, respectively). To get rid of the dependence on the aspect ratio $\Lambda$, and replace $\log\Lambda$ by $O(\log n)$, we use the same reduction that we used in the proof of Theorem \ref{thm:SpanTreeCoverForPathSeparators}, that is presented in Appendix \ref{app:AspectRatioReduction}.

\begin{theorem} \label{thm:LargeStretchTreeCoverForPathSeparators}
    Let $G=(V,E)$ be an $n$-vertex $(\ell,\pi)$-path-separable graph, let $0<\epsilon\leq1$ be a real parameter and let $k>0$ be an integer.
    \begin{enumerate}
        \item $G$ has a spanning forest cover $\mathcal{F}$ with stretch $O(k\log\log(\pi+4))$, a (not necessarily spanning) tree cover $\mathcal{T}$ with stretch $(32e+\epsilon)k+1=O(k)$, and an HST cover $\mathcal{H}$ with stretch $(12e+\epsilon)k$, all of them with size 
        \[O\left(k\cdot\ell\cdot\left(\frac{\pi}{\epsilon}\log\frac{\pi}{\epsilon}\right)^\frac{1}{k}\cdot\log^3n\right)~.\]
        
        If $G$ is connected, then by Observation \ref{obs:FromForestCoverToTreeCover}, $G$ also has a spanning tree cover with the same size and stretch $O(k\log\log(\pi+4))$.
        
        \item If $G$ is \textit{tree-like} $(\ell,\pi)$-path-separable, then the size of all these covers improves to 
        \[O\left(k\cdot\ell\cdot\left(\frac{\pi}{\epsilon}\right)^\frac{1}{k}\cdot\log^2n\right)~.\]
    \end{enumerate}
\end{theorem}

Note that for not-tree-like $(\ell,\pi)$-path-separable graphs, one can replace the parameter $k$ by $2k$, then the size of the covers in Theorem \ref{thm:LargeStretchTreeCoverForPathSeparators} becomes $O\left(\left(\frac{1}{\epsilon}\right)^{\frac{1}{k}}\cdot k\cdot\ell\cdot\pi^\frac{1}{k}\cdot\log^3n\right)$, while the stretch grows by a factor of at most $2$. In particular, the stretch of the spanning forest/tree cover remains $O(k\log\log(\pi+4))$.

Recall that $K_r$-minor-free graphs are (weakly) tree-like $(\ell,\pi)$-path-separable, for $\ell=O(r^2)$ and $\pi=O(r^{c_{AG}})$, where $c_{AG}=4602$, and that graphs with bounded genus $g$ are strongly tree-like $O(g)$-path-separable (see Items (2) and (3) in Theorem \ref{thm:PathSeparableFamilies}). Thus, the following corollary is derived by Theorem \ref{thm:LargeStretchTreeCoverForPathSeparators}.

\begin{corollary}
    Let $G$ be an $n$-vertex graph, and let $k>0$ be an integer.
    \begin{enumerate}
        \item If $G$ is $K_r$-minor-free, then it has a spanning forest cover $\mathcal{F}$ with stretch $O(k\log\log r)$, and a (not necessarily spanning) tree cover $\mathcal{T}$ and an HST cover $\mathcal{H}$, both with stretch $O(k)$. Each of the covers $\mathcal{F},\mathcal{T}$ and $\mathcal{H}$ has size 
        \[O\left(k\cdot r^{2+\frac{1}{k}}\cdot\log^2n\right)~.\]
        If $G$ is connected, then by Observation \ref{obs:FromForestCoverToTreeCover}, $G$ also has a spanning tree cover with the same size and stretch $O(k\log\log r)$.
        
        \item If $G$ has bounded genus $g$, then it has a spanning forest cover $\mathcal{F}$ with stretch $O(k\log\log g)$, and a (not necessarily spanning) tree cover $\mathcal{T}$ and an HST cover $\mathcal{H}$, both with stretch $O(k)$. Each of the covers $\mathcal{F},\mathcal{T}$ and $\mathcal{H}$ has size 
        \[O\left(k\cdot g^{\frac{1}{k}}\cdot\log^2n\right)~.\]
        If $G$ is connected, then by Observation \ref{obs:FromForestCoverToTreeCover}, $G$ also has a spanning tree cover with the same size and stretch $O(k\log\log g)$.
        The result applies for both orientable and non-orientable genus.
    \end{enumerate}

\end{corollary}

\section{Path-Reporting Spanners and Emulators} \label{sec:SpannersAndEmulators}

In this section, we show how tree covers give rise to path-reporting spanners and emulators. Spanners and emulators are well-studied objects (see Definition \ref{def:SpannersAndEmulators} for a formal definition). In this section we provide results for the stronger notions of \textit{path-reporting} spanners/emulators (see Definition \ref{def:PathReportingStructures}).

For the special case where the given graph $G$ is a tree, we next prove a useful fact, that $G$ admits a path-reporting spanner with optimal stretch, query time and size ($1$, $O(1)$ and $O(n)$, respectively). For this purpose, we use the following result by Harel and Tarjan \cite{HT84}.
\begin{theorem}[\cite{HT84}] \label{thm:LCAOracle}
    For every rooted tree $T$ on $n$ vertices there is an oracle with size $O(n)$, that given two vertices in $T$, returns their lowest common ancestor (LCA) in $T$ within $O(1)$ time.
\end{theorem}

\begin{lemma} \label{lemma:TreeOracle}
    Let $T=(V,E)$ be an undirected weighted tree with $n$ vertices. There is a path-reporting $1$-spanner $(S_T,D_T)$ with constant query time and size $O(n)$.
\end{lemma}

\begin{proof}

First, note that the set $S_T$ of underlying edges (in which all the output paths of $D_T$ must be contained) must be exactly $E(T)$, the edges of $T$. The size of this set is $n-1$.

Next, we construct the oracle $D_T$. For an arbitrary chosen root $r\in V$, the oracle $D_T$ stores, for every $v\in V$, the distance $d_T(r,v)$ and 
the parent $p(v)$ of $v$ in $T$ (i.e., the next vertex on the unique path from $v$ to the root $r$). In addition, $D_T$ stores the LCA oracle $M$ from Theorem \ref{thm:LCAOracle}, for the rooted tree $T$. Since the size of $M$ is $O(n)$, the total size of $D_T$ is also $O(n)$.

Given a query $(u,v)\in V^2$, the oracle $D_T$ invokes the LCA oracle $M$ to find the LCA $z$ of $u,v$ in $T$. Then, the distance $d_T(u,v)$ is computed within time $O(1)$ by
\[d_T(u,v)=d_T(u,z)+d_T(z,v)=d_T(r,u)-d_T(r,z)+d_T(r,v)-d_T(r,z)~.\]

If a $u$-$v$ path $P_{u,v}\subseteq S_T=E(T)$ is also required, the oracle $D_T$ uses the pointers $p()$, starting from $u,v$, until it gets to $z$. The concatenation of these two paths is the unique path in $T$ between $u,v$. The time for this computation is proportional to the number of edges in the output path.
    
\end{proof}

Applying Lemma \ref{lemma:TreeOracle} on every tree of a given tree cover produces a path-reporting emulator (if the tree cover is non-spanning) or a path-reporting spanner (if the tree cover is spanning). The following lemma provides a formal argument.

\begin{lemma} \label{lemma:TreeCoverToPREmulator}
Let $G=(V,E)$ be an undirected weighted graph, let $\mathcal{P}\subseteq V^2$ be a set of vertex pairs, and let $\mathcal{T}$ be a $\mathcal{P}$-pairwise tree cover for $G$ with stretch $\alpha$, (maximum) overlap $q$ and average overlap $s$. Then, the graph $G$ has a $\mathcal{P}$-pairwise path-reporting $\alpha$-emulator $(S,D)$ with query time $O(q)$ and size $O(s|V|)$. Moreover, if $\mathcal{T}$ is a spanning tree cover, then $(S,D)$ is a $\mathcal{P}$-pairwise path-reporting spanner with the same properties.
\end{lemma}

\begin{proof}

Given the tree cover $\mathcal{T}$, for every $T\in\mathcal{T}$, construct the path-reporting $1$-spanner $(S_T,D_T)$ from Lemma \ref{lemma:TreeOracle} (here, $S_T$ is simply $E(T)$, the edges of the tree $T$). We define an oracle $D$ to contain all the oracles $D_T$ for $T\in\mathcal{T}$. In addition, for every $u\in V$, the oracle $D$ contains a list $L(u)$ of pointers to all trees $T\in\mathcal{T}$ such that $u\in V(T)$.

Given a query $(u,v)\in\mathcal{P}$, the oracle $D$ iterates over all trees $T\in\mathcal{T}$ that appear both in $L(u)$ and in $L(v)$. For every such tree $T$, using the oracle $D_T$, the oracle $D$ computes the distance estimate $d_T(u,v)$, within time $O(1)$. This way, the tree $T_0$ with minimal value of $d_T(u,v)$ is found. As $\mathcal{T}$ is a $\mathcal{P}$-pairwise tree cover with stretch $\alpha$, this value is always at least $d_G(u,v)$, and if $(u,v)\in\mathcal{P}$, it is at most $\alpha\cdot d_G(u,v)$. Lastly, the oracle $D$ returns either the distance estimate $d_{T_0}(u,v)$, within time $O(1)$, or a path $P_{u,v}$ in $T_0$ with weight $d_{T_0}(u,v)$, within time $O(|P_{u,v}|)$. Note that the total query time of $D$ is $O(|L(u)|+|L(v)|)+O(1)=O(q)$ if an output path is not required, or $O(q+|P_{u,v}|)$ if an output path is required.

By Lemma \ref{lemma:TreeOracle}, the size of each oracle $D_T$ is $O(|V(T)|)$. Hence, the size of $D$ is
\[\sum_{T\in\mathcal{T}}O(|V(T)|)+\sum_{u\in V}O(|L(u)|)=O\left(\sum_{u\in V}|\{T\in\mathcal{T}\;|\;u\in V(T)\}|\right)=O(s|V|)~,\]
by the definition of the average overlap (see Definition \ref{def:PairwiseForestCover}).

Note that all of the output paths of $D$ are contained in the union $\bigcup_{T\in\mathcal{T}}E(T)$. Thus, we define the set $S$ of underlying edges to be this union. The size of $S$ is
\[|S|\leq\sum_{T\in\mathcal{T}}|E(T)|<\sum_{T\in\mathcal{T}}|V(T)|\leq s|V|~.\]

We conclude that $(S,D)$ is a path-reporting $\alpha$-emulator with query time $O(q)$ and size $O(s|V|)$. If $\mathcal{T}$ is a spanning tree cover, then by definition, it means that $E(T)\subseteq E$ for all $T\in\mathcal{T}$. Thus, $S=\bigcup_{T\in\mathcal{T}}E(T)\subseteq E$, i.e., $(S,D)$ is a path-reporting spanner.
    
\end{proof}

Lemma \ref{lemma:TreeCoverToPREmulator} provides a unified framework of converting tree covers into path-reporting emulators and spanners. In fact, one can verify that the well known results by \cite{TZ01} and \cite{MN06} can be achieved using this framework (however, with an improved query time that relies on their specific constructions).

We proceed by applying Lemma \ref{lemma:TreeCoverToPREmulator} on our new tree covers in this paper. By applying Lemma \ref{lemma:TreeCoverToPREmulator} on Corollary \ref{cor:BetterAverageOverlap}, we obtain the following theorem.

\begin{theorem} \label{thm:RamseyPRSpanner}
Let $G=(V,E)$ be an undirected weighted graph on $n$ vertices. Given an integer parameter $k\geq1$, there exists a path-reporting spanner with stretch $O(k\log\log n)$, query time $O(kn^{\frac{1}{k}})$ 
and size $O(n^{1+\frac{1}{k}})$.
\end{theorem}

\begin{remark}
    Though the result of Theorem \ref{thm:RamseyPRSpanner} is weaker than the recent path-reporting distance oracles of \cite{ES23,CZ24}, it nevertheless demonstrates that a linear-size path-reporting $O(\log n\cdot\log\log n)$-spanner can be constructed as a direct application of the spanning tree covers of \cite{ACEFN20}.
\end{remark}

For vertex-separable graphs, we apply Lemma \ref{lemma:TreeCoverToPREmulator} on Theorem \ref{thm:SmallTreewidthSpanTreeCover}. We obtain the following result.

\begin{theorem} \label{thm:SmallTreewidthEmulatorAndSpanner}
Let $G=(V,E)$ be an $n$-vertex undirected weighted $s$-vertex-separable graph, for some non-decreasing function $s=s(\theta)$. Let $k\geq1$ be an integer parameter, and denote $S(n,k)=\sum_{i=0}^{\log n}s\left(\frac{n}{2^i}\right)^{\frac{1}{k}}$. Then, $G$ has a path-reporting emulator with stretch $32ek+1$, and a path-reporting spanner with stretch $\hat{c}\cdot k\log\log s(n)$, both with query time $O\left(k\cdot S(n,k)\right)$ and size $O\left(k\cdot n\cdot S(n,k)\right)$.
\end{theorem}

The following corollaries are obtained by setting $s(\theta)=\theta^{\delta}$, for some $\delta\in(0,1]$, and setting $s(\theta)\equiv t(n)$, for some function $t$ of $n$. They are analogous to Corollaries \ref{cor:TWN^delta} and \ref{cor:FlatTW}.

\begin{corollary}
Let $G$ be an $n$-vertex $s$-vertex-separable graph, for $s=s(\theta)=\theta^\delta$, for some $\delta\in(0,1]$, and let $k\geq1$ be an integer. Then, $G$ has a path-reporting emulator with stretch $32ek+1$, and a path-reporting spanner with stretch $O(k\log\log n)$, both with query time $O\left(\frac{k^2}{\delta}n^{\frac{\delta}{k}}\right)$ and size $O\left(\frac{k^2}{\delta}n^{1+\frac{\delta}{k}}\right)$.
\end{corollary}

\begin{corollary}
Let $G$ be an $n$-vertex graph with treewidth $t=t(n)$, and let $k\geq1$ be an integer. Then, $G$ admits a path-reporting emulator with stretch $32ek+1$, and a path-reporting spanner with stretch $O(k\log\log t)$, both with query time $O\left(k\cdot t^{\frac{1}{k}}\log n\right)$ and size $O\left(k\cdot t^{\frac{1}{k}}\cdot n\log n\right)$.
\end{corollary}

These results for path-reporting emulators improve upon previous bounds on the size-stretch tradeoff \cite{TZ01,MN06,ES23,CZ24} when $s=O(n^{\delta})$, for a sufficiently small constant $\delta>0$. For path-reporting spanners, our results improve upon existing state-of-the-art bounds for this tradeoff \cite{ES23,CZ24,TZ01,C15} for $s=2^{O(\frac{n}{\log\log n})}$. Note, however, that our query time is inferior to that of \cite{TZ01,MN06,ES23,CZ24,C15}.

Applying Lemma \ref{lemma:TreeCoverToPREmulator} on Theorem \ref{thm:SpanTreeCoverForPathSeparators} for $p$-path-separable graphs provides a path-reporting spanner with stretch $1+\epsilon$, and query time and size overhead (i.e., size divided by $n$) $O(p\cdot\frac{\log^2n}{\epsilon})$. 
Next, we use a more subtle construction that employs the fact that the tree cover from Theorem \ref{thm:SpanTreeCoverForPathSeparators} consists of a different tree cover for each \textit{scale}. To determine which scales should be used, we employ the tree cover of Gupta et al. \cite{GKR04}. This eliminates one of the $\log n$ factors from the query time of the resulting path-reporting spanner.

\begin{theorem} \label{thm:SmallStretchImprovedQT}
    Let $G=(V,E)$ be an $n$-vertex $p$-path-separable graph, let $k>0$ be an integer parameter and let $\epsilon\in(0,1]$ be a real parameter. Then, $G$ admits a path-reporting spanner with stretch $1+\epsilon$, query time $O\left(p\cdot\frac{\log n}{\epsilon}\right)$, and size $O\left(n\cdot p\cdot\frac{\log^2n}{\epsilon}\right)$.
\end{theorem}

\begin{proof}

For simplicity, we show a path-reporting spanner with size $O\left(n\cdot p\cdot\frac{\log n\cdot\log\Lambda}{\epsilon}\right)$, instead of $O\left(n\cdot p\cdot\frac{\log^2n}{\epsilon}\right)$. To obtain the latter, we apply our new aspect ratio reduction (see Appendix \ref{app:PRSpannerAspectRatioReduction1}). Here, as well as in Appendix \ref{app:PRSpannerAspectRatioReduction1}, we assume that $\min_{u\neq v\in V}d_G(u,v)=\min_{e\in E}w(e)=1$, as otherwise we multiply the weight of each edge by $(\min_{u\neq v\in V}d_G(u,v))^{-1}$. Then, the aspect ratio $\Lambda$ is simply the diameter $\max_{u,v\in V}d_G(u,v)$.
We also assume that the graph $G$ is connected, as otherwise we construct a path-reporting spanner for every connected component, and their union is the desired path-reporting spanner for $G$. 

In this proof we use the following result of \cite{GKR04}, for tree covers of path-separable graphs.

\begin{theorem}[\cite{GKR04}] \label{thm:GKR}
    Every $n$-vertex $p$-path-separable graph admits a (spanning) tree cover with stretch $3$ and size $O(p\cdot\log n)$.
\end{theorem}

Applying Lemma \ref{lemma:TreeCoverToPREmulator} on this tree cover by \cite{GKR04} (with $\mathcal{P}=V^2$), we obtain a path-reporting $3$-spanner $(S_0,D_0)$ with query time $O(p\cdot\log n)$ and size $O(n\cdot p\cdot\log n)$. 

For every $i=1,2,...,\lceil\log_2\Lambda\rceil$, let $\mathcal{F}_{2^i}$ be the forest cover $\mathcal{F}_{\Delta}$ from Lemma \ref{lemma:OneScaleSimpleTreeCover}, for $\Delta=2^i$. By Observation \ref{obs:FromForestCoverToTreeCover}, for every $i$ there is also a tree cover $\mathcal{T}_{2^i}$ for $G$ with the same properties. That is, the size of $\mathcal{T}_{2^i}$ is $O(p\cdot\frac{\log n}{\epsilon})$, and for every $u,v\in V$ such that $d_G(u,v)\leq2^i$, there is a tree $T\in\mathcal{T}_{2^i}$ with $d_T(u,v)\leq d_G(u,v)+\epsilon\cdot2^i$. This implies that $\mathcal{T}_{2^i}$ is a $\mathcal{P}_{2^i}$-pairwise tree cover with stretch $1+2\epsilon$, for the set of vertex pairs $\mathcal{P}_{2^i}=\{(u,v)\in V^2\;|\;2^{i-1}<d_G(u,v)\leq2^i\}$. Indeed, for every pair $(u,v)\in\mathcal{P}_{2^i}$ we have $d_G(u,v)\leq2^i$, thus there is a tree $T\in\mathcal{T}_{2^i}$ with
\[d_T(u,v)\leq d_G(u,v)+\epsilon\cdot2^i<d_G(u,v)+2\epsilon\cdot d_G(u,v)=(1+2\epsilon)d_G(u,v)~.\]

We apply Lemma \ref{lemma:TreeCoverToPREmulator} on $\mathcal{T}_{2^i}$, and obtain a $\mathcal{P}_{2^i}$-pairwise path-reporting $(1+2\epsilon)$-spanner $(S_i,D_i)$ for $G$ with query time $O(p\cdot\frac{\log n}{\epsilon})$ and size $O(n\cdot p\cdot\frac{\log n}{\epsilon})$. In our new path-reporting spanner $(S,D)$, we define $S=\bigcup_{i=1}^{\lceil\log_2\Lambda\rceil}S_i$, and we store all the oracles $\{D_i\}_{i=0}^{\lceil\log_2\Lambda\rceil}$ in $D$.

Given a query $(u,v)\in V^2$, the oracle $D$ first applies $D_0$ on $(u,v)$, to get an estimate $\hat{d}_0(u,v)\in[d_G(u,v),3d_G(u,v)]$, within time $O(p\cdot\log n)$. Equivalently, we have $\frac{1}{3}\hat{d}_0(u,v)\leq d_G(u,v)\leq\hat{d}_0(u,v)$. Then, the oracle $D$ finds all the indices $i\in\{1,2,...,\lceil\log_2\Lambda\rceil\}$ such that the interval $[\frac{1}{3}\hat{d}_0(u,v),\hat{d}_0(u,v)]$ intersects the interval $(2^{i-1},2^i]$. We claim that there are at most three such values $i$. Indeed, if $x\in(2^{i-1},2^i]$ and $y\in(2^{j-1},2^j]$ are both in $[\frac{1}{3}\hat{d}_0(u,v),\hat{d}_0(u,v)]$, for some $i<j$, then 
\[2^{j-1}<y\leq\hat{d}_0(u,v)\leq3x\leq3\cdot2^i<2^{i+2}~,\]
i.e., $j\leq i+2$. Hence, there is an index $i\in\{1,2,...,\lceil\log_2\Lambda\rceil\}$ such that $d_G(u,v)\in(2^{i-1},2^i]\cup(2^i,2^{i+1}]\cup(2^{i+1},2^{i+2}]$, or equivalently, $(u,v)\in\mathcal{P}_{2^i}\cup\mathcal{P}_{2^{i+1}}\cup\mathcal{P}_{2^{i+2}}$. Thus, one of the oracles $D_i,D_{i+1},D_{i+2}$ outputs a distance estimate for $(u,v)$ with stretch at most $1+2\epsilon$. The oracle $D$ applies these oracles on $(u,v)$ and finds the index $j\in\{i,i+1,i+2\}$ such that $D_j$ outputs the minimal distance estimate $\hat{d}(u,v)$ for $(u,v)$. If a path is also required, $D$ applies $D_j$ on $(u,v)$ to obtain such an approximate path. In both cases, $\hat{d}(u,v)$ has stretch at most $1+2\epsilon$.

Note that the query time of $D$ equals the query time of $D_0$ plus the query times of $D_i,D_{i+1}$ and $D_{i+2}$. As these are all $O(p\cdot\frac{\log n}{\epsilon})$, the query time of $D$ is $O(p\cdot\frac{\log n}{\epsilon})$ as well.

Lastly, observe that the output paths of $D$ are all in $S=\bigcup_{i=1}^{\lceil\log_2\Lambda\rceil}S_i$, and that the size of $S$ and of $D$ are both at most the total size of the oracles $\{D_i\}_{i=0}^{\lceil\log_2\Lambda\rceil}$, which is $O(n\cdot p\cdot\frac{\log n\cdot\log\Lambda}{\epsilon})$.
    
\end{proof}

Consider Theorem \ref{thm:LargeStretchTreeCoverForPathSeparators}, in which we presented tree covers for $(\ell,\pi)$-path-separable graphs with stretch $O(k\log\log(\pi+4))$ and \textit{sublinear} size in $\pi$. Applying Lemma \ref{lemma:TreeCoverToPREmulator} on this theorem provides the following result.

\begin{theorem} \label{thm:LargeStretchDirectPRDO}
    Let $G=(V,E)$ be an $n$-vertex $(\ell,\pi)$-path-separable graph, let $k>0$ be an integer parameter and let $\epsilon\in(0,1]$ be a real parameter. Then, $G$ admits a path-reporting spanner (respectively, emulator) with stretch $O(k\log\log(\pi+4))$ (resp., $(32e+\epsilon)k+1$), query time $q_S=O(k\cdot\ell\cdot\pi^{\frac{1}{k}}\cdot\log^3n)$ (resp., $q_E=O\left(k\cdot\ell\cdot(\frac{\pi}{\epsilon}\log\frac{\pi}{\epsilon})^{\frac{1}{k}}\cdot\log^3n\right)$), and size $O(q_S\cdot n)$ (resp., $O(q_E\cdot n)$). If $G$ is also \textit{tree-like} $(\ell,\pi)$-path-separable, then $q_S$ improves to $O(k\cdot\ell\cdot\pi^{\frac{1}{k}}\cdot\log^2n)$ and $q_E$ improves to $O\left(k\cdot\ell\cdot(\frac{\pi}{\epsilon}\log\frac{\pi}{\epsilon})^{\frac{1}{k}}\cdot\log^2n\right)$.
\end{theorem}

A similar proof as the proof of Theorem \ref{thm:SmallStretchImprovedQT} improves the query time of these path-reporting spanners and emulators to $\ell\cdot O(\log\pi\cdot\log^3n+k\pi^{\frac{1}{k}}\cdot\log^2n)$. For tree-like separators, the bound improves to $\ell\cdot O(\log\pi\cdot\log^2n+k\pi^{\frac{1}{k}}\cdot\log n)$. To find the right scale, we use the distance oracle that is implied by Theorem \ref{thm:LargeStretchDirectPRDO}, with $k=\log\pi$. The proof of the following result is in Appendix \ref{app:PRSpannerAspectRatioReduction2}

\begin{theorem} \label{thm:LargeStretchImprovedQT}
    Let $G=(V,E)$ be an $n$-vertex $(\ell,\pi)$-path-separable graph, let $k>0$ be an integer parameter and let $\epsilon\in(0,1]$ be a real parameter. Then, $G$ admits a path-reporting spanner (respectively, emulator) with stretch $O(k\log\log(\pi+4))$ (resp., $(32e+\epsilon)k+1$), query time $\ell\cdot O(\log\pi\cdot\log^{3-\gamma}n+k\pi^{\frac{1}{k}}\cdot\log^{2-\gamma}n\cdot\log\log\pi)$ (resp., $\ell\cdot O\left(\log\pi\cdot\log^{3-\gamma}n+k\cdot(\frac{\pi}{\epsilon}\log\frac{\pi}{\epsilon})^{\frac{1}{k}}\cdot\log^{2-\gamma}n\cdot\log\log\pi\right)$), and size $O(k\cdot\ell\cdot\pi^{\frac{1}{k}}\cdot n\log^{3-\gamma}n)$ (resp., $O(k\cdot(\frac{\pi}{\epsilon}\log\frac{\pi}{\epsilon})^{\frac{1}{k}}\cdot n\log^{3-\gamma}n)$). Here, $\gamma=1$ if $G$ is also \textit{tree-like} $(\ell,\pi)$-path-separable, and otherwise $\gamma=0$.
\end{theorem}

Recall that $K_r$-minor-free graphs are (weakly) tree-like $(\ell,\pi)$-path-separable, for $\ell=O(r^2)$ and $\pi=O(r^{c_{AG}})$, where $c_{AG}=4602$, and that graphs with bounded genus $g$ are strongly tree-like $O(g)$-path-separable (see Items (2) and (3) in Theorem \ref{thm:PathSeparableFamilies}). Thus, the following corollaries are derived by Theorems \ref{thm:SmallStretchImprovedQT} and \ref{thm:LargeStretchImprovedQT}.

\begin{corollary}
    Let $G$ be a $K_r$-minor-free $n$-vertex graph and let $k>0$ be an integer parameter. Recall that $c_{AG}=4602$.
    \begin{enumerate}
        \item $G$ has a path-reporting spanner with stretch $1+\epsilon$, query time 
        $O\left(r^{c_{AG}}\cdot\frac{\log n}{\epsilon}\right),$
        and size $O\left(n\cdot r^{c_{AG}}\cdot\frac{\log^2n}{\epsilon}\right)$.
        
        \item $G$ has a path-reporting spanner with stretch $O(k\log\log r)$, and a path-reporting emulator with stretch $O(k)$, both with query time 
        $O(r^2\log r\cdot\log^2n+kr^{2+\frac{1}{k}}\cdot\log n\cdot\log\log r)$
        and size $O(kr^{2+\frac{1}{k}}\cdot n\log^2n)$.
    \end{enumerate}
\end{corollary}

\begin{corollary}
    Let $G$ be an $n$-vertex graph with bounded genus $g$, let $k>0$ be an integer parameter and let $\epsilon\in(0,1]$ be a real parameter.
    \begin{enumerate}
        \item $G$ has a path-reporting $(1+\epsilon)$-spanner with query time 
        $O\left(g\cdot\frac{\log n}{\epsilon}\right),$
        and size $O\left(n\cdot g\cdot\frac{\log^2n}{\epsilon}\right)$.
    
        \item $G$ has a path-reporting spanner with stretch $O(k\log\log g)$, query time 
        $O(\log g\cdot\log^2n+kg^{\frac{1}{k}}\cdot\log n\cdot\log\log g)$,
        and size $O(kg^{\frac{1}{k}}\cdot n\log^2n)$.
        
        \item $G$ has a path-reporting emulator with stretch $(32e+\epsilon)k+1$, query time\newline 
        $O\left(\log g\cdot\log^2n+k\cdot(\frac{g}{\epsilon}\log\frac{g}{\epsilon})^{\frac{1}{k}}\cdot\log n\cdot\log\log g\right)$, 
        and size $O(k\cdot(\frac{g}{\epsilon}\log\frac{g}{\epsilon})^{\frac{1}{k}}\cdot n\log^2n)$.
    \end{enumerate}
\end{corollary}

\subsection{Path-Reporting Spanners with Improved Stretch} \label{sec:SpanImprovedStretch}

Consider our new path-reporting spanners for general graphs (Theorem \ref{thm:RamseyPRSpanner}), vertex-separable graphs (Theorem \ref{thm:SmallTreewidthEmulatorAndSpanner}) and path-separable graphs (Theorems \ref{thm:SmallStretchImprovedQT} and \ref{thm:LargeStretchImprovedQT}). In almost all these results, the stretch depends, albeit very mildly, on the parameters of the separators. Namely, the stretch is generally $O(k\log\log t)$, where $t$ is the size of the recursive separators, the number of separator paths, or, in general graphs, simply $t=n$. An exception is the path-reporting spanner with stretch $1+\epsilon$ for $p$-path-separable graphs (see Theorem \ref{thm:SmallStretchImprovedQT}); however, the dependency on $p$ in the size and query time is linear (i.e., it is $\Omega(p)$).

In this section, we apply techniques from \cite{ES23} on our new path-reporting \textit{emulators}, to convert them into path-reporting $O(k)$-spanners, while only losing small factors in the size. Specifically, we use the following theorem from \cite{ES23}.

\begin{theorem}[\cite{ES23}, Theorem 3 and Remark 1] \label{thm:ESPreserver}
Given an undirected weighted graph $G=(V,E)$, an integer $\kappa\geq3$, a positive parameter $\epsilon\leq O(\log \kappa)$, and a set of pairs $\mathcal{P}\subseteq V^2$, there exists a path-reporting $\mathcal{P}$-pairwise $(1+\epsilon)$-spanner with query time $O(1)$ and size
\[O(|\mathcal{P}|\cdot\beta_{4/3}(\epsilon,\kappa)+n\log\kappa+n^{1+\frac{1}{\kappa}})~,\]
where $\beta_{4/3}(\epsilon,\kappa)=O(\frac{\log\kappa}{\epsilon})^{\lceil\log_{4/3}\kappa\rceil}$.

In addition, there exists a path-reporting $\mathcal{P}$-pairwise $(3+\epsilon)$-spanner with query time $O(1)$ and size
\[O(|\mathcal{P}|\cdot\kappa^{\log_{4/3}(12+\frac{40}{\epsilon})}+n\log\kappa+n^{1+\frac{1}{\kappa}})~.\]
\end{theorem}

In the next lemma, we apply the pairwise path-reporting spanners from Theorem \ref{thm:ESPreserver} on a path-reporting emulator, to obtain a path-reporting spanner (for \textit{all} queries $(u,v)\in V^2$).

\begin{lemma} \label{lemma:EmulatorToSpanner}
    Let $G=(V,E)$ be an undirected weighted $n$-vertex graph, and suppose that $G$ admits a path-reporting emulator $(H,D)$ with stretch $\alpha\geq1$, query time $q$ and size $M=\Omega(nt^{\frac{1}{k}})$, for some integer parameters $t,k>0$. Suppose also that $M=\Omega(n(\log k+\log\log n))$. Then, for every $\epsilon=O(\log k)$, the graph $G$ admits a path-reporting spanner with stretch $(1+\epsilon)\cdot\alpha$, query time $O(q)$ and size 
    \[O(M\cdot\beta_{4/3}(\epsilon,\kappa))~,\]
    where $\beta_{4/3}(\epsilon,\kappa)=O(\frac{\log\kappa}{\epsilon})^{\lceil\log_{4/3}\kappa\rceil}$ and $\kappa=\frac{k\log n}{\log t}$.
    
    In addition, $G$ admits a path-reporting spanner with stretch $O(\alpha)$, query time $O(q)$ and size 
    \[O(M\cdot\kappa^{O(1)})~.\]
\end{lemma}

\begin{proof}

Given the path-reporting emulator $(H,D)$, we construct a path-reporting spanner as follows. Consider the set $H\subseteq E$ as a set of pairs $\mathcal{P}=H\subseteq V^2$ (formally, we insert the pairs $(u,v),(v,u)$ to $\mathcal{P}$, for every $\{u,v\}\in H$). Let $(S,D_P)$ be the resulting path-reporting $\mathcal{P}$-pairwise $(1+\epsilon)$-spanner (resp., $(3+\epsilon)$-spanner) from Theorem \ref{thm:ESPreserver}. We define an oracle $D'$, that receives a query $(u,v)\in V^2$, and applies the oracle $D$ to obtain an estimate $\hat{d}(u,v)$. If a $u$-$v$ path is also required, the oracle $D'$ obtains a $u$-$v$ path $Q$ in $H$ (again by applying $D$), the applies the oracle $D_P$ on every edges of $Q$, and finally concatenates the resulting paths in $S$ into a single $u$-$v$ path in $S$.

It is immediate to verify that the query time of the resulting path-reporting spanner $(S,D')$ is $O(q)$. To analyze the stretch of $(S,D')$, let $(u,v)\in V^2$ be any query, and suppose that a $u$-$v$ path is required (otherwise, the stretch is at most $\alpha$, as in the path-reporting emulator $(H,D)$). Then, the $u$-$v$ path $Q$ obtained by $D$ has weight $w(Q)\leq\alpha\cdot d_G(u,v)$. Replacing every edge of $Q$ by a path in $S$ of stretch $1+\epsilon$ or $3+\epsilon$ results in a $u$-$v$ path in $S$ with stretch $1+\epsilon$ or $3+\epsilon$, respectively.

In the former case, where $(S,D_P)$ has stretch $1+\epsilon$, it has size $O(M\cdot\beta_{4/3}(\epsilon,\kappa)+n\log\kappa+n^{1+\frac{1}{\kappa}})$,
where $\beta_{4/3}(\epsilon,\kappa)=O(\frac{\log\kappa}{\epsilon})^{\lceil\log_{4/3}\kappa\rceil}$. We choose $\kappa=\frac{k\log n}{\log t}$ (for the given parameters $t,k>0$), and obtain a path-reporting spanner with size
\[O(M\cdot\beta_{4/3}(\epsilon,\kappa)+n\log\kappa+n^{1+\frac{1}{\kappa}})=O(M\cdot\beta_{4/3}(\epsilon,\kappa)+n\log\kappa+nt^{\frac{1}{k}})=O(M\cdot\beta_{4/3}(\epsilon,\kappa))~.\]
Here we used the assumptions that $M=\Omega(nt^{\frac{1}{k}})$ and $M=\Omega(n(\log k+\log\log n))$.

In the latter case, where $(S,D_P)$ has stretch $3+\epsilon$, it has size $O(M\cdot\kappa^{\log_{4/3}(12+\frac{40}{\epsilon})}+n\log\kappa+n^{1+\frac{1}{\kappa}})$. We again choose $\kappa=\frac{k\log n}{\log t}$, and a \textit{constant} $\epsilon$, and obtain a path-reporting spanner with size

\[O(M\cdot\kappa^{\log_{4/3}(12+\frac{40}{\epsilon})}+n\log\kappa+n^{1+\frac{1}{\kappa}})=O(M\cdot\kappa^{O(1)})~.\]
    
\end{proof}

Lemma \ref{lemma:EmulatorToSpanner}, when applied to our path-reporting emulators from Theorems \ref{thm:SmallTreewidthEmulatorAndSpanner} and \ref{thm:LargeStretchImprovedQT}, provides the following results for path-reporting spanners.

\begin{theorem} \label{thm:PRSpannersWithDecreasedStretch}
    Let $G=(V,E)$ be an undirected weighted $n$-vertex graph, let $k>0$ be any integer parameter and let $\epsilon\in(0,1]$ be any real parameter. In what follows, we denote $\beta_{4/3}(\epsilon,\kappa)=O(\frac{\log\kappa}{\epsilon})^{\lceil\log_{4/3}\kappa\rceil}$, for any parameter $\kappa$.

    \begin{enumerate}
        \item Suppose that $G$ is $s$-vertex-separable, for some non-decreasing function $s=s(\theta)$, and denote $S(n,k)=\sum_{i=0}^{\log n}s(\frac{n}{2^i})^{\frac{1}{k}}$. Then, $G$ admits a path-reporting spanner with stretch $(32e+\epsilon)k+1$ (resp., $O(k)$), query time $O(k\cdot S(n,k))$ and size $\beta_{4/3}(\epsilon,\kappa)\cdot n\cdot S(n,k)$ (resp., $O(\kappa^{O(1)}\cdot n\cdot S(n,k))$), where $\kappa=\frac{k\log n}{\log s(n/2)}$.
        
        In particular, if $s(\theta)=\theta^\delta$, for a constant $\delta\in(0,1]$, then $S(n,k)=O(\frac{k}{\delta}n^{\frac{\delta}{k}})$, and if the graph $G$ has treewidth $t(\theta)\equiv t$, then $S(n,k)=t^{\frac{1}{k}}\cdot\log n$.

        \item Suppose that $G$ is $(\ell,\pi)$-path-separable. Then, $G$ admits a path-reporting spanner with stretch $(32e+\epsilon)k+1$, query time 
        $\ell\cdot O\left(\log\pi\cdot\log^{3-\gamma}n+k\cdot(\frac{\pi}{\epsilon}\log\frac{\pi}{\epsilon})^{\frac{1}{k}}\cdot\log^{2-\gamma}n\cdot\log\log\pi\right)$
        and size $\beta_{4/3}(\epsilon,\kappa)\cdot(\frac{\pi}{\epsilon}\log\frac{\pi}{\epsilon})^{\frac{1}{k}}\cdot n\log^{3-\gamma}n)$, where $\kappa=\frac{k\log n}{\log\pi}$. Here, $\gamma=1$ if $G$ is also tree-like $(\ell,\pi)$-path-separable, and $\gamma=0$ otherwise.

        \item Suppose that $G$ is $(\ell,\pi)$-path-separable. Then, $G$ admits a path-reporting spanner with stretch $O(k)$, query time 
        $\ell\cdot O\left(\log\pi\cdot\log^{3-\gamma}n+k\cdot(\pi\log\pi)^{\frac{1}{k}}\cdot\log^{2-\gamma}n\cdot\log\log\pi\right)$
        and size $O(\kappa^{O(1)}\cdot(\pi\log\pi)^{\frac{1}{k}}\cdot n\log^{3-\gamma}n))$), where $\kappa=\frac{k\log n}{\log\pi}$. Here, $\gamma=1$ if $G$ is also tree-like $(\ell,\pi)$-path-separable, and $\gamma=0$ otherwise.
    \end{enumerate}
\end{theorem}

Recall that $K_r$-minor-free graphs are (weakly) tree-like $(\ell,\pi)$-path-separable, for $\ell=O(r^2)$ and $\pi=O(r^{c_{AG}})$, where $c_{AG}=4602$, and that graphs with bounded genus $g$ are strongly tree-like $O(g)$-path-separable (see Items (2) and (3) in Theorem \ref{thm:PathSeparableFamilies}). Thus, the following corollary is derived by Theorem \ref{thm:PRSpannersWithDecreasedStretch}.

\begin{corollary}
    Let $G=(V,E)$ be an undirected weighted $n$-vertex graph, let $k>0$ be any integer parameter and let $\epsilon\in(0,1]$ be any real parameter. In what follows, we denote $\beta_{4/3}(\epsilon,\kappa)=O(\frac{\log\kappa}{\epsilon})^{\lceil\log_{4/3}\kappa\rceil}$, for any parameter $\kappa$.

    \begin{enumerate}
        \item Suppose that $G$ is $K_r$-minor-free. Then, $G$ admits a path-reporting spanner with stretch $O(k)$, query time 
        $O\left(r^2\log r\cdot\log^2n+kr^{2+\frac{1}{k}}\cdot\log n\cdot\log\log r\right)$
        and size $O(\kappa^{O(1)}\cdot r^{2+\frac{1}{k}}\cdot n\log^2n)$, where $\kappa=\frac{k\log n}{\log r}$.

        \item Suppose that $G$ has bounded genus $g$. Then, $G$ admits a path-reporting spanner with stretch $(32e+\epsilon)k+1$, query time 
        $O\left(\log g\cdot\log^2n+k\cdot(\frac{g}{\epsilon}\log\frac{g}{\epsilon})^{\frac{1}{k}}\cdot\log n\cdot\log\log g\right)$
        and size $\beta_{4/3}(\epsilon,\kappa)\cdot(\frac{g}{\epsilon}\log\frac{g}{\epsilon})^{\frac{1}{k}}\cdot n\log^2n$, where $\kappa=\frac{k\log n}{\log g}$.

        \item Suppose that $G$ has bounded genus $g$. Then, $G$ admits a path-reporting spanner with stretch $O(k)$, query time 
        $O\left(\log g\cdot\log^2n+kg^{\frac{1}{k}}\cdot\log n\cdot\log\log g\right)$
        and size $O(\kappa^{O(1)}\cdot g^{\frac{1}{k}}\cdot n\log^2n)$), where $\kappa=\frac{k\log n}{\log g}$.
    \end{enumerate}
\end{corollary}

\subsection{Distance Oracles with Improved Stretch}

Recall that Theorem \ref{thm:SmallTreewidthEmulatorAndSpanner} provides a path-reporting emulator with stretch $32ek+1$, query time $O(k\cdot T(n,k))$ and size $O(k\cdot n\cdot T(n,k))$, for any $n$-vertex graph that is $t(\theta)$-vertex-separable, and any integer parameter $k\geq1$. Here, we provide an improved stretch of $4ek+1$. We note, however, that the structure that achieves this bound is not a path-reporting emulator, but only a \textit{distance oracle}. That is, given a query $(u,v)\in V^2$, this structure outputs only an estimate $\hat{d}(u,v)$ that satisfies $d_G(u,v)\leq\hat{d}(u,v)\leq(4ek+1)d_G(u,v)$, but not a path with weight $\hat{d}(u,v)$ in a fixed small-size emulator of $G$ (see Definition \ref{def:DistanceOracles}). The proof of the following theorem is in Appendix \ref{app:ImprovedStretch}.

\begin{theorem} \label{thm:SmallTreewidthDO}
Let $G$ be an $n$-vertex undirected weighted $s$-vertex-separable graph, for a non-decreasing function $s=s(\theta)$. Let $k\geq1$ be an integer, and denote $S(n,k)=\sum_{i=0}^{\log n}s\left(\frac{n}{2^i}\right)^{\frac{1}{k}}$. Then, $G$ has a distance oracle with stretch at most $4ek+1$, size $O(k\cdot n\cdot S(n,k))$, and query time $O(k\cdot S(n,k))$.
\end{theorem}

\clearpage
\pagestyle{empty}
\bibliography{hopset}

\appendix

\clearpage
\pagestyle{plain}
\pagenumbering{Roman}

\section{Comparison with Constructions for General Graphs} \label{app:ComparisonGeneralGraphs}

It is instructive to compare our HST covers for $s$-vertex-separable graphs (or graphs with treewidth $s=s(n)$) with the HST covers of \cite{MN06,NT12}. For a parameter $k=1,2,...$, our HST cover has stretch $12ek$ and size $O\left(k\cdot\sum_{i=0}^{\log n}s\left(\frac{n}{2^i}\right)^{\frac{1}{k}}\right)=O(k\cdot s^{\frac{1}{k}}\cdot\log n)$. The HST cover of \cite{NT12} for general graphs has the same stretch and size $O(kn^{\frac{1}{k}})$. Hence, for any constant $k$, our HST cover has smaller size as long as $s=n^{1-o(1)}$, i.e., our result provides a graceful generalization of the result of \cite{NT12}, parametrized by $s$. This is also the case for our non-spanning and spanning tree covers for this graph family. Their stretch is $O(k)$ and $O(k\log\log n)$, respectively, with the same constants hidden by the $O$-notation as in the respective state-of-the-art non-spanning and spanning tree covers for general graphs. Note that the tree covers of \cite{MN06,NT12,ACEFN20} are Ramsey ones (for general graphs), while our tree covers are not. Recall that a lower bound of \cite{BFN22,BLMN03} rules out existence of Ramsey tree covers with stretch $k$ and size $n^{o(\frac{1}{k\log k})}$ already for series-parallel graphs (that have treewidth $2$).

For strongly $p$-path-separable graphs, our HST cover has stretch $12e(1+\epsilon)k$ and size $O_\epsilon(kp^{\frac{1}{k}}\cdot\log^2n)$. For constant $k$, this bound improves the bound $O(kn^{\frac{1}{k}})$ for general graphs for $p\leq n^{1-\frac{1}{k+1}-o(1)}$, and this is also the case when comparing our non-spanning and spanning tree covers with the state-of-the-art bounds of \cite{NT12,Gupta01} and \cite{ACEFN20}, respectively.

For weakly $(\ell,\pi)$-path-separable graphs, our bounds leave larger room for improvement. Specifically, the size of our tree covers grows linearly with $\ell$, while ideally it should be proportional to $\ell^{\frac{1}{k}}$. Note, however, that the prime example of weakly $(\ell,\pi)$-path-separable graph is the class of $H$-minor-free graphs. The separator theorem of \cite{AG06} provides these graphs with a weakly tree-like $(\ell,\pi)$-path separator with $\ell\leq|E(H)|$ (in particular, $\ell\leq r^2$ for $H=K_r$) and $\pi\leq O(\ell\cdot RS(r)\cdot(\ell+RS(r)))=O(r^{c_{AG}})$ (the total number of paths in their separator is also at most $p=O(\ell\cdot RS(r)\cdot(\ell+RS(r)))$). Observe that for this graph class, $\ell=O(r^2)$ is much smaller than $\pi=O(r^{c_{AG}})$, $c_{AG}=4602$, and thus our bound of $O(k\ell\pi^{\frac{1}{k}}\cdot\log^2n)=O(kr^{2+\frac{1}{k}}\cdot\log^2n)$ on the size of our tree cover is much smaller (for $r=\Omega(\log n)$ and $k\geq2$) than the previously known bound of $O(p\cdot\frac{\log^2n}{\epsilon^2})=O(r^{c_{AG}}\cdot\frac{\log^2n}{\epsilon^2})$ on the size of the stretch-$(1+\epsilon)$ non-spanning tree covers for these graphs \cite{BFN22}.


\section{Aspect Ratio Reduction} \label{app:AspectRatioReduction}

Recall that the spanning forest covers from Lemma \ref{lemma:OneScaleSimpleTreeCover} and the covers from Corollary \ref{cor:EllPiOneScale} provide low stretch only for pairs of vertices $u,v\in V$ with $d_G(u,v)\approx\Delta$, for a given scale $\Delta>0$. In each of these results, constructing a cover for each relevant scale of the form $\Delta=2^i$, and unifying these covers, provides a cover with low stretch for \textit{all} pairs $u,v\in V$, at the cost of multiplying the size by $\log\Lambda$. For weighted graphs, the aspect ratio $\Lambda=\frac{\max_{u,v\in V}d_G(u,v)}{\min_{u\neq v\in V}d_G(u,v)}$ may be very large, thus the resulting size may be large as well.

In this appendix, we merge the covers for different scales in a more subtle way. See the discussion at the end of Section \ref{sec:TechnicalOverview} for some intuition. As a result, their size is only multiplied by a factor of $\log\frac{n}{\epsilon}$, rather than $\log\Lambda$. This improves the naive construction whenever $\Lambda=(\frac{n}{\epsilon})^{\omega(1)}$.
We start with (spanning and non-spanning) forest covers.

\pagebreak
\begin{lemma} \label{lemma:AspectRatioReductionForests}
    Let $G=(V,E)$ be an $n$-vertex undirected weighted graph, and let $\epsilon\in(0,1]$ and $\alpha\geq1$ be real parameters. Suppose that, for every $\Delta>0$, every minor $G^*=(V^*,E^*)$ of $G$ admits a forest cover $\mathcal{F}^*_{\Delta}$ with size at most $M(\epsilon,n)$, such that for every $u,v\in V^*$ with $d_{G^*}(u,v)\leq\Delta$ there is a forest $F\in\mathcal{F}^*_\Delta$ with $d_F(u,v)\leq\alpha\cdot(d_{G^*}(u,v)+\epsilon\Delta)$. Then, $G$ admits a forest cover $\mathcal{F}$ with size at most $M(\frac{\epsilon}{6},n)\cdot O(\log\frac{n}{\epsilon})$ and stretch $(1+\epsilon)\alpha$. Moreover, if for every minor $G^*$ and $\Delta$ the forest cover $\mathcal{F}^*_{\Delta}$ is spanning, then $\mathcal{F}$ is a spanning forest cover for $G$.
\end{lemma}

\begin{proof}

We assume that $\min_{u\neq v\in V}d_G(u,v)=\min_{e\in E}w(e)=1$, as otherwise we multiply the weight of each edge by $(\min_{u\neq v\in V}d_G(u,v))^{-1}$. Then, the aspect ratio $\Lambda$ is simply the diameter $\max_{u,v\in V}d_G(u,v)$.


Fix some $\Delta<\frac{n^2}{\epsilon}$, and consider the sequence 
\[\Delta_i=\Delta\cdot\left(\frac{n^2}{\epsilon}\right)^i~.\]
In the proof, we will show the existence of a $\mathcal{P}_\Delta$-pairwise forest cover $\mathcal{F}_\Delta$ for $G$  (see Definition \ref{def:PairwiseForestCover}), with size at most $M(\epsilon,n)$ and stretch $(1+6\epsilon)\alpha$, for the set of pairs
\[\mathcal{P}_\Delta=\left\{(u,v)\in V^2\;\bigg|\;\exists_{i\geq0}\;\frac{\Delta_i}{2}<d_G(u,v)\leq\Delta_i\right\}~.\]
As we will prove below, unifying these pairwise forest covers, for $O(\log\frac{n}{\epsilon})$ different values of $\Delta$, produces the desired forest cover for $G$.
To construct the pairwise forest cover $\mathcal{F}_\Delta$, we prove the following claim by induction over $h\geq0$.

\begin{claim} \label{claim:InductiveForestCover}
    Given a graph $G=(V,E)$ and a positive parameter $\Delta>0$, for every integer $h\geq0$ there is a $\mathcal{P}_{\Delta,h}$-pairwise forest cover $\mathcal{F}_{\Delta,h}$ for $G$, with size $M(\epsilon,n)$ and stretch $(1+6\epsilon)\alpha$, where
    \[\mathcal{P}_{\Delta,h}=\left\{(u,v)\in V^2\;\bigg|\;\frac{\Delta_i}{2}<d_G(u,v)\leq\Delta_i\text{ for some }\boldsymbol{0\leq i\leq h}\right\}~.\]
    Moreover, if for every minor $G^*$ and $\Delta$ the forest cover $\mathcal{F}^*_{\Delta}$ is spanning, then $\mathcal{F}_{\Delta,h}$ is a spanning forest cover too.
\end{claim}

\begin{proof}

For $h=0$, consider the sub-graph $G^*_{\Delta}$, that is obtained from $G$ by removing all edges of weight greater than $\Delta=\Delta_0$. This sub-graph (which is also a minor of $G$) has a forest cover $\mathcal{F}^*_{\Delta}$ of size at most $M(\epsilon,n)$, such that whenever $d_{G^*_{\Delta}}(u,v)\leq\Delta$, there is a forest $F\in\mathcal{F}^*_{\Delta}$ with $d_F(u,v)\leq\alpha\cdot(d_{G^*_{\Delta}}(u,v)+\epsilon\Delta)$. For $(u,v)\in\mathcal{P}_{\Delta,0}$ we have $\frac{\Delta}{2}<d_G(u,v)\leq\Delta$. In particular, the edges of any shortest $u$-$v$ path in $G$ are of weight at most $\Delta$, and therefore they exist in the sub-graph $G^*_{\Delta}$, and $d_{G^*_{\Delta}}(u,v)=d_G(u,v)\leq\Delta$. Hence, there is a forest $F\in\mathcal{F}^*_{\Delta}$ with
\[d_F(u,v)\leq\alpha\cdot(d_G(u,v)+\epsilon\Delta)<\alpha\cdot(d_G(u,v)+\epsilon\cdot2\cdot d_G(u,v))=(1+2\epsilon)\alpha\cdot d_G(u,v)~.\]
This shows that $\mathcal{F}^*_{\Delta}$ is a $\mathcal{P}_{\Delta,0}$-pairwise forest cover with size at most $M(\epsilon,n)$ and stretch $(1+2\epsilon)\alpha<(1+6\epsilon)\alpha$, thus proving the induction basis.

Next, fix some $h>0$. Consider the minor $G^*_{\Delta_h}$ that is obtained from $G$ by removing edges with a weight greater than $\Delta_h$, and contracting the edges with weight at most $\frac{\epsilon\Delta_h}{n^2}=\Delta_{h-1}$. For every vertex $v^*$ in $G^*_{\Delta_h}$, let $G[v^*]=(V[v^*],E[v^*])$ be a sub-graph of $G$, where $V[v^*]$ is the set of vertices in $V$ that are contracted into $v^*$ in $G^*_{\Delta_h}$, and $E[v^*]$ are the edges of $G$ with weight at most $\frac{\epsilon\Delta_h}{n^2}=\Delta_{h-1}$ between vertices of $V[v^*]$. Note that $G[v^*]$ is a connected graph, since its vertices and edges were all contracted to the same vertex $v^*$.

Let $\mathcal{F}^*_{\Delta_h}$ be a forest cover for $G^*_{\Delta_h}$, with size at most $M(\epsilon,n)$, such that if $u^*,v^*\in V(G^*_{\Delta_h})$ satisfy $d_{G^*_{\Delta_h}}(u^*,v^*)\leq\Delta_h$, then there is a forest $F\in\mathcal{F}^*_{\Delta_h}$ with $d_F(u^*,v^*)\leq\alpha\cdot(d_{G^*_{\Delta_h}}(u^*,v^*)+\epsilon\Delta_h)$.

For every vertex $v^*$, we apply the induction hypothesis on the graph $G[v^*]$, and obtain a $\mathcal{P}_{\Delta,h-1}$-pairwise forest cover $\mathcal{F}_{\Delta,h-1}[v^*]$ with size at most $M(\epsilon,n)$ and stretch at most $(1+6\epsilon)\alpha$. By Observation \ref{obs:FromForestCoverToTreeCover}, there is also a $\mathcal{P}_{\Delta,h-1}$-pairwise \textit{tree} cover $\mathcal{T}_{\Delta,h-1}[v^*]$ for $G[v^*]$, with the same size and stretch, which is spanning if $\mathcal{F}_{\Delta,h-1}[v^*]$ is spanning. The latter happens, in particular, if for every minor $G^*$ and $\Delta$ the forest cover $\mathcal{F}^*_\Delta$ is spanning, as implied by the induction hypothesis.

We now merge the forest cover $\mathcal{F}^*_{\Delta_h}$ (for the graph $G^*_{\Delta_h}$) with the pairwise tree covers $\mathcal{T}_{\Delta,h-1}[v^*]$ (for the graphs $G[v^*]$), for every $v^*\in V(G^*_{\Delta_h})$. Under some arbitrary enumeration of the forests in these covers, let $F^j$ be the $j$'th forest of $\mathcal{F}^*_{\Delta_h}$, and let $T^j[v^*]$ be the $j$'th tree of $\mathcal{T}_{\Delta,h-1}[v^*]$, for every vertex $v^*$ in $G^*_{\Delta_h}$. Here, since the sizes of $\mathcal{F}^*_{\Delta_h}$ and $\{\mathcal{T}_{\Delta,h-1}[v^*]\}_{v^*}$ are at most $M(\epsilon,n)$, the index $j$ is at most $M(\epsilon,n)$. If there are less than $j$ forests in $\mathcal{F}^*_{\Delta_h}$, we arbitrarily assign $F^j=F^1\in\mathcal{F}^*_{\Delta_h}$. Similarly, if for some $v^*$ the size of $\{\mathcal{T}_{\Delta,h-1}[v^*]\}_{v^*}$ is less than $j$, we define $T^j[v^*]=T^1[v^*]$.

Fix some index $j\leq M(\epsilon,n)$ and consider the forest $F^j\in\mathcal{F}^*_{\Delta_h}$. Recall that $F^j$ is not necessarily a spanning forest of the minor $G^*_{\Delta_h}$. That is, every edge $\{u^*,v^*\}$ of $F^j$ may exist in $G^*_{\Delta_h}$ or not. In the former case, this edge originates at some edge $\{u,v\}$ of $G$, with the same weight $w(u,v)=w(u^*,v^*)$, such that $u$ and $v$ were contracted to $u^*$ and $v^*$, respectively (i.e., $u\in V(G[u^*])$ and $v\in V(G[v^*])$). In the latter case, fix two arbitrary vertices $u\in V[u^*]$ and $v\in V[v^*]$ and define a new edge $\{u,v\}\in\binom{V}{2}$ with weight $d_G(u,v)$.
In both cases, we say that the edge $\{u^*,v^*\}$ \textit{corresponds} to the edge $\{u,v\}$, and vice versa.

We are now ready to define the merged forest $\bar{F}^j$. This forest consists of all the corresponding edges to the edges of $F^j$, and of the edges of the trees $T^j[v^*]$, for every $v^*\in V(G^*_{\Delta_h})$. Note that $\bar{F}^j$ is a forest with vertices in $V$, i.e., a (not necessarily spanning) forest of $G$. Moreover, if $\mathcal{F}^*_{\Delta_h}$ and $\{\mathcal{T}_{\Delta,h-1}[v^*]\}_{v^*}$ are spanning, then $\bar{F}^j$ \textit{is} a spanning forest of $G$. To observe that $\bar{F}^j$ is indeed a forest, note that it cannot contain a cycle in any of the trees $\{T^j[v^*]\}_{v^*}$, so if there is a cycle in $\bar{F}^j$, a subset of its edges corresponds to a cycle in $F^j$, in contradiction to $F^j$ being a forest. The desired forest cover $\mathcal{F}_{\Delta,h}$ is now defined as 
\[\mathcal{F}_{\Delta,h}=\{\bar{F}^j\;|\;1\leq j\leq M(\epsilon,n)\}~.\]

Fix some $(u,v)\in\mathcal{P}_{\Delta,h}$. To analyze the stretch of $\mathcal{F}_{\Delta,h}$ for the pair $(u,v)$, we consider two cases. First, if $(u,v)\in\mathcal{P}_{\Delta,h-1}\subseteq\mathcal{P}_{\Delta,h}$, then it means, in particular, that $d_G(u,v)\leq\Delta_{h-1}$, and therefore there is a shortest $u$-$v$ path in $G$ that consists only of edges with weight at most $\Delta_{h-1}=\frac{\epsilon\Delta_h}{n^2}$. In the graph $G^*_{\Delta_h}$, this entire path is contracted to the same vertex $v^*$. Thus, it is fully contained in the graph $G[v^*]$, and we have
\[\frac{\Delta_i}{2}<d_{G[v^*]}(u,v)=d_G(u,v)\leq\Delta_i~,\]
for some $0\leq i\leq h-1$. Since $\mathcal{T}_{\Delta,h-1}[v^*]$ is a $\mathcal{P}_{\Delta,h-1}$-pairwise tree cover with stretch $(1+6\epsilon)\alpha$, where $\mathcal{P}_{\Delta,h-1}$ is defined over the graph $G[v^*]$, we conclude that there is a tree $T^j[v^*]\in\mathcal{T}_{\Delta,h-1}[v^*]$ (where $T^j[v^*]$ is the $j$'th tree of $\mathcal{T}_{\Delta,h-1}[v^*]$) with 
\[d_{T^j[v^*]}(u,v)\leq(1+6\epsilon)\alpha\cdot d_{G[v^*]}(u,v)=(1+6\epsilon)\alpha\cdot d_{G}(u,v)~.\]
The forest $\bar{F}^j$ contains $T^j[v^*]$, and thus $d_{\bar{F}^j}(u,v)\leq(1+6\epsilon)\alpha\cdot d_{G}(u,v)$ as well.

The second case is that $\frac{\Delta_h}{2}<d_G(u,v)\leq\Delta_h$, i.e., $(u,v)\in\mathcal{P}_{\Delta,h}\setminus\mathcal{P}_{\Delta,h-1}$. Suppose that $u$ and $v$ were contracted to vertices $u^*$ and $v^*$, respectively, in the minor $G^*_{\Delta_h}$. Since $d_G(u,v)\leq\Delta_h$, every edge on the shortest $u$-$v$ path in $G$ is either contracted or present in $G^*_{\Delta_h}$, and therefore $d_{G^*_{\Delta_h}}(u^*,v^*)\leq d_G(u,v)\leq\Delta_h$. Thus, by our assumption on the forest cover $\mathcal{F}^*_{\Delta_h}$, there is a forest $F^j\in\mathcal{F}^*_{\Delta_h}$ (which is the $j$'th forest in $\mathcal{F}^*_{\Delta_h}$) such that
\begin{equation} \label{eq:OuterPath}
    d_{F^j}(u^*,v^*)\leq\alpha\cdot(d_{G^*_{\Delta_h}}(u^*,v^*)+\epsilon\Delta_h)\leq\alpha\cdot(d_G(u,v)+\epsilon\Delta_h)~.
\end{equation}

Consider the unique $u^*$-$v^*$ path $P^*=(u^*_0,u^*_1,u^*_2,...,u^*_l)$ in $F^j$, which has weight $w(P^*)=d_{F^j}(u^*,v^*)\leq\alpha\cdot(d_G(u,v)+\epsilon\Delta_h)$. For every $1\leq t\leq l$, the edge $\{u^*_{t-1},u^*_t\}$ corresponds to some edge $\{x_{t-1},y_t\}\in\binom{V}{2}$, where $x_{t-1}\in V[u^*_{t-1}]$ and $y_t\in V[u^*_t]$. Denoting $y_0=u$ and $x_l=v$, we can say that $x_t,y_t\in V[u^*_t]$ for every $0\leq t\leq l$. For every $0\leq t\leq l$, recall that $T^j[u^*_t]$ is a tree in the tree cover $\mathcal{T}_{\Delta,h-1}[u^*_t]$. Thus, it contains a path $P_t$ between $y_t$ and $x_t$, such that every edge $\{a,b\}$ on $P_t$ is an edge of $G[u^*_t]$, or has weight $d_{G[u^*_t]}(a,b)$. In both cases, the weight of the edge $\{a,b\}$ is the weight of some simple path in $G[u^*_t]$, and thus bounded by $n\cdot\frac{\epsilon\Delta_h}{n^2}=\frac{\epsilon\Delta_h}{n}$ (since, by definition, all edges in $G[u^*_t]$ have weight at most $\frac{\epsilon\Delta_h}{n^2}$). This proves that the weight of the path $P_t$ is at most $|P_t|\cdot\frac{\epsilon\Delta_h}{n}$.

Concatenating the paths $P_t$, for every $0\leq t\leq l$, with the edges $\{x_{t-1},y_t\}$, for every $1\leq t\leq l$, we obtain a $u$-$v$ path $P_{u,v}$ in $\bar{F}^j$. To bound its weight, we analyze the weight of an edge $\{x_{t-1},y_t\}$, for some $1\leq t\leq l$. Recall that this edge corresponds to the edge $\{u^*_{t-1},u^*_t\}$ of $F^j$. If the forest $F^j$ is spanning, then $\{x_{t-1},y_t\}$ is an edge of $G^*_{\Delta_h}$, and we have $w(x_{t-1},y_t)=w(u^*_{t-1},u^*_t)$. Otherwise, this weight was defined as $w(x_{t-1},y_t)=d_G(x_{t-1},y_t)$. Consider an $x_{t-1}$-$y_t$ simple path $Q_t$ in $G$ that follows the edges that correspond to the edges in the shortest $u^*_{t-1}$-$u^*_t$ path in $G^*_{\Delta_h}$, and connects any two of them with an arbitrary simple path in an appropriate graph $G[s^*]$. These simple paths have a total of at most $n-1$ edges (of $G$), each of weight at most $\frac{\epsilon\Delta_h}{n^2}$, by the definition of the graphs $G[s^*]$. Thus, $w(Q_t)\leq d_{G^*_{\Delta_h}}(u^*_{t-1},u^*_t)+(n-1)\cdot\frac{\epsilon\Delta_h}{n^2}<d_{G^*_{\Delta_h}}(u^*_{t-1},u^*_t)+\frac{\epsilon\Delta_h}{n}$, and therefore the weight of $\{x_{t-1},y_t\}$ is $w(x_{t-1},y_t)=d_G(x_{t-1},y_t)<d_{G^*_{\Delta_h}}(u^*_{t-1},u^*_t)+\frac{\epsilon\Delta_h}{n}$. Since in non-spanning forest covers, the weight of each edge is at least the distance between its endpoints, we conclude that $w(x_{t-1},y_t)<w(u^*_{t-1},u^*_t)+\frac{\epsilon\Delta_h}{n}$.

In both cases above, of $F^j$ being spanning or non-spanning, we obtained the bound $w(x_{t-1},y_t)<w(u^*_{t-1},u^*_t)+\frac{\epsilon\Delta_h}{n}$ on the weight of every edge of the form $\{x_{t-1},y_t\}$, $1\leq t\leq l$, that participates in the path $P_{u,v}$ in the forest $\bar{F}^j$. Using this bound, we get
\begin{eqnarray*}
    w(P_{u,v})&=&\sum_{t=1}^l w(x_{t-1},y_t)+\sum_{t=0}^l w(P_t)
    <\sum_{t=1}^l(w(u^*_{t-1},u^*_t)+\frac{\epsilon\Delta_h}{n})+
    \frac{\epsilon\Delta_h}{n}\cdot\sum_{t=0}^l|P_t|\\
    &\leq&d_{F^j}(u^*,v^*)+\epsilon\Delta_h+\epsilon\Delta_h\\
    &\stackrel{(\ref{eq:OuterPath})}{\leq}&\alpha\cdot(d_G(u,v)+\epsilon\Delta_h)+2\epsilon\Delta_h
    \leq\alpha\cdot(d_G(u,v)+3\epsilon\Delta_h)\\
    &\leq&\alpha\cdot(d_G(u,v)+6\epsilon\cdot d_G(u,v))
    =(1+6\epsilon)\alpha\cdot d_G(u,v)~.
\end{eqnarray*}

In conclusion, we saw that whenever $(u,v)\in\mathcal{P}_{\Delta,h}$, the forest cover $\mathcal{F}_{\Delta,h}$, that has size at most $M(\epsilon,n)$, has a forest $\bar{F}^j\in\mathcal{F}_{\Delta,h}$ with $d_{\bar{F}^j}(u,v)\leq(1+6\epsilon)\alpha\cdot d_G(u,v)$. We conclude that $\mathcal{F}_{\Delta,h}$ is a $\mathcal{P}_{\Delta,h}$-pairwise forest cover with stretch $(1+6\epsilon)\alpha$ and size at most $M(\epsilon,n)$.

\end{proof}

Next, we note that for $h=\left\lceil\log_{n^2/\epsilon}\left(\frac{\Lambda}{\Delta}\right)\right\rceil$, every $u,v\in V$ have $d_G(u,v)\leq\Lambda\leq\Delta\cdot\left(\frac{n^2}{\epsilon}\right)^h=\Delta_h$, and therefore $\mathcal{P}_{\Delta,h}=\mathcal{P}_{\Delta}=\left\{(u,v)\in V^2\;\bigg|\;\exists_{i\geq0}\;\frac{\Delta_i}{2}<d_G(u,v)\leq\Delta_i\right\}$. By Claim \ref{claim:InductiveForestCover}, there is a $\mathcal{P}_{\Delta}$-pairwise forest cover $\mathcal{F}_{\Delta}$ with size at most $M(\epsilon,n)$ and stretch $(1+6\epsilon)\alpha$. Moreover, by Claim \ref{claim:InductiveForestCover}, this forest cover is spanning if $\mathcal{F}^*_{\Delta_h}$ is spanning.

We define
\[\mathcal{F}=\bigcup_{t=0}^{\lceil\log_2(\frac{n^2}{\epsilon})\rceil}\mathcal{F}_{2^t}~,\]
where $\mathcal{F}_{2^t}$ is the pairwise forest cover $\mathcal{F}_{\Delta}$ with $\Delta=2^t$.
Since the size of each $\mathcal{F}_{\Delta}$ is at most $M(\epsilon,n)$, the size of the forest cover $\mathcal{F}$ is $M(\epsilon,n)\cdot O(\log\frac{n^2}{\epsilon})=M(\epsilon,n)\cdot O(\log\frac{n}{\epsilon})$.

To analyze the stretch of $\mathcal{F}$, fix some $u,v\in V$. There is a unique index $h\geq0$ such that $(\frac{n^2}{\epsilon})^h<d_G(u,v)\leq(\frac{n^2}{\epsilon})^{h+1}$. Then, there is a unique index $0\leq t\leq\lceil\log_2(\frac{n^2}{\epsilon})\rceil-1$ such that $2^t\cdot(\frac{n^2}{\epsilon})^h<d_G(u,v)\leq2^{t+1}\cdot(\frac{n^2}{\epsilon})^h$. Denoting $\Delta=2^{t+1}$, we have
\[\frac{\Delta_h}{2}=\frac{\Delta}{2}\cdot\left(\frac{n^2}{\epsilon}\right)^h<d_G(u,v)\leq\Delta\cdot\left(\frac{n^2}{\epsilon}\right)^h=\Delta_h~.\]
Thus, $(u,v)\in\mathcal{P}_{\Delta}$, and there is a forest $F\in\mathcal{F}_{\Delta}\subseteq\mathcal{F}$ with $d_F(u,v)\leq(1+6\epsilon)\alpha\cdot d_G(u,v)$. 

We conclude that $\mathcal{F}$ has stretch at most $(1+6\epsilon)\alpha$ and size at most $M(\epsilon,n)\cdot O(\log\frac{n}{\epsilon})$, and that it is a spanning forest cover if $\mathcal{F}^*_\Delta$ is a spanning forest cover for every minor $G^*$ and every $\Delta>0$. Replacing $\epsilon$ by $\frac{\epsilon}{6}$, we obtain the desired result.
    
\end{proof}

Next, we prove an analogous lemma to Lemma \ref{lemma:AspectRatioReductionForests} for HST covers.

\begin{lemma} \label{lemma:AspectRatioReductionHSTs}
    Let $G=(V,E)$ be an $n$-vertex undirected weighted graph, and let $\epsilon\in(0,1]$ and $\alpha\geq1$ be real parameters. Suppose that, for every $\Delta>0$, every minor $G^*=(V^*,E^*)$ of $G$ admits an HST cover $\mathcal{H}^*_{\Delta}$ with size at most $M(\epsilon,n)$, such that for every $u,v\in V^*$ with $d_{G^*}(u,v)\leq\Delta$ there is an HST $T\in\mathcal{H}^*_\Delta$ with $\rho_T(u,v)\leq\alpha\cdot(d_{G^*}(u,v)+\epsilon\Delta)$. Then, $G$ admits an HST cover $\mathcal{H}$ with size at most $M(\frac{\epsilon}{4},n)\cdot O(\frac{\log n}{\epsilon})$ and stretch $(1+\epsilon)\alpha$.
\end{lemma}

\begin{proof}

As the proof of Lemma \ref{lemma:AspectRatioReductionForests}, we assume w.l.o.g. that $\min_{u\neq v\in V}d_G(u,v)=1$, and thus $\Lambda=diam(G)=\max_{u,v\in V}d_G(u,v)$. 

Fix some $\Delta<\frac{n}{\epsilon}$, and consider the sequence 
\[\Delta_i=\Delta\cdot\left(\frac{n}{\epsilon}\right)^i~.\]
This sequence differs from the one in the proof of Lemma \ref{lemma:AspectRatioReductionForests} by replacing $n^2$ by $n$. We again define
\begin{eqnarray*}
    \mathcal{P}_\Delta&=&\left\{(u,v)\in V^2\;\bigg|\;\exists_{i\geq0}\;\frac{\Delta_i}{2}<d_G(u,v)\leq\Delta_i\right\}~,\\
    \mathcal{P}_{\Delta,h}&=&\left\{(u,v)\in V^2\;\bigg|\;\frac{\Delta_i}{2}<d_G(u,v)\leq\Delta_i\text{ for some }\boldsymbol{0\leq i\leq h}\right\}~.
\end{eqnarray*}
Note that $\mathcal{P}_\Delta=\mathcal{P}_{\Delta,h}$ for $h=\left\lceil\log_{n/\epsilon}\left(\frac{\Lambda}{\Delta}\right)\right\rceil$. Thus, to obtain a $\mathcal{P}_\Delta$-pairwise HST cover for $G$, it is enough to prove the following claim (by induction over $h\geq0$).

\begin{claim} \label{claim:InductiveHSTCover}
    Given a graph $G=(V,E)$ and a positive parameter $\Delta>0$, for every integer $h\geq0$ there is a $\mathcal{P}_{\Delta,h}$-pairwise HST cover $\mathcal{H}_{\Delta,h}$ for $G$, with size $M(\epsilon,n)$ and stretch $(1+4\epsilon)\alpha$.
\end{claim}

\begin{proof}

For $h=0$, consider the sub-graph $G^*_{\Delta}$, that is obtained from $G$ by removing all edges of weight greater than $\Delta=\Delta_0$. This sub-graph (which is also a minor of $G$) has an HST cover $\mathcal{H}^*_{\Delta}$ of size $M(\epsilon,n)$, such that whenever $d_{G^*_{\Delta}}(u,v)\leq\Delta$, there is an HST $T\in\mathcal{H}^*_{\Delta}$ with $\rho_T(u,v)\leq\alpha\cdot(d_{G^*_{\Delta}}(u,v)+\epsilon\Delta)$. For $(u,v)\in\mathcal{P}_{\Delta,0}$ we have $\frac{\Delta}{2}<d_G(u,v)\leq\Delta$. In particular, the edges of any shortest $u$-$v$ path in $G$ are of weight at most $\Delta$, and therefore they exist in the sub-graph $G^*_{\Delta}$, and $d_{G^*_{\Delta}}(u,v)=d_G(u,v)\leq\Delta$. Hence, there is an HST $T\in\mathcal{H}^*_{\Delta}$ with
\[\rho_T(u,v)\leq\alpha\cdot(d_G(u,v)+\epsilon\Delta)<(1+2\epsilon)\alpha\cdot d_G(u,v)<(1+4\epsilon)\alpha\cdot d_G(u,v)~.\]
This shows that $\mathcal{H}^*_{\Delta}$ is a $\mathcal{P}_{\Delta,0}$-pairwise HST cover with size $M(\epsilon,n)$ and stretch $(1+4\epsilon)\alpha$, thus concludes the induction basis.

Next, fix some $h>0$. Consider the minor $G^*_{\Delta_h}$ that is obtained from $G$ by removing edges with a weight greater than $\Delta_h$, and contracting the edges with weight at most $\frac{\epsilon\Delta_h}{n}$. Let $\mathcal{H}^*_{\Delta_h}$ be an HST cover for $G^*_{\Delta_h}$, with size $M(\epsilon,n)$, such that if $u^*,v^*\in V(G^*_{\Delta_h})$ satisfy $d_{G^*_{\Delta_h}}(u^*,v^*)\leq\Delta_h$, then there is an HST $T\in\mathcal{H}^*_{\Delta_h}$ with $\rho_T(u^*,v^*)\leq\alpha\cdot(d_{G^*_{\Delta_h}}(u^*,v^*)+\epsilon\Delta_h)$.

For every vertex $v^*$ in $G^*_{\Delta_h}$, we denote by $G[v^*]$ the induced graph of $G$ on the collection of vertices in $V$ that are contracted into $v^*$ in $G^*_{\Delta_h}$. Note that every two vertices in $G[v^*]$ have a path between them, with edges of weight at most $\frac{\epsilon\Delta_h}{n}=\Delta_{h-1}$. For this graph, we apply the induction hypothesis, and obtain a $\mathcal{P}_{\Delta,h-1}$-pairwise HST cover $\mathcal{H}_{\Delta,h-1}[v^*]$ with size $M(\epsilon,n)$ and stretch $(1+4\epsilon)\alpha$.

We now merge the HST cover $\mathcal{H}^*_{\Delta_h}$ (for the graph $G^*_{\Delta_h}$) with the pairwise HST covers $\mathcal{H}_{\Delta,h-1}[v^*]$ (for the graph $G[v^*]$), for every $v^*\in V(G^*_{\Delta_h})$. Under some arbitrary enumeration of the HSTs in these HST covers, let $T^j$ be the $j$'th HST of $\mathcal{H}^*_{\Delta_h}$, and let $T^j[v^*]$ be the $j$'th HST of $\mathcal{H}_{\Delta,h-1}[v^*]$, for every vertex $v^*$ in $G^*_{\Delta_h}$. Here, since the sizes of $\mathcal{H}^*_{\Delta_h}$ and $\{\mathcal{H}_{\Delta,h-1}[v^*]\}_{v^*}$ are at most $M(\epsilon,n)$, the index $j$ is at most $M(\epsilon,n)$. If there are less than $j$ HSTs in $\mathcal{H}^*_{\Delta_h}$, we arbitrarily assign $T^j=T^1\in\mathcal{H}^*_{\Delta_h}$. Similarly, if for some $v^*$ the size of $\{\mathcal{H}_{\Delta,h-1}[v^*]\}_{v^*}$ is less than $j$, we define $T^j[v^*]=T^1[v^*]$.

Recall that the vertices $v^*\in V(G^*_{\Delta_h})$ are exactly the leaves of the HST $T^j$. We replace each leaf $v^*$ in $T^j$ by $r^j[v^*]$ -- the root of the HST $T^j[v^*]$. In addition, we add an additive factor of $\epsilon\Delta_h$ to the label of every non-leaf vertex of $T^j$. Denote by $\bar{T}^j$ the resulting tree, and let $\mathcal{H}_{\Delta,h}=\{\bar{T}^1,\bar{T}^2,\bar{T}^3,...\}$. We next prove that each of the trees in $\mathcal{H}_{\Delta,h}$ are HSTs, and thus $\mathcal{H}_{\Delta,h}$ is an HST cover with size at most $M(\epsilon,n)$.

To verify that $\bar{T}^j\in\mathcal{H}_{\Delta,h}$ is an HST, for any $j\leq M(\epsilon,n)$, it is enough to show that the label of $r^j[v^*]$, for any $v^*\in V(G^*_{\Delta_h})$, is not larger than the labels of all the inner vertices in $T^j$, to which we added $\epsilon\Delta_h$. Fix $j$ and $v^*$. Note that without loss of generality, we may assume that the label of $r^j[v^*]$ is at most $diam(G[v^*])$, as otherwise we can replace it by $diam(G[v^*])$ without damaging the HST property or the approximation of the HST $T^j[v^*]$. Thus, it is enough to show that $diam(G[v^*])\leq\epsilon\Delta_h$. Indeed, for every two vertices $u,v$ in $G[v^*]$, since they were contracted to the same vertex $v^*$ of $G^*_{\Delta_h}$, there is a $u$-$v$ path in $G[v^*]$ with edges of weight at most $\frac{\epsilon\Delta_h}{n}$ each. In particular, $diam(G[v^*])<n\cdot\frac{\epsilon\Delta_h}{n}=\epsilon\Delta_h$. This proves that $\bar{T}^j$ is indeed an HST.

Fix some $(u,v)\in\mathcal{P}_{\Delta,h}$. We consider two cases. First, if $(u,v)\in\mathcal{P}_{\Delta,h-1}\subseteq\mathcal{P}_{\Delta,h}$, then it means, in particular, that $d_G(u,v)\leq\Delta_{h-1}$, and therefore there is a shortest $u$-$v$ path in $G$ that consists only of edges with weight at most $\Delta_{h-1}=\frac{\epsilon\Delta_h}{n}$. In the graph $G^*_{\Delta_h}$, this entire path is contracted to the same vertex $v^*$. Thus, $u$ and $v$ are both in the graph $G[v^*]$, and moreover, we have
\[\frac{\Delta_i}{2}<d_{G[v^*]}(u,v)=d_G(u,v)\leq\Delta_i~,\]
for some $0\leq i\leq h-1$. Since $\mathcal{H}_{\Delta,h-1}[v^*]$ is a $\mathcal{P}_{\Delta,h-1}$-pairwise HST cover with stretch $(1+4\epsilon)\alpha$, where $\mathcal{P}_{\Delta,h-1}$ is defined over the graph $G[v^*]$, we conclude that there is an HST $T^j[v^*]\in\mathcal{H}_{\Delta,h-1}[v^*]$ (where $T^j[v^*]$ is the $j$'th HST of $\mathcal{H}_{\Delta,h-1}[v^*]$) with 
\[\rho_{T^j[v^*]}(u,v)\leq(1+4\epsilon)\alpha\cdot d_{G[v^*]}(u,v)=(1+4\epsilon)\alpha\cdot d_{G}(u,v)~.\]
In the HST $\bar{T}^j$, the vertices $u,v$ are leaves of the same sub-tree, which is a copy of $T^j[v^*]$, thus $\rho_{\bar{T}^j}(u,v)\leq(1+4\epsilon)\alpha\cdot d_{G}(u,v)$ as well.

The second case is when $\frac{\Delta_h}{2}<d_G(u,v)\leq\Delta_h$, i.e., $(u,v)\in\mathcal{P}_{\Delta,h}\setminus\mathcal{P}_{\Delta,h-1}$. Suppose that $u$ and $v$ were contracted to vertices $u^*$ and $v^*$, respectively, in the graph $G^*_{\Delta_h}$. Since $d_G(u,v)\leq\Delta_h$, every edge on the shortest $u$-$v$ path in $G$ is either contracted or present in $G^*_{\Delta_h}$, and therefore $d_{G^*_{\Delta_h}}(u^*,v^*)\leq d_G(u,v)\leq\Delta_h$. Thus, by our assumption on the HST cover $\mathcal{H}^*_{\Delta_h}$, there is an HST $T^j\in\mathcal{H}^*_{\Delta_h}$ (which is the $j$'th HST in $\mathcal{H}^*_{\Delta_h}$) such that
\[\rho_{T^j}(u^*,v^*)\leq\alpha\cdot(d_{G^*_{\Delta_h}}(u^*,v^*)+\epsilon\Delta_h)\leq\alpha\cdot(d_G(u,v)+\epsilon\Delta_h)~.\]
Recall that $\rho_{T^j}(u^*,v^*)$ is the label of the LCA of $u^*$ and $v^*$ in $T^j$, and that the HST $\bar{T}^j$ contains a copy of $T^j$, with labels increased by $\epsilon\Delta_h$. Also, the leaves of $T^j$ are replaced by some sub-trees. Specifically, the sub-trees that replace the leaves $u^*,v^*$ contain $u,v$, respectively, and therefore 
\begin{eqnarray*}
    \rho_{\bar{T}^j}(u,v)&\leq&\rho_{T^j}(u^*,v^*)+\epsilon\Delta_h
    \leq\alpha\cdot(d_G(u,v)+\epsilon\Delta_h)+\epsilon\Delta_h\\
    &\leq&\alpha\cdot(d_G(u,v)+2\epsilon\Delta_h)<\alpha\cdot(d_G(u,v)+4\epsilon\cdot d_G(u,v))\\
    &=&(1+4\epsilon)\alpha\cdot d_G(u,v)~.
\end{eqnarray*}

In both cases we saw that whenever $(u,v)\in\mathcal{P}_{\Delta,h}$, the HST cover $\mathcal{H}_{\Delta,h}$, that has size at most $M(\epsilon,n)$, has an HST $\bar{T}^j\in\mathcal{H}_{\Delta,h}$ with $\rho_{\bar{T}^j}(u,v)\leq(1+4\epsilon)\alpha\cdot d_G(u,v)$. We conclude that $\mathcal{H}_{\Delta,h}$ is a $\mathcal{P}_{\Delta,h}$-pairwise HST cover with stretch $(1+4\epsilon)\alpha$ and size at most $M(\epsilon,n)$.

\end{proof}

For $h=\left\lceil\log_{n/\epsilon}\left(\frac{\Lambda}{\Delta}\right)\right\rceil$ we have $\mathcal{P}_{\Delta,h}=\mathcal{P}_{\Delta}$, and therefore Claim \ref{claim:InductiveHSTCover} implies that there is a $\mathcal{P}_{\Delta}$-pairwise HST cover $\mathcal{H}_{\Delta}$ with size at most $M(\epsilon,n)$ and stretch $(1+4\epsilon)\alpha$. We define
\[\mathcal{H}=\bigcup_{t=0}^{\lceil\log_2(\frac{n}{\epsilon})\rceil}\mathcal{H}_{2^t}~,\]
where $\mathcal{H}_{2^t}$ is the pairwise HST cover $\mathcal{H}_{\Delta}$ for $\Delta=2^t$.
The size of $\mathcal{H}$ is $M(\epsilon,n)\cdot O(\log\frac{n}{\epsilon})$.

For the stretch analysis, fix some $u,v\in V$. There are unique indices $h\geq0$ and $0\leq t\leq \lceil\log_2(\frac{n}{\epsilon})\rceil-1$ such that $2^t\cdot(\frac{n}{\epsilon})^h<d_G(u,v)\leq2^{t+1}\cdot(\frac{n}{\epsilon})^h$. Denoting $\Delta=2^{t+1}$, we have
\[\frac{\Delta_h}{2}=\frac{\Delta}{2}\cdot\left(\frac{n}{\epsilon}\right)^h<d_G(u,v)\leq\Delta\cdot\left(\frac{n}{\epsilon}\right)^h=\Delta_h~.\]
Thus, $(u,v)\in\mathcal{P}_{\Delta}$, and there is an HST $T\in\mathcal{H}_{\Delta}\subseteq\mathcal{H}$ with $\rho_T(u,v)\leq(1+4\epsilon)\alpha\cdot d_G(u,v)$. 

We conclude that $\mathcal{H}$ has stretch at most $(1+4\epsilon)\alpha$ and size at most $M(\epsilon,n)\cdot O(\log\frac{n}{\epsilon})$. Replacing $\epsilon$ by $\frac{\epsilon}{4}$, we obtain the desired result.
    
\end{proof}

\subsection{Completing the Proof of Theorem \ref{thm:SmallStretchImprovedQT}} \label{app:PRSpannerAspectRatioReduction1}

Recall that in Section \ref{sec:SpannersAndEmulators} we proved a weaker version of Theorem \ref{thm:SmallStretchImprovedQT}, where the size of the resulting path-reporting spanner, for $n$-vertex $p$-path-separable graphs, is $O(n\cdot p\cdot\frac{\log n\cdot\log\Lambda}{\epsilon})$, instead of $O(n\cdot p\cdot\frac{\log^2n}{\epsilon})$. In this section, we show that the latter can be achieved by essentially the same proof as of Theorem \ref{thm:SmallStretchImprovedQT}, by replacing the pairwise forest covers from Lemma \ref{lemma:OneScaleSimpleTreeCover} with the pairwise forest covers $\mathcal{F}_\Delta$ from the proof of Lemma \ref{lemma:AspectRatioReductionForests} in this appendix.

\begin{proof}[Proof of Theorem \ref{thm:SmallStretchImprovedQT}]

We assume that $\min_{u\neq v\in V}d_G(u,v)=\min_{e\in E}w(e)=1$, thus $\Lambda=diam(G)$. We also assume that the graph $G$ is connected, as otherwise we construct a path-reporting spanner for every connected component, and their union is the desired path-reporting spanner for $G$. In addition, we assume that $\epsilon>\frac{1}{n}$, as otherwise the desired size of the path-reporting spanner is at least $n^2$, and such path-reporting spanner is trivial.

Let $(S_0,D_0)$ be the path-reporting $3$-spanner that is obtained by applying Lemma \ref{lemma:TreeCoverToPREmulator} on the tree cover in Theorem \ref{thm:GKR} (by \cite{GKR04}) with $\mathcal{P}=V^2$. It has query time $O(p\cdot\log n)$ and size $O(n\cdot p\cdot\log n)$. 

For every $\Delta\in\{2^t\}_{t=1}^{\lceil\log_2\frac{n^2}{\epsilon}\rceil}$, let $\mathcal{F}_\Delta$ be the $\mathcal{P}_\Delta$-pairwise forest cover $\mathcal{F}_{\Delta}$ from the proof of Lemma \ref{lemma:AspectRatioReductionForests}. Since the graph $G$ is $p$-path-separable, it satisfies the condition of this lemma with $\alpha=1+\epsilon$ and $M(\epsilon,n)=O(p\cdot\frac{\log n}{\epsilon})$, by Lemma \ref{lemma:OneScaleSimpleTreeCover}. Thus, $\mathcal{F}_{\Delta}$ is a $\mathcal{P}_{\Delta}$-pairwise forest cover with stretch $(1+6\epsilon)(1+\epsilon)\leq1+13\epsilon$ and size $M(\epsilon,n)=O(p\cdot\frac{\log n}{\epsilon})$. Recall that $\mathcal{P}_\Delta$ is defined by
\[\mathcal{P}_\Delta=\left\{(u,v)\in V^2\;\bigg|\;\exists_{i\geq0}\;\frac{\Delta}{2}\cdot\left(\frac{n^2}{\epsilon}\right)^i<d_G(u,v)\leq\Delta\cdot\left(\frac{n^2}{\epsilon}\right)^i\right\}~.\]

By Observation \ref{obs:FromForestCoverToTreeCover}, for every $t$ there is also a $\mathcal{P}_{2^t}$-pairwise tree cover $\mathcal{T}_{2^t}$ for $G$ with stretch $1+13\epsilon$ and size $O(p\cdot\frac{\log n}{\epsilon})$.

We apply Lemma \ref{lemma:TreeCoverToPREmulator} on $\mathcal{T}_{2^t}$, and obtain a $\mathcal{P}_{2^t}$-pairwise path-reporting $(1+13\epsilon)$-spanner $(S_t,D_t)$ for $G$ with query time $O(p\cdot\frac{\log n}{\epsilon})$ and size $O(n\cdot p\cdot\frac{\log n}{\epsilon})$. In our new path-reporting spanner $(S,D)$, we define $S=\bigcup_{t=1}^{\lceil\log_2\frac{n^2}{\epsilon}\rceil}S_t$, and we store all the oracles $\{D_t\}_{t=0}^{\lceil\log_2\frac{n^2}{\epsilon}\rceil}$, and the oracle $D_0$ in $D$.

Given a query $(u,v)\in V^2$, the oracle $D$ first applies $D_0$ on $(u,v)$, to get an estimate $\hat{d}_0(u,v)\in[d_G(u,v),3d_G(u,v)]$, i.e., $\frac{1}{3}\hat{d}_0(u,v)\leq d_G(u,v)\leq\hat{d}_0(u,v)$, within time $O(p\cdot\log n)$. Then, the oracle $D$ finds all the values $t\in\{1,2,...,\lceil\log_2\frac{n^2}{\epsilon}\rceil\}$ and $i\geq0$ such that the interval $[\frac{1}{3}\hat{d}_0(u,v),\hat{d}_0(u,v)]$ intersects the interval $\mathcal{I}_{t,i}=\left(2^{t-1}\cdot\left(\frac{n^2}{\epsilon}\right)^i,2^t\cdot\left(\frac{n^2}{\epsilon}\right)^i\right]$. We claim that there are at most $2$ values of $i$ such that $[\frac{1}{3}\hat{d}_0(u,v),\hat{d}_0(u,v)]\cap\mathcal{I}_{t,i}\neq\emptyset$ (for $t$ in the relevant range), and that for every such $i$ there are at most $3$ possible values of $t$. Indeed, if $x\in\mathcal{I}_{t,i}$ and $y\in\mathcal{I}_{s,j}$ are both in $[\frac{1}{3}\hat{d}_0(u,v),\hat{d}_0(u,v)]$, for some $i<j$ and $1\leq s,t\leq\lceil\log_2\frac{n^2}{\epsilon}\rceil$, then 
\[\left(\frac{n^2}{\epsilon}\right)^j\leq2^{s-1}\cdot\left(\frac{n^2}{\epsilon}\right)^j<y\leq\hat{d}_0(u,v)\leq3x\leq3\cdot2^t\cdot\left(\frac{n^2}{\epsilon}\right)^i<6\left(\frac{n^2}{\epsilon}\right)^{i+1}<\left(\frac{n^2}{\epsilon}\right)^{i+2}~,\]
i.e., $j\leq i+1$ (here we assumed that $\frac{n^2}{\epsilon}>6$, which is true, in particular, for every $n>2$). Also, if $i=j$, we get from the same inequality that $2^{s-1}\cdot\left(\frac{n^2}{\epsilon}\right)^i<3\cdot2^t\cdot\left(\frac{n^2}{\epsilon}\right)^i<2^{t+2}\cdot\left(\frac{n^2}{\epsilon}\right)^i$, and therefore $s\leq t+2$ (and symmetrically, $t\leq s+2$). Thus, there is a set $J\subseteq\{1,2,...,\lceil\log_2\frac{n^2}{\epsilon}\rceil\}$ of at most six indices such that $(u,v)$ is in $\bigcup_{t\in J}\mathcal{P}_{2^t}$, and therefore one of the oracles $\{D_t\}_{t\in J}$ outputs a distance estimate for $(u,v)$ with stretch at most $1+13\epsilon$. The oracle $D$ applies these oracles on $(u,v)$ and finds the index $t\in J$ such that $D_t$ outputs the minimal distance estimate $\hat{d}(u,v)$ for $(u,v)$. If a path is also required, $D$ applies $D_t$ on $(u,v)$ to obtain such an approximate path. In both cases, $\hat{d}(u,v)$ has stretch at most $1+13\epsilon$.

Note that the query time of $D$ equals the query time of $D_0$ plus the query times of the (at most six) oracles $\{D_t\}_{t\in J}$. As these are all $O(p\cdot\frac{\log n}{\epsilon})$, the query time of $D$ is $O(p\cdot\frac{\log n}{\epsilon})$ as well.

The output paths of $D$ are all in $S=\bigcup_{t=1}^{\lceil\log_2\frac{n^2}{\epsilon}\rceil}S_t$. The size of $S$ and of $D$ are both at most the total size of the oracles $\{D_i\}_{i=0}^{\lceil\log_2\frac{n^2}{\epsilon}\rceil}$, which is $O(n\cdot p\cdot\frac{\log n}{\epsilon}\cdot\log\frac{n^2}{\epsilon})=O(n\cdot p\cdot\frac{\log^2n}{\epsilon})$ (here we used our assumption that $\epsilon>\frac{1}{n}$). Replacing $\epsilon$ with $\frac{\epsilon}{13}$, we get the desired result.
\end{proof}

\subsection{Completing the Proof of Theorem \ref{thm:LargeStretchImprovedQT}} \label{app:PRSpannerAspectRatioReduction2}

\begin{proof}[Proof of Theorem \ref{thm:LargeStretchImprovedQT}]

As in the proof in Appendix \ref{app:PRSpannerAspectRatioReduction1}, we assume that $\min_{u\neq v\in V}d_G(u,v)=\min_{e\in E}w(e)=1$ (thus $\Lambda=diam(G)$) that $G$ is connected, and that $\epsilon>\frac{1}{n}$. In this proof we construct the desired path-reporting spanner with stretch $O(k\log\log(\pi+4))$. The construction of the path-reporting emulator with stretch $(32e+\epsilon)k+1$ is identical, when using the (not necessarily spanning) tree covers of Corollary \ref{cor:EllPiOneScale}, instead of the spanning forest cover of the same corollary. We elaborate on this variation at the end of this proof.

Apply Theorem \ref{thm:LargeStretchDirectPRDO} on $G$, with $k=2\log\pi$ and $\epsilon=1$. We obtain a path-reporting emulator $(S_0,D_0)$ with stretch $\alpha_0=O(\log\pi)$, query time $O(\ell\cdot\log\pi\cdot\log^{3-\gamma}n)$ and size $O(\ell\cdot\log\pi\cdot n\log^{3-\gamma}n)$. 

The graph $G$ is $(\ell,\pi)$-path-separable. Therefore, by Corollary \ref{cor:EllPiOneScale}, it satisfies the condition of Lemma \ref{lemma:AspectRatioReductionForests} with $\alpha=O(k\log\log(\pi+4))$ and $M(\epsilon,n)=O(k\cdot\ell\cdot(\frac{\pi}{\epsilon}\log\frac{\pi}{\epsilon})^{\frac{1}{k}}\cdot\log^{2-\gamma}n)$, by  (these parameters hold for \textit{spanning} tree covers). Thus, the proof of Lemma \ref{lemma:AspectRatioReductionForests}, in conjunction with Observation \ref{obs:FromForestCoverToTreeCover}, shows that $G$ has a $\mathcal{P}_\Delta$-pairwise spanning tree cover $\mathcal{T}_{\Delta}$, for every $\Delta\in\{2^t\}_{t=1}^{\lceil\log_2\frac{n^2}{\epsilon}\rceil}$, with stretch $(1+6\epsilon)\alpha=O(k\log\log(\pi+4))$ and size $M(\epsilon,n)=O(k\cdot\ell\cdot(\frac{\pi}{\epsilon}\log\frac{\pi}{\epsilon})^{\frac{1}{k}}\cdot\log^{2-\gamma}n)$, where 
\[\mathcal{P}_\Delta=\left\{(u,v)\in V^2\;\bigg|\;\exists_{i\geq0}\;\frac{\Delta}{2}\cdot\left(\frac{n^2}{\epsilon}\right)^i<d_G(u,v)\leq\Delta\cdot\left(\frac{n^2}{\epsilon}\right)^i\right\}~.\]

For every $t=1,2,...,\lceil\log_2\frac{n^2}{\epsilon}\rceil$, we apply Lemma \ref{lemma:TreeCoverToPREmulator} on $\mathcal{T}_{2^t}$, and obtain a $\mathcal{P}_{2^t}$-pairwise path-reporting $O(k\log\log(\pi+4))$-spanner $(S_t,D_t)$ for $G$ with query time $O(k\cdot\ell\cdot(\frac{\pi}{\epsilon}\log\frac{\pi}{\epsilon})^{\frac{1}{k}}\cdot\log^{2-\gamma}n)$ and size $O(k\cdot\ell\cdot(\frac{\pi}{\epsilon}\log\frac{\pi}{\epsilon})^{\frac{1}{k}}\cdot n\log^{2-\gamma}n)$. We define a new path-reporting spanner $(S,D)$ by $S=\bigcup_{t=1}^{\lceil\log_2\frac{n^2}{\epsilon}\rceil}S_t$, where $D$ stores all the oracles $\{D_t\}_{t=0}^{\lceil\log_2\frac{n^2}{\epsilon}\rceil}$ and the oracle $D_0$.

Given a query $(u,v)\in V^2$, the oracle $D$ first applies $D_0$ on $(u,v)$, to get an estimate $\hat{d}_0(u,v)$ for which $\frac{1}{\alpha_0}\hat{d}_0(u,v)\leq d_G(u,v)\leq\hat{d}_0(u,v)$, within time $O(\ell\cdot\log\pi\cdot\log^{3-\gamma}n)$. Then, the oracle $D$ finds all the values $t\in\{1,2,...,\lceil\log_2\frac{n^2}{\epsilon}\rceil\}$ and $i\geq0$ such that the interval $[\frac{1}{\alpha_0}\hat{d}_0(u,v),\hat{d}_0(u,v)]$ intersects the interval $\mathcal{I}_{t,i}=\left(2^{t-1}\cdot\left(\frac{n^2}{\epsilon}\right)^i,2^t\cdot\left(\frac{n^2}{\epsilon}\right)^i\right]$. We claim that there are at most $2$ values of $i$ such that $[\frac{1}{\alpha_0}\hat{d}_0(u,v),\hat{d}_0(u,v)]\cap\mathcal{I}_{t,i}\neq\emptyset$ (for $t$ in the relevant range), and that for every such $i$ there are at most $\log_2\alpha_0+1$ possible values of $t$. Indeed, if $x\in\mathcal{I}_{t,i}$ and $y\in\mathcal{I}_{s,j}$ are both in $[\frac{1}{\alpha_0}\hat{d}_0(u,v),\hat{d}_0(u,v)]$, for some $i<j$ and $1\leq s,t\leq\lceil\log_2\frac{n^2}{\epsilon}\rceil$, then 
\[\left(\frac{n^2}{\epsilon}\right)^j\leq2^{s-1}\cdot\left(\frac{n^2}{\epsilon}\right)^j<y\leq\hat{d}_0(u,v)\leq\alpha_0\cdot x\leq\alpha_0\cdot2^t\cdot\left(\frac{n^2}{\epsilon}\right)^i<2\alpha_0\cdot\left(\frac{n^2}{\epsilon}\right)^{i+1}<\left(\frac{n^2}{\epsilon}\right)^{i+2}~,\]
i.e., $j\leq i+1$ (here we assumed that $2\alpha_0$, which is $O(\log\pi)=O(\log n)$, is at most $\frac{n^2}{\epsilon}$ -- this is true for any large enough $n$). Also, if $i=j$, we get from the same inequality that $2^{s-1}\cdot\left(\frac{n^2}{\epsilon}\right)^i<\alpha_0\cdot2^t\cdot\left(\frac{n^2}{\epsilon}\right)^i=2^{t+\log_2\alpha_0}\cdot\left(\frac{n^2}{\epsilon}\right)^i$, and therefore $s\leq t+\log_2\alpha_0$ (and symmetrically, $t\leq s+\log_2\alpha_0$). Thus, there is a set $J\subseteq\{1,2,...,\lceil\log_2\frac{n^2}{\epsilon}\rceil\}$ of indices, with size $|J|=O(\log\alpha_0)=O(\log\log\pi)$ such that $(u,v)$ is in $\bigcup_{t\in J}\mathcal{P}_{2^t}$, and therefore one of the oracles $\{D_t\}_{t\in J}$ outputs a distance estimate for $(u,v)$ with stretch $O(k\log\log(\pi+4))$. The oracle $D$ applies these oracles on $(u,v)$ and finds the index $t\in J$ such that $D_t$ outputs the minimal distance estimate $\hat{d}(u,v)$ for $(u,v)$. If a path is also required, $D$ applies $D_t$ on $(u,v)$ to obtain such an approximate path in $S=\bigcup_{t=1}^{\lceil\log_2\frac{n^2}{\epsilon}\rceil}S_t$. In both cases, $\hat{d}(u,v)$ has stretch at most $O(k\log\log(\pi+4))$.

Note that the query time of $D$ equals the query time of $D_0$ plus the query times of the $O(\log\log\pi)$ oracles $\{D_t\}_{t\in J}$. As these query times are all $O(k\cdot\ell\cdot(\frac{\pi}{\epsilon}\log\frac{\pi}{\epsilon})^{\frac{1}{k}}\cdot\log^{2-\gamma}n)$, the query time of $D$ is $O(\ell\cdot\log\pi\cdot\log^{3-\gamma}n+k\cdot\ell\cdot(\frac{\pi}{\epsilon}\log\frac{\pi}{\epsilon})^{\frac{1}{k}}\cdot\log^{2-\gamma}n\cdot\log\log\pi)$.

The size of $S$ and of $D$ are both at most the total size of the oracles $\{D_i\}_{i=0}^{\lceil\log_2\frac{n^2}{\epsilon}\rceil}$, which is 
\[\left\lceil\log_2\frac{n^2}{\epsilon}\right\rceil\cdot O\left(k\cdot\ell\cdot\left(\frac{\pi}{\epsilon}\log\frac{\pi}{\epsilon}\right)^{\frac{1}{k}}\cdot n\log^{2-\gamma}n\right)=O\left(k\cdot\ell\cdot\left(\frac{\pi}{\epsilon}\log\frac{\pi}{\epsilon}\right)^{\frac{1}{k}}\cdot n\log^{3-\gamma}n\right)~,\]
where we used our assumption that $\epsilon>\frac{1}{n}$.

For a path-reporting emulator, we use the very same proof, but with $(S_t,D_t)$ defined as the $\mathcal{P}_{2^t}$-pairwise path-reporting $(1+6\epsilon)(32ek+1)$-emulator that is implied by Lemma \ref{lemma:TreeCoverToPREmulator} when invoked on the tree cover $\mathcal{T}_{2^t}$ from the proof of Lemma \ref{lemma:AspectRatioReductionForests}. By Corollary \ref{cor:EllPiOneScale}, the tree cover $\mathcal{T}_{2^t}$ that is constructed\footnote{In fact, Lemma \ref{lemma:AspectRatioReductionForests} constructs a \textit{forest} cover $\mathcal{F}_{2^t}$, but using Observation \ref{obs:FromForestCoverToTreeCover}, we can assume that this is a tree cover.} in Lemma \ref{lemma:AspectRatioReductionForests} is not necessarily spanning, has stretch $(1+6\epsilon)(32ek+1)<(32e+528\epsilon)k+1$, and has size $O(k\cdot\ell\cdot(\frac{\pi}{\epsilon}\log\frac{\pi}{\epsilon})^{\frac{1}{k}}\cdot\log^{2-\gamma}n)$. The rest of the proof is identical, which results in a path-reporting emulator with the desired properties, when replacing $\epsilon$ by $\frac{\epsilon}{528}$.
\end{proof}

\section{Distance Oracles with Improved Stretch} \label{app:ImprovedStretch}

Theorem \ref{thm:SmallTreewidthEmulatorAndSpanner} applies Lemma \ref{lemma:TreeCoverToPREmulator} on Theorem \ref{thm:SmallTreewidthSpanTreeCover}, to obtain a path-reporting $(32ek+1)$-emulator. In this section, we show how the stretch can be improved to $4ek+1$. We note, however, that the structure we present here is not a path-reporting emulator, but only a \textit{distance oracle}. That is, given a query $(u,v)\in V^2$, this structure outputs only an estimate $\hat{d}(u,v)$ that satisfies $d_G(u,v)\leq\hat{d}(u,v)\leq(4ek+1)\cdot d_G(u,v)$, but not a path with weight $\hat{d}(u,v)$ in a fixed small-size emulator of $G$.

To construct such a distance oracle, we do not use the tree cover (or the HST cover) from Theorem \ref{thm:SmallTreewidthSpanTreeCover} as a black-box, but instead we go back to the proof of Lemma \ref{lemma:PairwiseSpanTreeCover}. We roughly describe its details here. Given a graph $G=(V,E)$, a parameter $k$, and a subset $A\subseteq V$ of size $t$, we iteratively find $I=O(k\cdot t^{\frac{1}{k}})$ disjoint subsets $S_1,S_2,...,S_I\subseteq A$ and corresponding trees $T_1,T_2,...,T_I$. Here, the vertex set of $T_i$ is $V(T_i)=V\setminus\bigcup_{j=1}^{i-1}S_j$. These trees has the property that for any $u,v\in V$, if $i$ is the minimal index such that a shortest $u$-$v$ path intersects $S_i$, then
\[d_G(u,v)\leq d_{T_i}(u,v)\leq(32ek+1)d_G(u,v)~.\footnote{In fact, the inequality $d_G(u,v)\leq d_{T_i}(u,v)$ holds for every tree $T_i$.}\]
Thus, every pair of vertices $u,v\in V$ with a shortest path that intersects $A$ has a tree in this collection, in which there is a low-stretch path between $u,v$.

The subsets $\{S_i\}_{i=1}^I$ and the trees $\{T_i\}_{i=1}^I$ are found using Theorem \ref{thm:RamseyNotSubTree}. In the proof of the next lemma, we replace each tree $T_i$ by an HST \textit{only over $S_i$}, as opposed to $V\setminus\bigcup_{j=1}^{i-1}S_j$. Instead, for every $v\in V\setminus\bigcup_{j=1}^{i}S_j$, we store the closest vertex $s_v\in S_i$ to $v$, and the distance $d_G(v,s_v)$. Using this information (for every $i$), we are able to extract an estimate of $d_G(u,v)$, for every $u,v\in V$ for which a shortest path intersects $A$.

\begin{lemma} \label{lemma:PairwiseDO}
Let $G=(V,E)$ be an undirected weighted graph on $n$ vertices, and fix some $A\subseteq V$ with size $|A|=t$. Define the set $\mathcal{P}_A\subseteq V^2$ as the set of all pairs $(u,v)$ such that there is a shortest $u$-$v$ path that contains a vertex of $A$. Given an integer parameter $k\geq1$, there exists a $\mathcal{P}_A$-pairwise distance oracle with stretch $4ek+1$, size $O(knt^{\frac{1}{k}})$ and query time $O(k\cdot t^{\frac{1}{k}})$. That is, there is an oracle $D_A$ of size $O(knt^{\frac{1}{k}})$, such that for every query $(u,v)\in\mathcal{P}_A$, it outputs an estimate $\hat{d}(u,v)$ within time $O(k\cdot t^{\frac{1}{k}})$, that satisfies
\[d_G(u,v)\leq\hat{d}(u,v)\leq(4ek+1)d_G(u,v)~.\]
\end{lemma}

\begin{proof}

Similarly to the proof of Lemma \ref{lemma:PairwiseSpanTreeCover}, our construction is recursive over $t$, the size of the \textit{demand set} $A\subseteq V$. Let $\alpha(t),\Sigma(n,t)$ and $q(t)$ be upper bounds for the stretch, the size and the query time of the resulting $\mathcal{P}_A$-pairwise distance oracle, for a graph $G$ with $n$ vertices and a set $A$ with $t$ vertices.

For $t=1$ and $A=\{r\}$, the oracle $D_A$ stores all the distances of the form $d_G(v,r)$, for every $v\in V$. Given a query $(u,v)\in\mathcal{P}_A$, the oracle $D_A$ outputs the value $\hat{d}(u,v)=d_G(u,r)+d_G(v,r)$. Recall that $(u,v)\in\mathcal{P}_A$ means that $r$ is on a shortest path between $u,v$, and thus, $\hat{d}(u,v)=d_G(u,v)$. We conclude that the stretch of $D_A$ is $1<4ek+1$. Clearly, the size and the query time of $D_A$ are $O(n)$ (words) and $O(1)$, respectively. Hence, we set $\alpha(1)=4ek+1$, $\Sigma(n,1)=O(n)$ and $q(1)=O(1)$.

If $t>1$, we use Theorem \ref{thm:RamseyUM} to find a subset $S\subseteq A$ of size $t^{1-\frac{1}{k}}$ and an ultrametric $\rho$ over $S$, such that
\begin{equation} \label{eq:UMSubset}
    d_G(u,v)\leq\rho(u,v)\leq2ek\cdot d_G(u,v)~,
\end{equation}
for every $u,v\in S$. Recall that by Fact \ref{fact:UM<->HST}, there is an HST $T=(V_T,E_T)$ with the vertices of $S$ as leaves and with labels $\ell:V_T\rightarrow\mathbb{R}_{\geq0}$ such that $\rho(u,v)=\ell(LCA_T(u,v))$, for every $u,v\in S$. Applying Theorem \ref{thm:LCAOracle} on the HST $T$, we obtain an oracle $D_T$ with size\footnote{Generally, we may assume that an HST does not contain any vertex of degree $2$, and therefore the number of vertices in an HST is at most a constant times the number of its leaves.} $O(|V_T|)=O(|S|)=O(t^{1-\frac{1}{k}})$, that can output, within a constant time, the LCA of any two vertices $u,v\in S$. Therefore, if $D_T$ stores the labels $\ell$ as well, it can output $\rho(u,v)$, for every $u,v\in S$ within $O(1)$ time.

We now define another oracle $\bar{D}_T$, that contains $D_T$, and also stores, for every vertex $v\in V\setminus S$, its closest vertex $s_v\in S$ and the distance $d_G(v,s_v)$. The size of $\bar{D}_T$ is $O(n)+O(t^{1-\frac{1}{k}})=O(n)$. Given a query $(u,v)$, the oracle $\bar{D}_T$ outputs the estimate 
\[\hat{d}(u,v)=d_G(u,s_u)+\rho(s_u,s_v)+d_G(s_v,v)~.\]
Here, the first and the third components are explicitly stored in $\bar{D}_T$, and the second is the output of the oracle $D_T$ on the query $(s_v,s_u)$. The query time is clearly still $O(1)$.

Consider the graph $G'=G[V\setminus S]$ and the new demand set $A'=A\setminus S$. Let $n'=n-|S|\leq n-t^{1-\frac{1}{k}}$. denote the number of vertices in $G'$, and let $t'=t-|S|\leq t-t^{1-\frac{1}{k}}$ denote the number of vertices in $A'$. Recursively, let $D_{A'}$ be a $\mathcal{P}_{A'}$-pairwise distance oracle with stretch $\alpha(t')$, size $\Sigma(n',t')$ and query time $q(t')$. We define the oracle $D_A$ to store both oracles $\bar{D}_T$ and $D_{A'}$. Given a query $(u,v)\in\mathcal{P}_A$, the oracle $D_A$ first gets an estimate $\hat{d}_1(u,v)$ from $\bar{D}_T$, then gets an estimate $\hat{d}_2(u,v)$ from $D_{A'}$. Then, it outputs the minimal value among $\hat{d}_1(u,v),\hat{d}_2(u,v)$.

Note that the size and the query time of the oracle $D_A$ are bounded by $O(n)+\Sigma(n',t')\leq O(n)+\Sigma(n-t^{1-\frac{1}{k}},t-t^{1-\frac{1}{k}})$ and $O(1)+q(t')\leq O(1)+q(t-t^{1-\frac{1}{k}})$. Thus, we can set
\[\Sigma(n,t)=O(n)+\Sigma(n-t^{1-\frac{1}{k}},t-t^{1-\frac{1}{k}}),\text{  and  }q(t)=O(1)+q(t-t^{1-\frac{1}{k}})~.\]
Note that these are the same recursive relations as in the proof of Lemma \ref{lemma:PairwiseSpanTreeCover} (see (\ref{eq:SizeRec}) and (\ref{eq:RecDepth})). Hence, by the same proof, we get the bounds $\Sigma(n,t)=O\left(knt^{\frac{1}{k}}\right)$ and $q(t)=O(k\cdot t^{\frac{1}{k}})$ on the size and the query time, respectively.

To analyze the stretch, let $(u,v)\in\mathcal{P}_A$ be a query. If there is a shortest $u$-$v$ path that intersects $A'$, then the oracle $D_{A'}$ provides an estimate with stretch at most $\alpha(t')$ for this query. Otherwise, by the definition of $\mathcal{P}_A$, there must be a shortest $u$-$v$ path that contains a vertex of $A\setminus A'=S$. Denote this vertex by $z\in S$. Recall that $s_u,s_v$ are the closest vertices in $S$ to $u,v$, respectively, and therefore,
\[d_G(s_u,s_v)\leq d_G(s_u,u)+d_G(u,z)+d_G(z,v)+d_G(v,s_v)\leq2(d_G(u,z)+d_G(z,v))=2d_G(u,v)~.\]
Using this inequality, we bound the output of $\bar{D}_T$, which in turn bounds the output of $D_A$.
\begin{eqnarray*}
    \bar{D}_T(u,v)&=&d_G(u,s_u)+\rho(s_u,s_v)+d_G(s_v,v)\\
    &\stackrel{(\ref{eq:UMSubset})}{\leq}&d_G(u,s_u)+2ek\cdot d_G(s_u,s_v)+d_G(s_v,v)\\
    &\leq&d_G(u,s_u)+2ek\cdot2d_G(u,v)+d_G(s_v,v)\\
    &\leq&d_G(u,z)+2ek\cdot2d_G(u,v)+d_G(z,v)\\
    &=&d_G(u,v)+2ek\cdot2d_G(u,v)=(4ek+1)d_G(u,v)~.
\end{eqnarray*}

We obtain the following bound on the stretch of $D_A$.
\[\alpha(t)\leq\max\{\alpha(t'),4ek+1\}~.\]
A simple inductive proof, similar to the one in the proof of Lemma \ref{lemma:PairwiseSpanTreeCover}, gives a bound of $4ek+1$ on the stretch of our oracle $D_A$.

\end{proof}

Applying Lemma \ref{lemma:PairwiseDO} on the separators of a vertex-separable graph, we get the following result.

\begin{theorem*}[Theorem \ref{thm:SmallTreewidthDO}]
Let $G=(V,E)$ be an $n$-vertex undirected weighted $s$-vertex-separable graph, for some non-decreasing function $s=s(\theta)$. Let $k\geq1$ be an integer parameter, and denote $S(n,k)=\sum_{i=0}^{\log n}s\left(\frac{n}{2^i}\right)^{\frac{1}{k}}$. There exists a distance oracle for $G$ with stretch at most $4ek+1$, size $O(k\cdot n\cdot S(n,k))$, and query time $O(k\cdot S(n,k))$.
\end{theorem*}

\begin{proof}

We prove the theorem by induction over $n$, the number of vertices. For any constant $n$ the theorem is trivial. For a general $n>1$, assume by induction that any graph $G'$ with $n'<n$ vertices, that satisfies the condition in the theorem, has a distance oracle $D'$ with stretch $4ek+1$, size $O(k\cdot n'\cdot S(n',k))$ and query time $O(k\cdot S(n',k))$. In particular, we assume that the size of $D'$ is at most $(\tilde{c}+1)\cdot k\cdot n'\cdot S(n',k)$, where $\tilde{c}$ is the constant hidden in the $O$-notation in the size of the pairwise distance oracle of Lemma \ref{lemma:PairwiseDO}.

Let $A\subseteq V$ be a separator of size $|A|\leq s(n)$, and let $C_1,C_2,...,C_\ell$ be the connected components of the graph $G'=G[V\setminus A]$. If $n_i$ is the number of vertices in $C_i$, then for every $i$ we have $n_i\leq\frac{n}{2}$. Let $D_A$ be a $\mathcal{P}_A$-pairwise distance oracle as in Lemma \ref{lemma:PairwiseDO}. Recall that $\mathcal{P}_A$ is the set of all pairs $(u,v)\in V^2$ for which a shortest $u$-$v$ path intersects $A$, and $D_A$ outputs a distance estimate for any such pair, with stretch $4ek+1$. The size of $D_A$ is $O(kns(n)^{\frac{1}{k}})$, and its query time is $O(ks(n)^{\frac{1}{k}})$.

For every $i=1,2,...,\ell$, the component $C_i$ are also $s(\theta)$-vertex-separable. By the induction hypothesis, there is a distance oracle $D_i$ for the sub-graph $C_i$, with stretch $4ek+1$, size $(\tilde{c}+1)\cdot k\cdot n_i\cdot S(n_i,k)$ and query time $O(k\cdot S(n_i,k))$. We define a new distance oracle $D$ for the graph $G$, that stores the oracle $D_A$ and the oracles $\{D_i\}_{i=1}^\ell$, and a variable $i(v)$, for each $v\in V\setminus A$, that indicates the unique index $i$ such that $v$ is in $C_i$. Given a query $(u,v)\in V^2$, the oracle $D$ applies the oracle $D_A$ on this query, and obtains an estimate $\hat{d}_A(u,v)$. Then, if $i(u)=i(v)=i$, it applies the oracle $D_i$ on $(u,v)$ to obtain another estimate $\hat{d}_i(u,v)$. Lastly, $D$ outputs the smaller among the estimates $\hat{d}_A(u,v),\hat{d}_i(u,v)$ (in case $i(u)\neq i(v)$, it simply outputs $\hat{d}_A(u,v)$).

To analyze the stretch of the oracle $D$, we consider two cases. If $(u,v)\in\mathcal{P}_A$, i.e., there is a shortest path between $u,v$ that intersects $A$, then by Lemma \ref{lemma:PairwiseDO} we have $\hat{d}_A(u,v)\leq(4ek+1)d_G(u,v)$, and since the resulting estimate of $D$ is always at most $\hat{d}_A(u,v)$, we conclude that this estimate has stretch at most $4ek+1$ as well. Otherwise, if every shortest path between $u,v$ do not intersect $A$, then in particular, $u,v$ are in the same connected component $C_i$ of $G'=G[V\setminus A]$. In this case, by the induction hypothesis, the oracle $D_i$ outputs an estimate $\hat{d}_i(u,v)\leq(4ek+1)d_G(u,v)$. The resulting estimate of $D$ is at most $\hat{d}_i(u,v)$, and therefore it has stretch at most $4ek+1$ as well.

The size of the oracle $D$ consists of the size of $D_A$, which is at most $\tilde{c}\cdot k\cdot n\cdot s(n)^{\frac{1}{k}}$ by Lemma \ref{lemma:PairwiseDO}, of the sum of the sizes of the oracles $D_i$, which (by the induction hypothesis) is
\[\sum_{i=1}^{\ell}(\tilde{c}+1)\cdot k\cdot n_i\cdot S(n_i,k)\leq (\tilde{c}+1)\cdot k\sum_{i=1}^{\ell}n_i\cdot S\left(\frac{n}{2},k\right)<(\tilde{c}+1)\cdot k\cdot n\cdot S\left(\frac{n}{2},k\right)~,\]
and of the variables $\{i(v)\}_{v\in V\setminus A}$, which have a total size of less than $n$ words. We conclude that the total size of $D$ is at most
\begin{eqnarray*}
    \tilde{c}\cdot k\cdot n\cdot s(n)^{\frac{1}{k}}+(\tilde{c}+1)\cdot k\cdot n\cdot S\left(\frac{n}{2},k\right)+n&<&(\tilde{c}+1)\cdot k\cdot n\cdot\left(s(n)^{\frac{1}{k}}+S\left(\frac{n}{2},k\right)\right)\\
    &=&(\tilde{c}+1)\cdot k\cdot n\cdot S(n,k)~.
\end{eqnarray*}
The last equality is due to the definition of $S(n,k)=\sum_{i=0}^{\log n}s(\frac{n}{2^i})^\frac{1}{k}$.

Note that for every query $(u,v)\in V^2$, the oracle $D$ queries $D_A$, which takes $O(k\cdot s(n)^{\frac{1}{k}})$ time, by Lemma \ref{lemma:PairwiseDO}. Then, if $i(u)=i(v)=i$, the oracle $D$ queries $D_i$ as well, which takes $O(k\cdot S(n_i,k))=O\left(k\cdot S\left(\frac{n}{2},k\right)\right)$ time, by the induction hypothesis.Again, by the definition of $S(n,k)$, we conclude that the query time of $D$ is at most 
\[O\left(k\cdot s(n)^{\frac{1}{k}}+k\cdot S\left(\frac{n}{2},k\right)\right)=O(k\cdot S(n,k))~.\]
    
\end{proof}

\begin{corollary}
Let $G$ be an $n$-vertex $s$-vertex-separable graph, for $s=s(\theta)=\theta^\delta$, for some $\delta\in(0,1]$, and let $k\geq1$ be an integer. Then, $G$ has a distance oracle with stretch $4ek+1$, size $O\left(\frac{k^2}{\delta}n^{1+\frac{\delta}{k}}\right)$, and query time $O\left(\frac{k^2}{\delta}n^{\frac{\delta}{k}}\right)$.
\end{corollary}

\begin{corollary}
Let $G$ be an $n$-vertex graph with treewidth $t=t(n)$, and let $k\geq1$ be an integer. Then, $G$ admits a distance oracle with stretch $4ek+1$, query time $O\left(k\cdot t^{\frac{1}{k}}\log n\right)$ and size $O\left(k\cdot t^{\frac{1}{k}}\cdot n\log n\right)$.
\end{corollary}

\section{Distance Labeling and Routing Schemes} \label{app:DLandRS}

While a distance oracle provides a centralized data structure, that can approximately compute distances or shortest paths in a graph, an alternative goal is to \textit{distribute} this information among the graph vertices. That is, with only a few words of information stored in each vertex of a graph $G$, we want to be able to approximately compute $d_G(u,v)$, or a shortest $u$-$v$ path in $G$, given only the information stored in the vertices $u,v$. If we are only interested in the distance $d_G(u,v)$ itself, this approach is called a \textit{distance labeling scheme}. In case we are interested in a shortest $u$-$v$ path, we use the notion of \textit{routing schemes}.

\subsection{Distance Labeling Schemes} \label{sec:DistanceLabeling}

A distance labeling scheme of a graph is an assignment of a short \textbf{label} to each vertex, such that the (possibly approximate) distance between any two vertices can be computed efficiently by considering only their labels. The label of a vertex may contain any data, and its size is measured in words. Results on distance labeling schemes include \cite{Mat96,P00a,GPPR04,TZ01}. Usually, the main focus while constructing distance labeling schemes, is to reduce the stretch and the maximal, or sometimes average, size of a label. Next, we formally define distance labeling schemes.

\begin{definition} \label{def:DistanceLabelingSchemes}
Given an $n$-vertex undirected weighted graph $G=(V,E)$, a \textbf{distance labeling scheme} $\mathcal{L}(G)$ is a collection of labels $\{l(v)\}_{v\in V}\subseteq\{0,1\}^*$, and an algorithm $\mathcal{Q}$ that receives the labels $l(u),l(v)$ of two vertices $u,v\in V$, and returns an estimate $\hat{d}(u,v)$. We say that $\mathcal{L}(G)$ has \textbf{stretch} $\alpha$, for some real number $\alpha>0$, if for every $u,v\in V$,
\[d_G(u,v)\leq\hat{d}(u,v)\leq\alpha\cdot d_G(u,v)~.\]
We denote by $|l(v)|$ the length (in words) of the label $l(v)$. Then, we define the \textbf{(maximum) label size} of $\mathcal{L}(G)$ as $\max_{v\in V}|l(v)|$, and the \textbf{average label size} of $\mathcal{L}(G)$ as $\frac{1}{n}\sum_{v\in V}|l(v)|$.
\end{definition}

Consider the following result of Peleg \cite{P00b}, regarding distance labeling schemes in trees.

\begin{theorem}[\cite{P00a,P00b}] \label{thm:TreeLabels}
Let $T=(V,E)$ be an undirected weighted tree with $n$ vertices. There exists a distance labeling scheme for $T$ with stretch $1$, maximum label size $O(\log n)$.
\end{theorem}

Theorem \ref{thm:TreeLabels}, and Theorem \ref{thm:LCAOracle}, imply a direct way to construct a distance labeling scheme for a graph $G=(V,E)$ that is accompanied by a (not necessarily spanning) tree cover $\mathcal{T}$ or an HST cover $\mathcal{T}^{HST}$, respectively. Namely, the label of each vertex $v\in V$ consists of its labels in $T$, for every $T\in\mathcal{T}$ or for every $T\in\mathcal{T}^{HST}$. Given the labels of $u,v\in V$, we can simply iterate over all such trees $T$, and using the labels of $u,v$ in $T$, compute their distance estimate in $T$. Finally, we return the minimum estimate we found. The following lemma formalizes this scheme.

\begin{lemma} \label{lemma:DistanceLabelingFromTreeCover}
Suppose that an undirected weighted $n$-vertex graph $G=(V,E)$ has a tree cover $\mathcal{T}$, or an HST cover $\mathcal{T}^{HST}$, with stretch $\alpha$ and $t$ trees. Then, there is a distance labeling scheme for $G$ with stretch $\alpha$ and maximum label size $O(t\log n)$. Moreover, using the labels of two vertices $u,v\in V$, it is possible to find a tree $T\in\mathcal{T}$, or an HST $T\in\mathcal{T}^{HST}$, such that $d_T(u,v)\leq\alpha\cdot d_G(u,v)$, or respectively, $\rho_T(u,v)\leq\alpha\cdot d_G(u,v)$.
\end{lemma}

Applying Lemma \ref{lemma:DistanceLabelingFromTreeCover} on our new results for tree covers and HST covers, we obtain the following result (for each result, the tree cover or HST cover that it relies on is specified).

\begin{theorem} \label{thm:PathSeparableDistanceLabeling}
Let $G=(V,E)$ be an $n$-vertex undirected weighted graph. Let $k\geq1$ be an integer parameter, and let $\epsilon\in(0,1]$ be a real parameter.

\begin{enumerate}
    \item If $G$ is $s$-vertex-separable, where $s=s(\theta)=\theta^\delta$, for some $\delta\in(0,1]$, then $G$ has a distance labeling scheme with stretch $12ek$ and maximum label size $O\left(\frac{k^2}{\delta}n^{\frac{\delta}{k}}\log n\right)$ (Corollary \ref{cor:TWN^delta}).

    \item If $G$ has treewidth $t=t(n)$, then there exists a distance labeling scheme for $G$ with stretch $12ek$ and maximum label size $O\left(k\cdot t^{\frac{1}{k}}\cdot\log^2n\right)$ (Corollary \ref{cor:FlatTW}).

    \item If $G$ is $p$-path-separable, then there exists a distance labeling scheme for $G$ with stretch $1+\epsilon$ and maximum label size $O(p\cdot\frac{\log^3n}{\epsilon})$ (Theorem \ref{thm:SpanTreeCoverForPathSeparators}).

    \item If $G$ is $(\ell,\pi)$-path-separable, then there exists a distance labeling scheme for $G$ with stretch $(12e+\epsilon)k$ and maximum label size $O\left(k\cdot\ell\cdot\left(\frac{\pi}{\epsilon}\log\frac{\pi}{\epsilon}\right)^\frac{1}{k}\cdot\log^4n\right)$ (Theorem \ref{thm:LargeStretchTreeCoverForPathSeparators}).

    \item If $G$ is tree-like $(\ell,\pi)$-path-separable, then there exists a distance labeling scheme for $G$ with stretch $(12e+\epsilon)k$ and maximum label size $O\left(k\cdot\ell\cdot\left(\frac{\pi}{\epsilon}\right)^\frac{1}{k}\cdot\log^3n\right)$ (Theorem \ref{thm:LargeStretchTreeCoverForPathSeparators}).
\end{enumerate}

\end{theorem}

Recall that $K_r$-minor-free graphs are (weakly) tree-like $(\ell,\pi)$-path-separable, for $\ell=O(r^2)$ and $\pi=O(r^{c_{AG}})$, where $c_{AG}=4602$, and that graphs with bounded genus $g$ are strongly tree-like $O(g)$-path-separable (see Items (2) and (3) in Theorem \ref{thm:PathSeparableFamilies}). Thus, the following corollary is derived by Theorem \ref{thm:PathSeparableDistanceLabeling}.

\begin{corollary}
    Let $G=(V,E)$ be an $n$-vertex undirected weighted graph. Let $k\geq1$ be an integer parameter, and let $\epsilon\in(0,1]$ be a real parameter.

    \begin{enumerate}
        \item If $G$ is $K_r$-minor-free, then there exists a distance labeling scheme for $G$ with stretch $1+\epsilon$ and maximum label size $O(r^{c_{AG}}\cdot\frac{\log^3n}{\epsilon})$.
        
        \item If $G$ has a bounded genus $g$, then there exists a distance labeling scheme for $G$ with stretch $1+\epsilon$ and maximum label size $O(g\cdot\frac{\log^3n}{\epsilon})$.

        \item If $G$ is $K_r$-minor-free, then there exists a distance labeling scheme for $G$ with stretch $O(k)$ and maximum label size $O\left(kr^{2+\frac{1}{k}}\cdot\log^3n\right)$.

        \item If $G$ has a bounded genus $g$, then there exists a distance labeling scheme for $G$ with stretch $O(k)$ and maximum label size $O\left(kg^\frac{1}{k}\cdot\log^3n\right)$.
\end{enumerate}

\end{corollary}

\subsection{Routing Schemes} \label{sec:RoutingSchemes}

A routing scheme $\mathcal{R}$ for an undirected and weighted graph $G=(V,E)$ is a distributed method to pass messages from one vertex to another. It consists of several objects. First, each vertex $v\in V$ is assigned a \textbf{label} $l(v)$, that contains required data for sending a message to destination $v$. Additionally, every vertex $v\in V$ is provided with a \textbf{routing table} $R(v)$, which is a data structure that helps $v$ route messages. When a vertex $u$ wants to send a message to a destination $v$, it accepts as an input its own routing table $R(u)$ and its label $l(u)$, as well as the label $l(v)$ of the destination $v$. Given this information, it is allowed to create a short \textbf{header} $h=h(u,v)$ and attach it to the message. Most importantly, it needs to compute the \textit{next hop}, i.e., the neighbor $u'$ of $u$ to which it will send the message (possibly with the header attached to it). The vertex $u'$, and any subsequent vertex which receives the message, acts similarly: based on its routing table $R(u')$, the destination label $l(v)$, and the header $h$, it computes the next neighbor $u''$ of $u'$, recomputes the header, and sends the updated message to $u''$.

This process creates a \textit{routing path} 
\[P_{\mathcal{R}}(u,v)=(u,u',u'',...)~.\] 
The routing scheme $\mathcal{R}$ is said to be \textbf{correct} for the graph $G=(V,E)$, if for every origin $u\in V$ and destination $v\in V$, the routing path $P_{\mathcal{R}}(u,v)$ ends in $v$. It is said to have \textbf{stretch} $\alpha$ if
\[w(P_{\mathcal{R}}(u,v))\leq\alpha\cdot d_G(u,v)~,\]
where $w(P_{\mathcal{R}}(u,v))=\sum_{e\in P_{\mathcal{R}}(u,v)}w(e)$.
We are also interested in minimizing the sizes of the labels $\{l(v)\}_{v\in V}$, the headers $\{h(u,v)\}_{u,v\in V}$ and the routing tables $\{R(u)\}_{u\in V}$ (all measured in words).

Next, we show that spanning tree covers imply routing schemes. For this purpose, we use the following routing scheme for trees.

\begin{theorem}[\cite{TZ01-spaa}] \label{thm:TreeRoutingScheme}
Every weighted tree $T=(V,E)$ on $n$ vertices has a routing scheme with stretch $1$, labels of size $O(\log n)$ and tables of size $O(1)$, all measured in words (there are no headers). 
\end{theorem}

\begin{lemma} \label{lemma:RoutingSchemeFromTreeCover}
Suppose that an undirected weighted graph $G=(V,E)$ has $n$ vertices and a spanning tree cover $\mathcal{T}$ with stretch $\alpha$ and overlap $t$.
Then, $G$ has a routing scheme with stretch $\alpha$, labels of size $O(t\log n)$, tables of size $O(t)$ and headers of size $O(1)$. 
\end{lemma}

\begin{proof}

Given the graph $G=(V,E)$, for every vertex $v\in V$ and a tree $T\in\mathcal{T}$ that contains $v$, let $l_T(v)$ and $R_T(v)$ be the label and the routing table of $v$, in the routing scheme of the tree $T$ from Theorem \ref{thm:TreeRoutingScheme}. In addition, let $l'(v)$ be the distance label of $v$ from Lemma \ref{lemma:DistanceLabelingFromTreeCover}, when applied to the tree cover $\mathcal{T}$.

We define a routing scheme for the graph $G$ as follows. For every $v\in V$, the label $l(v)$ consists of the label $l'(v)$ and all the labels $l_T(v)$, for every $T\in\mathcal{T}$ that contains $v$. The table $R(v)$ consists of all tables $R_T(v)$, again for all $T\in\mathcal{T}$ that contain $v$.

Given an origin $u\in V$ and a destination $v\in V$, assume that $u$ is given the label $l(v)$. Note that this label contains $l'(v)$. Also, the origin $u$ knows $l'(u)$, as it is a part of its own label $l(u)$. Using these two labels, the vertex $u$ finds the tree $T\in\mathcal{T}$ that satisfies $d_T(u,v)\leq\alpha\cdot d_G(u,v)$. Next, $u$ creates a header $h(u,v)$ that consists of the identity number of the tree $T$. Using the table $R_T(u)$, the vertex $u$ passes the message to the next vertex on the $u$-$v$ path in $T$. Each such vertex $u'$ proceeds to route the message, using the routing scheme of the tree $T$. Note that $u'$ has the identity of $T$ and the suitable table $R_T(u')$ (in $R(u')$).

Since the tree $T\in\mathcal{T}$ that was found at the first stage of the routing satisfies $d_T(u,v)\leq\alpha\cdot d_G(u,v)$, and since after that the routing is performed using the routing scheme from Theorem \ref{thm:TreeRoutingScheme} (which has stretch $1$) on the tree $T$, we conclude that the routing path has stretch $\alpha$. The header $h(u,v)$ consists of only one word (the identity of the tree $T$), and therefore the headers $h(u,v)$ are of size $O(1)$.

To analyze the size of the labels $\{l(v)\}_{v\in V}$, note that by Theorem \ref{thm:TreeRoutingScheme}, each label $l_T(v)$ is of size $O(\log n)$, for each of the trees $T\in\mathcal{T}$ that contain $v$. There are at most $t$ such trees, with an average of $s$ trees per vertex $v$. In addition, by Lemma \ref{lemma:DistanceLabelingFromTreeCover}, the size of the label $l'(v)$ is at most $O(t\log n)$, but $O(s\log n)$ on average. We get a total size of $O(t\log n)$ for the label $l(v)$, and $O(s\log n)$ on average.

As for the routing tables $\{R(u)\}_{u\in V}$, recall that the size of the table $R_T(u)$, for each of the $t$ trees ($s$ on average) $T\in\mathcal{T}$ that contain $u$, is $O(1)$, by Theorem \ref{thm:TreeRoutingScheme}. We conclude that the size of the table $R(u)$ is at most $O(t)$, but $O(s)$ on average.


\end{proof}

Applying Lemma \ref{lemma:RoutingSchemeFromTreeCover} on the tree cover from Theorem \ref{thm:SmallTreewidthSpanTreeCover}, and on Corollaries \ref{cor:TWN^delta} and \ref{cor:FlatTW}, we obtain the following result for vertex-separable graphs.

\begin{theorem} \label{thm:SmallTreewidthRouting}
Let $G=(V,E)$ be an $n$-vertex undirected weighted $s$-vertex-separable graph, for some non-decreasing function $s=s(\theta)$. Let $k\geq1$ be an integer parameter, and denote $S(n,k)=\sum_{i=0}^{\log n}s\left(\frac{n}{2^i}\right)^{\frac{1}{k}}$. There exists a routing scheme for $G$ with stretch $O(k\log\log s(n))$, with labels, tables and headers of sizes $O(k\cdot S(n,k)\log n)$, $O(k\cdot S(n,k))$ and $O(1)$, respectively. In particular,
\begin{enumerate}
    \item if $s=s(\theta)=\theta^\delta$, for some $\delta\in(0,1]$, the stretch is $O(k\log\log n)$, with labels and tables of sizes $O(\frac{k^2}{\delta}n^{\frac{\delta}{k}}\log n)$ and $O(\frac{k^2}{\delta}n^{\frac{\delta}{k}})$, respectively.
    \item if $G$ has treewidth $t=t(n)$, the stretch is $O(k\log\log n)$, with labels and tables of sizes $O(k\cdot t^{\frac{1}{k}}\cdot\log^2n)$ and $O(k\cdot t^{\frac{1}{k}}\cdot\log n)$, respectively.
\end{enumerate}
\end{theorem}

For $t(n)\leq2^{\frac{c\log n}{\log\log n}}$, for some sufficiently small constant $c>0$, the tradeoff between the stretch and the aggregate size of the routing table and label stored at each individual vertex is better in our scheme than in the state-of-the-art schemes \cite{TZ01-spaa,C13}. However, our label size is worse than the label size in \cite{TZ01-spaa,C13}.

For path-separable graphs, we apply Lemma \ref{lemma:RoutingSchemeFromTreeCover} on the spanning tree covers from Theorem \ref{thm:SpanTreeCoverForPathSeparators} and Theorem \ref{thm:LargeStretchTreeCoverForPathSeparators}.

\begin{theorem} \label{thm:PathSeparableRouting}
Let $G=(V,E)$ be a connected $n$-vertex undirected weighted graph, let $k\geq1$ be an integer parameter, and let $0<\epsilon\leq1$ be a real parameter.
\begin{enumerate}
    \item If $G$ is $p$-path-separable, then $G$ has a routing scheme with stretch $1+\epsilon$, labels of size $O(p\cdot\frac{\log^3n}{\epsilon})$, tables of size $O(p\cdot\frac{\log^2n}{\epsilon})$ and headers of size $O(1)$.
    \item If $G$ is $(\ell,\pi)$-path-separable, then $G$ has a routing scheme with stretch $O(k\log\log(\pi+4))$, labels of size 
    $O(k\cdot\ell\cdot\pi^{\frac{1}{k}}\cdot\log^4n)$, tables of size $O(k\cdot\ell\cdot\pi^{\frac{1}{k}}\cdot\log^3n)$ and headers of size $O(1)$.
    \item If $G$ is tree-like $(\ell,\pi)$-path-separable, then the sizes of the labels and tables improve to $O(k\cdot\ell\cdot\pi^{\frac{1}{k}}\cdot\log^3n)$ and $O(k\cdot\ell\cdot\pi^{\frac{1}{k}}\cdot\log^2n)$, respectively.
\end{enumerate}
\end{theorem}

Recall that $K_r$-minor-free graphs are (weakly) tree-like $(\ell,\pi)$-path-separable, for $\ell=O(r^2)$ and $\pi=O(r^{c_{AG}})$, where $c_{AG}=4602$, and that graphs with bounded genus $g$ are strongly tree-like $O(g)$-path-separable (see Items (2) and (3) in Theorem \ref{thm:PathSeparableFamilies}). Thus, the following corollary is derived by Theorem \ref{thm:PathSeparableRouting}.

\begin{corollary}
    Let $G=(V,E)$ be a connected $n$-vertex undirected weighted graph, let $k\geq1$ be an integer parameter, and let $0<\epsilon\leq1$ be a real parameter.
\begin{enumerate}
    \item If $G$ is $K_r$-minor-free, then $G$ has a routing scheme with stretch $1+\epsilon$, labels of size $O(r^{c_{AG}}\cdot\frac{\log^3n}{\epsilon})$, tables of size $O(r^{c_{AG}}\cdot\frac{\log^2n}{\epsilon})$ and headers of size $O(1)$.
    \item If $G$ has a bounded genus $g$, then $G$ has a routing scheme with stretch $1+\epsilon$, labels of size $O(g\cdot\frac{\log^3n}{\epsilon})$, tables of size $O(g\cdot\frac{\log^2n}{\epsilon})$ and headers of size $O(1)$.
    \item If $G$ is $K_r$-minor-free, then $G$ has a routing scheme with stretch $O(k\log\log r)$, labels of size 
    $O(kr^{\frac{1}{k}}\cdot\log^3n)$, tables of size $O(kr^{\frac{1}{k}}\cdot\log^2n)$ and headers of size $O(1)$.
    \item If $G$ has a bounded genus $g$, then $G$ has a routing scheme with stretch $O(k\log\log g)$, labels of size 
    $O(kg^{\frac{1}{k}}\cdot\log^3n)$, tables of size $O(kg^{\frac{1}{k}}\cdot\log^2n)$ and headers of size $O(1)$.
\end{enumerate}

\end{corollary}

\end{document}